\documentclass{article}

\pdfoutput=0

\usepackage{amssymb}
\usepackage{amsfonts}
\usepackage{amsmath}
\usepackage{amsthm}

\newtheorem{theorem}{Theorem}
\newtheorem{lemma}[theorem]{Lemma}

\usepackage{epic,eepic}

\newcommand{\F}{{\mathbb{F}}}
\newcommand{\R}{{\mathbb{R}}}
\newcommand{\Q}{{\mathbb{Q}}}
\newcommand{\Z}{{\mathbb{Z}}}
\newcommand{\forrandom}{\makebox(5.5,7){\begin{picture}(5,7)
\path(4.5,0)(4.5,7)
\path(2,4)(0,7)
\put(4.5,2){\oval(9,4)[l]}
\end{picture}}}

\title{The computational complexity of Minesweeper}
\author{Michiel de Bondt \\ %
	Radboud University Nijmegen, The Netherlands \\ 
	M.deBondt@math.ru.nl}

\begin{document}

\maketitle

\begin{abstract} \noindent
We show that the Minesweeper game is PP-hard, when the object is to
locate all mines with the highest probability. When the probability
of locating all mines may be infinitesimal, the Minesweeper game is even
PSPACE-complete. In our construction, the player can reveal a boolean circuit
in polynomial time, after guessing an initial square with no surrounding
mines, a guess that has 99 percent probability of success. Subsequently,
the mines must be located with a maximum probability of success.

Furthermore, we show that determining the solvability of a partially uncovered
Minesweeper board is NP-complete with hexagonal and triangular grids as well 
as a square grid, extending a similar result for square grids only by R. Kaye. 
Actually finding the mines with a maximum probability of success is again
PP-hard or PSPACE-complete respectively.

Our constructions are 
in such a way that the number of mines can be computed in polynomial time
and hence a possible mine counter does not provide additional information.
The results are obtained by replacing the dyadic gates in \cite{kaye}
by two primitives which makes life more easy in this context.

\end{abstract}

\section{Introduction}

On almost every computer, one can play the game called Minesweeper. Minesweeper
is played on a grid of square compartments, each of them surrounded by eight
other such compartments, except on the border of the grid. The object is
the game is to find all compartments which do not contain a mine. Any compartments
that does not contain a mine contains a number which indicates how many of the 
surrounding compartments contains a mine. These numbers can be used to locate the mines.
But such a number is only revealed after clicking on the compartment, something 
that is fatal when the compartment does contain a mine. 

Richard Kaye showed in \cite{kaye} that determining the solvability of a 
Minesweeper board is NP-complete, that is, whether the numbers revealed thus far
correspond to a distribution of the mines. Although he ignores the fact
that the total number of mines is also known in the Minesweeper game, it is a 
very nice article, where he builds a logical circuit on the Minesweeper
board, in order to compute a boolean expression. Next, he enforces the outcome
of the boolean expression to be true, whence the problem becomes finding 
boolean values for the variables such that the logical expression evaluates 
to true.

One can modify the game such that it can be played on other grids, which
has already been done several times.
We are going to show that we can determine the satisfiability of a boolean
circuit by evaluating the solvability of a minesweeper board,
for all three regular tesselations of the plane (triangular, normal and
hexagonal Minesweeper). For this purpose, we build circuitry as well.

Since Minesweeper is a game with numbers, it is natural to associate
the boolean values {\em false} and {\em true} with $0$ and $1$ respectively.
In order to make boolean circuits on a minesweeper board, we first 
need a way to code wires on it. This is done in figure \ref{wire}, where
the wire on the square grid is taken from \cite{kaye}.

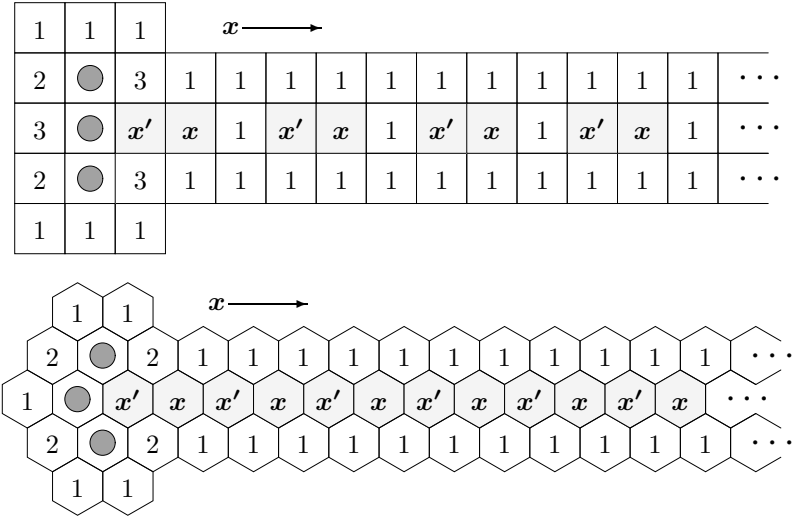
\begin{figure}[!htp]
\begin{center}
\begin{picture}(288.17,95.00)(0.00,0.00)

\mathversion{bold}
\filltype{shade}
\texture{40004000 0 0 0 00400040 0 0 0
         40004000 0 0 0 00400040 0 0 0
         40004000 0 0 0 00400040 0 0 0
         40004000 0 0 0 00400040 0 0 0}
\whiten

\path(0.00,19.00)(19.00,19.00)(19.00,0.00)(0.00,0.00)(0.00,19.00)
\put(9.50,5.50){\makebox(0,0)[b]{1}}

\path(0.00,38.00)(19.00,38.00)(19.00,19.00)(0.00,19.00)(0.00,38.00)
\put(9.50,24.50){\makebox(0,0)[b]{2}}

\path(0.00,57.00)(19.00,57.00)(19.00,38.00)(0.00,38.00)(0.00,57.00)
\put(9.50,43.50){\makebox(0,0)[b]{3}}

\path(0.00,76.00)(19.00,76.00)(19.00,57.00)(0.00,57.00)(0.00,76.00)
\put(9.50,62.50){\makebox(0,0)[b]{2}}

\path(0.00,95.00)(19.00,95.00)(19.00,76.00)(0.00,76.00)(0.00,95.00)
\put(9.50,81.50){\makebox(0,0)[b]{1}}

\path(19.00,19.00)(38.00,19.00)(38.00,0.00)(19.00,0.00)(19.00,19.00)
\put(28.50,5.50){\makebox(0,0)[b]{1}}

\path(19.00,95.00)(38.00,95.00)(38.00,76.00)(19.00,76.00)(19.00,95.00)
\put(28.50,81.50){\makebox(0,0)[b]{1}}

\path(38.00,19.00)(57.00,19.00)(57.00,0.00)(38.00,0.00)(38.00,19.00)
\put(47.50,5.50){\makebox(0,0)[b]{1}}

\path(38.00,38.00)(57.00,38.00)(57.00,19.00)(38.00,19.00)(38.00,38.00)
\put(47.50,24.50){\makebox(0,0)[b]{3}}

\shade
\path(38.00,57.00)(57.00,57.00)(57.00,38.00)(38.00,38.00)(38.00,57.00)
\put(47.50,43.50){\makebox(0,0)[b]{$x'$}}

\path(38.00,76.00)(57.00,76.00)(57.00,57.00)(38.00,57.00)(38.00,76.00)
\put(47.50,62.50){\makebox(0,0)[b]{3}}

\path(38.00,95.00)(57.00,95.00)(57.00,76.00)(38.00,76.00)(38.00,95.00)
\put(47.50,81.50){\makebox(0,0)[b]{1}}

\path(57.00,38.00)(76.00,38.00)(76.00,19.00)(57.00,19.00)(57.00,38.00)
\put(66.50,24.50){\makebox(0,0)[b]{1}}

\shade
\path(57.00,57.00)(76.00,57.00)(76.00,38.00)(57.00,38.00)(57.00,57.00)
\put(66.50,43.50){\makebox(0,0)[b]{$x$}}

\path(57.00,76.00)(76.00,76.00)(76.00,57.00)(57.00,57.00)(57.00,76.00)
\put(66.50,62.50){\makebox(0,0)[b]{1}}

\path(76.00,38.00)(95.00,38.00)(95.00,19.00)(76.00,19.00)(76.00,38.00)
\put(85.50,24.50){\makebox(0,0)[b]{1}}

\path(76.00,57.00)(95.00,57.00)(95.00,38.00)(76.00,38.00)(76.00,57.00)
\put(85.50,43.50){\makebox(0,0)[b]{1}}

\path(76.00,76.00)(95.00,76.00)(95.00,57.00)(76.00,57.00)(76.00,76.00)
\put(85.50,62.50){\makebox(0,0)[b]{1}}

\path(95.00,38.00)(114.00,38.00)(114.00,19.00)(95.00,19.00)(95.00,38.00)
\put(104.50,24.50){\makebox(0,0)[b]{1}}

\shade
\path(95.00,57.00)(114.00,57.00)(114.00,38.00)(95.00,38.00)(95.00,57.00)
\put(104.50,43.50){\makebox(0,0)[b]{$x'$}}

\path(95.00,76.00)(114.00,76.00)(114.00,57.00)(95.00,57.00)(95.00,76.00)
\put(104.50,62.50){\makebox(0,0)[b]{1}}

\path(114.00,38.00)(133.00,38.00)(133.00,19.00)(114.00,19.00)(114.00,38.00)
\put(123.50,24.50){\makebox(0,0)[b]{1}}

\shade
\path(114.00,57.00)(133.00,57.00)(133.00,38.00)(114.00,38.00)(114.00,57.00)
\put(123.50,43.50){\makebox(0,0)[b]{$x$}}

\path(114.00,76.00)(133.00,76.00)(133.00,57.00)(114.00,57.00)(114.00,76.00)
\put(123.50,62.50){\makebox(0,0)[b]{1}}

\path(133.00,38.00)(152.00,38.00)(152.00,19.00)(133.00,19.00)(133.00,38.00)
\put(142.50,24.50){\makebox(0,0)[b]{1}}

\path(133.00,57.00)(152.00,57.00)(152.00,38.00)(133.00,38.00)(133.00,57.00)
\put(142.50,43.50){\makebox(0,0)[b]{1}}

\path(133.00,76.00)(152.00,76.00)(152.00,57.00)(133.00,57.00)(133.00,76.00)
\put(142.50,62.50){\makebox(0,0)[b]{1}}

\path(152.00,38.00)(171.00,38.00)(171.00,19.00)(152.00,19.00)(152.00,38.00)
\put(161.50,24.50){\makebox(0,0)[b]{1}}

\shade
\path(152.00,57.00)(171.00,57.00)(171.00,38.00)(152.00,38.00)(152.00,57.00)
\put(161.50,43.50){\makebox(0,0)[b]{$x'$}}

\path(152.00,76.00)(171.00,76.00)(171.00,57.00)(152.00,57.00)(152.00,76.00)
\put(161.50,62.50){\makebox(0,0)[b]{1}}

\path(171.00,38.00)(190.00,38.00)(190.00,19.00)(171.00,19.00)(171.00,38.00)
\put(180.50,24.50){\makebox(0,0)[b]{1}}

\shade
\path(171.00,57.00)(190.00,57.00)(190.00,38.00)(171.00,38.00)(171.00,57.00)
\put(180.50,43.50){\makebox(0,0)[b]{$x$}}

\path(171.00,76.00)(190.00,76.00)(190.00,57.00)(171.00,57.00)(171.00,76.00)
\put(180.50,62.50){\makebox(0,0)[b]{1}}

\path(190.00,38.00)(209.00,38.00)(209.00,19.00)(190.00,19.00)(190.00,38.00)
\put(199.50,24.50){\makebox(0,0)[b]{1}}

\path(190.00,57.00)(209.00,57.00)(209.00,38.00)(190.00,38.00)(190.00,57.00)
\put(199.50,43.50){\makebox(0,0)[b]{1}}

\path(190.00,76.00)(209.00,76.00)(209.00,57.00)(190.00,57.00)(190.00,76.00)
\put(199.50,62.50){\makebox(0,0)[b]{1}}

\path(209.00,38.00)(228.00,38.00)(228.00,19.00)(209.00,19.00)(209.00,38.00)
\put(218.50,24.50){\makebox(0,0)[b]{1}}

\shade
\path(209.00,57.00)(228.00,57.00)(228.00,38.00)(209.00,38.00)(209.00,57.00)
\put(218.50,43.50){\makebox(0,0)[b]{$x'$}}

\path(209.00,76.00)(228.00,76.00)(228.00,57.00)(209.00,57.00)(209.00,76.00)
\put(218.50,62.50){\makebox(0,0)[b]{1}}

\path(228.00,38.00)(247.00,38.00)(247.00,19.00)(228.00,19.00)(228.00,38.00)
\put(237.50,24.50){\makebox(0,0)[b]{1}}

\shade
\path(228.00,57.00)(247.00,57.00)(247.00,38.00)(228.00,38.00)(228.00,57.00)
\put(237.50,43.50){\makebox(0,0)[b]{$x$}}

\path(228.00,76.00)(247.00,76.00)(247.00,57.00)(228.00,57.00)(228.00,76.00)
\put(237.50,62.50){\makebox(0,0)[b]{1}}

\path(247.00,38.00)(266.00,38.00)(266.00,19.00)(247.00,19.00)(247.00,38.00)
\put(256.50,24.50){\makebox(0,0)[b]{1}}

\path(247.00,57.00)(266.00,57.00)(266.00,38.00)(247.00,38.00)(247.00,57.00)
\put(256.50,43.50){\makebox(0,0)[b]{1}}

\path(247.00,76.00)(266.00,76.00)(266.00,57.00)(247.00,57.00)(247.00,76.00)
\put(256.50,62.50){\makebox(0,0)[b]{1}}

\path(285.00,19.00)(266.00,19.00)(266.00,38.00)(285.00,38.00)
\put(275.50,28.50){\makebox(0,0){$\cdot$}}
\put(281.83,28.50){\makebox(0,0){$\cdot$}}
\put(288.17,28.50){\makebox(0,0){$\cdot$}}

\path(285.00,38.00)(266.00,38.00)(266.00,57.00)(285.00,57.00)
\put(275.50,47.50){\makebox(0,0){$\cdot$}}
\put(281.83,47.50){\makebox(0,0){$\cdot$}}
\put(288.17,47.50){\makebox(0,0){$\cdot$}}

\path(285.00,57.00)(266.00,57.00)(266.00,76.00)(285.00,76.00)
\put(275.50,66.50){\makebox(0,0){$\cdot$}}
\put(281.83,66.50){\makebox(0,0){$\cdot$}}
\put(288.17,66.50){\makebox(0,0){$\cdot$}}

\texture{cccccccc 0 0 0 cccccccc 0 0 0
         cccccccc 0 0 0 cccccccc 0 0 0
         cccccccc 0 0 0 cccccccc 0 0 0
         cccccccc 0 0 0 cccccccc 0 0 0}
\whiten

\path(19.00,38.00)(38.00,38.00)(38.00,19.00)(19.00,19.00)(19.00,38.00)
\put(28.50,28.50){\makebox(0,0){\circle*{9.50}}}

\path(19.00,57.00)(38.00,57.00)(38.00,38.00)(19.00,38.00)(19.00,57.00)
\put(28.50,47.50){\makebox(0,0){\circle*{9.50}}}

\path(19.00,76.00)(38.00,76.00)(38.00,57.00)(19.00,57.00)(19.00,76.00)
\put(28.50,66.50){\makebox(0,0){\circle*{9.50}}}

\put(85,85.5){\makebox(0,0)[r]{$x$}}
\put(86,85.5){\vector(1,0){30}}
\mathversion{normal}

\end{picture}\\[10pt]
\begin{picture}(297.67,88.00)(0.00,0.00)

\mathversion{bold}
\filltype{shade}
\texture{40004000 0 0 0 00400040 0 0 0
         40004000 0 0 0 00400040 0 0 0
         40004000 0 0 0 00400040 0 0 0
         40004000 0 0 0 00400040 0 0 0}
\whiten

\path(0.00,49.50)(0.00,38.50)(9.50,33.00)(19.00,38.50)
     (19.00,49.50)(9.50,55.00)(0.00,49.50)
\put(9.50,40.00){\makebox(0,0)[b]{1}}

\path(9.50,33.00)(9.50,22.00)(19.00,16.50)(28.50,22.00)
     (28.50,33.00)(19.00,38.50)(9.50,33.00)
\put(19.00,23.50){\makebox(0,0)[b]{2}}

\path(9.50,66.00)(9.50,55.00)(19.00,49.50)(28.50,55.00)
     (28.50,66.00)(19.00,71.50)(9.50,66.00)
\put(19.00,56.50){\makebox(0,0)[b]{2}}

\path(19.00,16.50)(19.00,5.50)(28.50,0.00)(38.00,5.50)
     (38.00,16.50)(28.50,22.00)(19.00,16.50)
\put(28.50,7.00){\makebox(0,0)[b]{1}}

\path(19.00,82.50)(19.00,71.50)(28.50,66.00)(38.00,71.50)
     (38.00,82.50)(28.50,88.00)(19.00,82.50)
\put(28.50,73.00){\makebox(0,0)[b]{1}}

\path(38.00,16.50)(38.00,5.50)(47.50,0.00)(57.00,5.50)
     (57.00,16.50)(47.50,22.00)(38.00,16.50)
\put(47.50,7.00){\makebox(0,0)[b]{1}}

\shade
\path(38.00,49.50)(38.00,38.50)(47.50,33.00)(57.00,38.50)
     (57.00,49.50)(47.50,55.00)(38.00,49.50)
\put(47.50,40.00){\makebox(0,0)[b]{$x'$}}

\path(38.00,82.50)(38.00,71.50)(47.50,66.00)(57.00,71.50)
     (57.00,82.50)(47.50,88.00)(38.00,82.50)
\put(47.50,73.00){\makebox(0,0)[b]{1}}

\path(47.50,33.00)(47.50,22.00)(57.00,16.50)(66.50,22.00)
     (66.50,33.00)(57.00,38.50)(47.50,33.00)
\put(57.00,23.50){\makebox(0,0)[b]{2}}

\path(47.50,66.00)(47.50,55.00)(57.00,49.50)(66.50,55.00)
     (66.50,66.00)(57.00,71.50)(47.50,66.00)
\put(57.00,56.50){\makebox(0,0)[b]{2}}

\shade
\path(57.00,49.50)(57.00,38.50)(66.50,33.00)(76.00,38.50)
     (76.00,49.50)(66.50,55.00)(57.00,49.50)
\put(66.50,40.00){\makebox(0,0)[b]{$x$}}

\path(66.50,33.00)(66.50,22.00)(76.00,16.50)(85.50,22.00)
     (85.50,33.00)(76.00,38.50)(66.50,33.00)
\put(76.00,23.50){\makebox(0,0)[b]{1}}

\path(66.50,66.00)(66.50,55.00)(76.00,49.50)(85.50,55.00)
     (85.50,66.00)(76.00,71.50)(66.50,66.00)
\put(76.00,56.50){\makebox(0,0)[b]{1}}

\shade
\path(76.00,49.50)(76.00,38.50)(85.50,33.00)(95.00,38.50)
     (95.00,49.50)(85.50,55.00)(76.00,49.50)
\put(85.50,40.00){\makebox(0,0)[b]{$x'$}}

\path(85.50,33.00)(85.50,22.00)(95.00,16.50)(104.50,22.00)
     (104.50,33.00)(95.00,38.50)(85.50,33.00)
\put(95.00,23.50){\makebox(0,0)[b]{1}}

\path(85.50,66.00)(85.50,55.00)(95.00,49.50)(104.50,55.00)
     (104.50,66.00)(95.00,71.50)(85.50,66.00)
\put(95.00,56.50){\makebox(0,0)[b]{1}}

\shade
\path(95.00,49.50)(95.00,38.50)(104.50,33.00)(114.00,38.50)
     (114.00,49.50)(104.50,55.00)(95.00,49.50)
\put(104.50,40.00){\makebox(0,0)[b]{$x$}}

\path(104.50,33.00)(104.50,22.00)(114.00,16.50)(123.50,22.00)
     (123.50,33.00)(114.00,38.50)(104.50,33.00)
\put(114.00,23.50){\makebox(0,0)[b]{1}}

\path(104.50,66.00)(104.50,55.00)(114.00,49.50)(123.50,55.00)
     (123.50,66.00)(114.00,71.50)(104.50,66.00)
\put(114.00,56.50){\makebox(0,0)[b]{1}}

\shade
\path(114.00,49.50)(114.00,38.50)(123.50,33.00)(133.00,38.50)
     (133.00,49.50)(123.50,55.00)(114.00,49.50)
\put(123.50,40.00){\makebox(0,0)[b]{$x'$}}

\path(123.50,33.00)(123.50,22.00)(133.00,16.50)(142.50,22.00)
     (142.50,33.00)(133.00,38.50)(123.50,33.00)
\put(133.00,23.50){\makebox(0,0)[b]{1}}

\path(123.50,66.00)(123.50,55.00)(133.00,49.50)(142.50,55.00)
     (142.50,66.00)(133.00,71.50)(123.50,66.00)
\put(133.00,56.50){\makebox(0,0)[b]{1}}

\shade
\path(133.00,49.50)(133.00,38.50)(142.50,33.00)(152.00,38.50)
     (152.00,49.50)(142.50,55.00)(133.00,49.50)
\put(142.50,40.00){\makebox(0,0)[b]{$x$}}

\path(142.50,33.00)(142.50,22.00)(152.00,16.50)(161.50,22.00)
     (161.50,33.00)(152.00,38.50)(142.50,33.00)
\put(152.00,23.50){\makebox(0,0)[b]{1}}

\path(142.50,66.00)(142.50,55.00)(152.00,49.50)(161.50,55.00)
     (161.50,66.00)(152.00,71.50)(142.50,66.00)
\put(152.00,56.50){\makebox(0,0)[b]{1}}

\shade
\path(152.00,49.50)(152.00,38.50)(161.50,33.00)(171.00,38.50)
     (171.00,49.50)(161.50,55.00)(152.00,49.50)
\put(161.50,40.00){\makebox(0,0)[b]{$x'$}}

\path(161.50,33.00)(161.50,22.00)(171.00,16.50)(180.50,22.00)
     (180.50,33.00)(171.00,38.50)(161.50,33.00)
\put(171.00,23.50){\makebox(0,0)[b]{1}}

\path(161.50,66.00)(161.50,55.00)(171.00,49.50)(180.50,55.00)
     (180.50,66.00)(171.00,71.50)(161.50,66.00)
\put(171.00,56.50){\makebox(0,0)[b]{1}}

\shade
\path(171.00,49.50)(171.00,38.50)(180.50,33.00)(190.00,38.50)
     (190.00,49.50)(180.50,55.00)(171.00,49.50)
\put(180.50,40.00){\makebox(0,0)[b]{$x$}}

\path(180.50,33.00)(180.50,22.00)(190.00,16.50)(199.50,22.00)
     (199.50,33.00)(190.00,38.50)(180.50,33.00)
\put(190.00,23.50){\makebox(0,0)[b]{1}}

\path(180.50,66.00)(180.50,55.00)(190.00,49.50)(199.50,55.00)
     (199.50,66.00)(190.00,71.50)(180.50,66.00)
\put(190.00,56.50){\makebox(0,0)[b]{1}}

\shade
\path(190.00,49.50)(190.00,38.50)(199.50,33.00)(209.00,38.50)
     (209.00,49.50)(199.50,55.00)(190.00,49.50)
\put(199.50,40.00){\makebox(0,0)[b]{$x'$}}

\path(199.50,33.00)(199.50,22.00)(209.00,16.50)(218.50,22.00)
     (218.50,33.00)(209.00,38.50)(199.50,33.00)
\put(209.00,23.50){\makebox(0,0)[b]{1}}

\path(199.50,66.00)(199.50,55.00)(209.00,49.50)(218.50,55.00)
     (218.50,66.00)(209.00,71.50)(199.50,66.00)
\put(209.00,56.50){\makebox(0,0)[b]{1}}

\shade
\path(209.00,49.50)(209.00,38.50)(218.50,33.00)(228.00,38.50)
     (228.00,49.50)(218.50,55.00)(209.00,49.50)
\put(218.50,40.00){\makebox(0,0)[b]{$x$}}

\path(218.50,33.00)(218.50,22.00)(228.00,16.50)(237.50,22.00)
     (237.50,33.00)(228.00,38.50)(218.50,33.00)
\put(228.00,23.50){\makebox(0,0)[b]{1}}

\path(218.50,66.00)(218.50,55.00)(228.00,49.50)(237.50,55.00)
     (237.50,66.00)(228.00,71.50)(218.50,66.00)
\put(228.00,56.50){\makebox(0,0)[b]{1}}

\shade
\path(228.00,49.50)(228.00,38.50)(237.50,33.00)(247.00,38.50)
     (247.00,49.50)(237.50,55.00)(228.00,49.50)
\put(237.50,40.00){\makebox(0,0)[b]{$x'$}}

\path(237.50,33.00)(237.50,22.00)(247.00,16.50)(256.50,22.00)
     (256.50,33.00)(247.00,38.50)(237.50,33.00)
\put(247.00,23.50){\makebox(0,0)[b]{1}}

\path(237.50,66.00)(237.50,55.00)(247.00,49.50)(256.50,55.00)
     (256.50,66.00)(247.00,71.50)(237.50,66.00)
\put(247.00,56.50){\makebox(0,0)[b]{1}}

\shade
\path(247.00,49.50)(247.00,38.50)(256.50,33.00)(266.00,38.50)
     (266.00,49.50)(256.50,55.00)(247.00,49.50)
\put(256.50,40.00){\makebox(0,0)[b]{$x$}}

\path(256.50,33.00)(256.50,22.00)(266.00,16.50)(275.50,22.00)
     (275.50,33.00)(266.00,38.50)(256.50,33.00)
\put(266.00,23.50){\makebox(0,0)[b]{1}}

\path(256.50,66.00)(256.50,55.00)(266.00,49.50)(275.50,55.00)
     (275.50,66.00)(266.00,71.50)(256.50,66.00)
\put(266.00,56.50){\makebox(0,0)[b]{1}}

\path(285.00,49.50)(275.50,55.00)(266.00,49.50)
     (266.00,38.50)(275.50,33.00)(285.00,38.50)
\put(275.50,44.00){\makebox(0,0){$\cdot$}}
\put(281.83,44.00){\makebox(0,0){$\cdot$}}
\put(288.17,44.00){\makebox(0,0){$\cdot$}}

\path(294.50,33.00)(285.00,38.50)(275.50,33.00)
     (275.50,22.00)(285.00,16.50)(294.50,22.00)
\put(285.00,27.50){\makebox(0,0){$\cdot$}}
\put(291.33,27.50){\makebox(0,0){$\cdot$}}
\put(297.67,27.50){\makebox(0,0){$\cdot$}}

\path(294.50,66.00)(285.00,71.50)(275.50,66.00)
     (275.50,55.00)(285.00,49.50)(294.50,55.00)
\put(285.00,60.50){\makebox(0,0){$\cdot$}}
\put(291.33,60.50){\makebox(0,0){$\cdot$}}
\put(297.67,60.50){\makebox(0,0){$\cdot$}}

\texture{cccccccc 0 0 0 cccccccc 0 0 0
         cccccccc 0 0 0 cccccccc 0 0 0
         cccccccc 0 0 0 cccccccc 0 0 0
         cccccccc 0 0 0 cccccccc 0 0 0}
\whiten

\path(19.00,49.50)(19.00,38.50)(28.50,33.00)(38.00,38.50)
     (38.00,49.50)(28.50,55.00)(19.00,49.50)
\put(28.50,44.00){\makebox(0,0){\circle*{9.50}}}

\path(28.50,33.00)(28.50,22.00)(38.00,16.50)(47.50,22.00)
     (47.50,33.00)(38.00,38.50)(28.50,33.00)
\put(38.00,27.50){\makebox(0,0){\circle*{9.50}}}

\path(28.50,66.00)(28.50,55.00)(38.00,49.50)(47.50,55.00)
     (47.50,66.00)(38.00,71.50)(28.50,66.00)
\put(38.00,60.50){\makebox(0,0){\circle*{9.50}}}

\put(84.5,80){\makebox(0,0)[r]{$x$}}
\put(85.5,80){\vector(1,0){30}}
\mathversion{normal}

\end{picture}
\end{center} \vspace{-20pt}
\caption{A wire that conducts $x$, running from the starting point of $x$} \label{wire}
\end{figure}

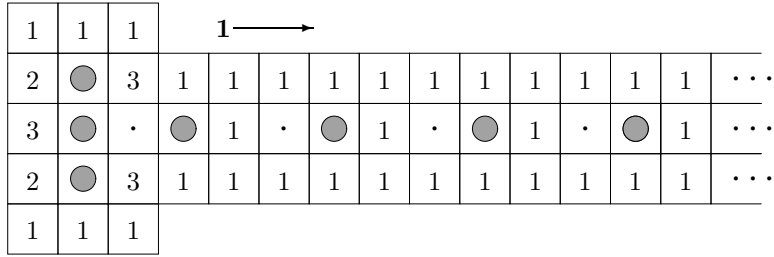
\begin{figure}[!htp]
\begin{center}
\begin{picture}(288.17,95.00)(0.00,0.00)

\mathversion{bold}
\filltype{shade}
\texture{40004000 0 0 0 00400040 0 0 0
         40004000 0 0 0 00400040 0 0 0
         40004000 0 0 0 00400040 0 0 0
         40004000 0 0 0 00400040 0 0 0}
\whiten

\path(0.00,19.00)(19.00,19.00)(19.00,0.00)(0.00,0.00)(0.00,19.00)
\put(9.50,5.50){\makebox(0,0)[b]{1}}

\path(0.00,38.00)(19.00,38.00)(19.00,19.00)(0.00,19.00)(0.00,38.00)
\put(9.50,24.50){\makebox(0,0)[b]{2}}

\path(0.00,57.00)(19.00,57.00)(19.00,38.00)(0.00,38.00)(0.00,57.00)
\put(9.50,43.50){\makebox(0,0)[b]{3}}

\path(0.00,76.00)(19.00,76.00)(19.00,57.00)(0.00,57.00)(0.00,76.00)
\put(9.50,62.50){\makebox(0,0)[b]{2}}

\path(0.00,95.00)(19.00,95.00)(19.00,76.00)(0.00,76.00)(0.00,95.00)
\put(9.50,81.50){\makebox(0,0)[b]{1}}

\path(19.00,19.00)(38.00,19.00)(38.00,0.00)(19.00,0.00)(19.00,19.00)
\put(28.50,5.50){\makebox(0,0)[b]{1}}

\path(19.00,95.00)(38.00,95.00)(38.00,76.00)(19.00,76.00)(19.00,95.00)
\put(28.50,81.50){\makebox(0,0)[b]{1}}

\path(38.00,19.00)(57.00,19.00)(57.00,0.00)(38.00,0.00)(38.00,19.00)
\put(47.50,5.50){\makebox(0,0)[b]{1}}

\path(38.00,38.00)(57.00,38.00)(57.00,19.00)(38.00,19.00)(38.00,38.00)
\put(47.50,24.50){\makebox(0,0)[b]{3}}

\path(38.00,57.00)(57.00,57.00)(57.00,38.00)(38.00,38.00)(38.00,57.00)
\put(47.50,47.50){\makebox(0,0){$\cdot$}}

\path(38.00,76.00)(57.00,76.00)(57.00,57.00)(38.00,57.00)(38.00,76.00)
\put(47.50,62.50){\makebox(0,0)[b]{3}}

\path(38.00,95.00)(57.00,95.00)(57.00,76.00)(38.00,76.00)(38.00,95.00)
\put(47.50,81.50){\makebox(0,0)[b]{1}}

\path(57.00,38.00)(76.00,38.00)(76.00,19.00)(57.00,19.00)(57.00,38.00)
\put(66.50,24.50){\makebox(0,0)[b]{1}}

\path(57.00,76.00)(76.00,76.00)(76.00,57.00)(57.00,57.00)(57.00,76.00)
\put(66.50,62.50){\makebox(0,0)[b]{1}}

\path(76.00,38.00)(95.00,38.00)(95.00,19.00)(76.00,19.00)(76.00,38.00)
\put(85.50,24.50){\makebox(0,0)[b]{1}}

\path(76.00,57.00)(95.00,57.00)(95.00,38.00)(76.00,38.00)(76.00,57.00)
\put(85.50,43.50){\makebox(0,0)[b]{1}}

\path(76.00,76.00)(95.00,76.00)(95.00,57.00)(76.00,57.00)(76.00,76.00)
\put(85.50,62.50){\makebox(0,0)[b]{1}}

\path(95.00,38.00)(114.00,38.00)(114.00,19.00)(95.00,19.00)(95.00,38.00)
\put(104.50,24.50){\makebox(0,0)[b]{1}}

\path(95.00,57.00)(114.00,57.00)(114.00,38.00)(95.00,38.00)(95.00,57.00)
\put(104.50,47.50){\makebox(0,0){$\cdot$}}

\path(95.00,76.00)(114.00,76.00)(114.00,57.00)(95.00,57.00)(95.00,76.00)
\put(104.50,62.50){\makebox(0,0)[b]{1}}

\path(114.00,38.00)(133.00,38.00)(133.00,19.00)(114.00,19.00)(114.00,38.00)
\put(123.50,24.50){\makebox(0,0)[b]{1}}

\path(114.00,76.00)(133.00,76.00)(133.00,57.00)(114.00,57.00)(114.00,76.00)
\put(123.50,62.50){\makebox(0,0)[b]{1}}

\path(133.00,38.00)(152.00,38.00)(152.00,19.00)(133.00,19.00)(133.00,38.00)
\put(142.50,24.50){\makebox(0,0)[b]{1}}

\path(133.00,57.00)(152.00,57.00)(152.00,38.00)(133.00,38.00)(133.00,57.00)
\put(142.50,43.50){\makebox(0,0)[b]{1}}

\path(133.00,76.00)(152.00,76.00)(152.00,57.00)(133.00,57.00)(133.00,76.00)
\put(142.50,62.50){\makebox(0,0)[b]{1}}

\path(152.00,38.00)(171.00,38.00)(171.00,19.00)(152.00,19.00)(152.00,38.00)
\put(161.50,24.50){\makebox(0,0)[b]{1}}

\path(152.00,57.00)(171.00,57.00)(171.00,38.00)(152.00,38.00)(152.00,57.00)
\put(161.50,47.50){\makebox(0,0){$\cdot$}}

\path(152.00,76.00)(171.00,76.00)(171.00,57.00)(152.00,57.00)(152.00,76.00)
\put(161.50,62.50){\makebox(0,0)[b]{1}}

\path(171.00,38.00)(190.00,38.00)(190.00,19.00)(171.00,19.00)(171.00,38.00)
\put(180.50,24.50){\makebox(0,0)[b]{1}}

\path(171.00,76.00)(190.00,76.00)(190.00,57.00)(171.00,57.00)(171.00,76.00)
\put(180.50,62.50){\makebox(0,0)[b]{1}}

\path(190.00,38.00)(209.00,38.00)(209.00,19.00)(190.00,19.00)(190.00,38.00)
\put(199.50,24.50){\makebox(0,0)[b]{1}}

\path(190.00,57.00)(209.00,57.00)(209.00,38.00)(190.00,38.00)(190.00,57.00)
\put(199.50,43.50){\makebox(0,0)[b]{1}}

\path(190.00,76.00)(209.00,76.00)(209.00,57.00)(190.00,57.00)(190.00,76.00)
\put(199.50,62.50){\makebox(0,0)[b]{1}}

\path(209.00,38.00)(228.00,38.00)(228.00,19.00)(209.00,19.00)(209.00,38.00)
\put(218.50,24.50){\makebox(0,0)[b]{1}}

\path(209.00,57.00)(228.00,57.00)(228.00,38.00)(209.00,38.00)(209.00,57.00)
\put(218.50,47.50){\makebox(0,0){$\cdot$}}

\path(209.00,76.00)(228.00,76.00)(228.00,57.00)(209.00,57.00)(209.00,76.00)
\put(218.50,62.50){\makebox(0,0)[b]{1}}

\path(228.00,38.00)(247.00,38.00)(247.00,19.00)(228.00,19.00)(228.00,38.00)
\put(237.50,24.50){\makebox(0,0)[b]{1}}

\texture{cccccccc 0 0 0 cccccccc 0 0 0
         cccccccc 0 0 0 cccccccc 0 0 0
         cccccccc 0 0 0 cccccccc 0 0 0
         cccccccc 0 0 0 cccccccc 0 0 0}
\whiten

\path(228.00,57.00)(247.00,57.00)(247.00,38.00)(228.00,38.00)(228.00,57.00)
\put(237.50,47.50){\makebox(0,0){\circle*{9.50}}}

\path(228.00,76.00)(247.00,76.00)(247.00,57.00)(228.00,57.00)(228.00,76.00)
\put(237.50,62.50){\makebox(0,0)[b]{1}}

\path(247.00,38.00)(266.00,38.00)(266.00,19.00)(247.00,19.00)(247.00,38.00)
\put(256.50,24.50){\makebox(0,0)[b]{1}}

\path(247.00,57.00)(266.00,57.00)(266.00,38.00)(247.00,38.00)(247.00,57.00)
\put(256.50,43.50){\makebox(0,0)[b]{1}}

\path(247.00,76.00)(266.00,76.00)(266.00,57.00)(247.00,57.00)(247.00,76.00)
\put(256.50,62.50){\makebox(0,0)[b]{1}}

\path(285.00,19.00)(266.00,19.00)(266.00,38.00)(285.00,38.00)
\put(275.50,28.50){\makebox(0,0){$\cdot$}}
\put(281.83,28.50){\makebox(0,0){$\cdot$}}
\put(288.17,28.50){\makebox(0,0){$\cdot$}}

\path(285.00,38.00)(266.00,38.00)(266.00,57.00)(285.00,57.00)
\put(275.50,47.50){\makebox(0,0){$\cdot$}}
\put(281.83,47.50){\makebox(0,0){$\cdot$}}
\put(288.17,47.50){\makebox(0,0){$\cdot$}}

\path(285.00,57.00)(266.00,57.00)(266.00,76.00)(285.00,76.00)
\put(275.50,66.50){\makebox(0,0){$\cdot$}}
\put(281.83,66.50){\makebox(0,0){$\cdot$}}
\put(288.17,66.50){\makebox(0,0){$\cdot$}}

\path(19.00,38.00)(38.00,38.00)(38.00,19.00)(19.00,19.00)(19.00,38.00)
\put(28.50,28.50){\makebox(0,0){\circle*{9.50}}}

\path(19.00,57.00)(38.00,57.00)(38.00,38.00)(19.00,38.00)(19.00,57.00)
\put(28.50,47.50){\makebox(0,0){\circle*{9.50}}}

\path(19.00,76.00)(38.00,76.00)(38.00,57.00)(19.00,57.00)(19.00,76.00)
\put(28.50,66.50){\makebox(0,0){\circle*{9.50}}}

\path(57.00,57.00)(76.00,57.00)(76.00,38.00)(57.00,38.00)(57.00,57.00)
\put(66.50,47.50){\makebox(0,0){\circle*{9.50}}}

\path(114.00,57.00)(133.00,57.00)(133.00,38.00)(114.00,38.00)(114.00,57.00)
\put(123.50,47.50){\makebox(0,0){\circle*{9.50}}}

\path(171.00,57.00)(190.00,57.00)(190.00,38.00)(171.00,38.00)(171.00,57.00)
\put(180.50,47.50){\makebox(0,0){\circle*{9.50}}}

\path(228.00,57.00)(247.00,57.00)(247.00,38.00)(228.00,38.00)(228.00,57.00)
\put(237.50,47.50){\makebox(0,0){\circle*{9.50}}}

\put(84.5,85.5){\makebox(0,0)[r]{$1$}}
\put(85.5,85.5){\vector(1,0){30}}
\mathversion{normal}

\end{picture}
\end{center} \vspace{-20pt}
\caption{A wire that conducts one} \label{one}
\end{figure}

\begin{figure}
\begin{center}
\begin{picture}(288.17,95.00)(0.00,0.00)

\mathversion{bold}
\filltype{shade}
\texture{40004000 0 0 0 00400040 0 0 0
         40004000 0 0 0 00400040 0 0 0
         40004000 0 0 0 00400040 0 0 0
         40004000 0 0 0 00400040 0 0 0}
\whiten

\path(0.00,19.00)(19.00,19.00)(19.00,0.00)(0.00,0.00)(0.00,19.00)
\put(9.50,5.50){\makebox(0,0)[b]{1}}

\path(0.00,38.00)(19.00,38.00)(19.00,19.00)(0.00,19.00)(0.00,38.00)
\put(9.50,24.50){\makebox(0,0)[b]{2}}

\path(0.00,57.00)(19.00,57.00)(19.00,38.00)(0.00,38.00)(0.00,57.00)
\put(9.50,43.50){\makebox(0,0)[b]{3}}

\path(0.00,76.00)(19.00,76.00)(19.00,57.00)(0.00,57.00)(0.00,76.00)
\put(9.50,62.50){\makebox(0,0)[b]{2}}

\path(0.00,95.00)(19.00,95.00)(19.00,76.00)(0.00,76.00)(0.00,95.00)
\put(9.50,81.50){\makebox(0,0)[b]{1}}

\path(19.00,19.00)(38.00,19.00)(38.00,0.00)(19.00,0.00)(19.00,19.00)
\put(28.50,5.50){\makebox(0,0)[b]{1}}

\path(19.00,95.00)(38.00,95.00)(38.00,76.00)(19.00,76.00)(19.00,95.00)
\put(28.50,81.50){\makebox(0,0)[b]{1}}

\path(38.00,19.00)(57.00,19.00)(57.00,0.00)(38.00,0.00)(38.00,19.00)
\put(47.50,5.50){\makebox(0,0)[b]{1}}

\path(38.00,38.00)(57.00,38.00)(57.00,19.00)(38.00,19.00)(38.00,38.00)
\put(47.50,24.50){\makebox(0,0)[b]{3}}

\path(38.00,76.00)(57.00,76.00)(57.00,57.00)(38.00,57.00)(38.00,76.00)
\put(47.50,62.50){\makebox(0,0)[b]{3}}

\path(38.00,95.00)(57.00,95.00)(57.00,76.00)(38.00,76.00)(38.00,95.00)
\put(47.50,81.50){\makebox(0,0)[b]{1}}

\path(57.00,38.00)(76.00,38.00)(76.00,19.00)(57.00,19.00)(57.00,38.00)
\put(66.50,24.50){\makebox(0,0)[b]{1}}

\path(57.00,57.00)(76.00,57.00)(76.00,38.00)(57.00,38.00)(57.00,57.00)
\put(66.50,47.50){\makebox(0,0){$\cdot$}}

\path(57.00,76.00)(76.00,76.00)(76.00,57.00)(57.00,57.00)(57.00,76.00)
\put(66.50,62.50){\makebox(0,0)[b]{1}}

\path(76.00,38.00)(95.00,38.00)(95.00,19.00)(76.00,19.00)(76.00,38.00)
\put(85.50,24.50){\makebox(0,0)[b]{1}}

\path(76.00,57.00)(95.00,57.00)(95.00,38.00)(76.00,38.00)(76.00,57.00)
\put(85.50,43.50){\makebox(0,0)[b]{1}}

\path(76.00,76.00)(95.00,76.00)(95.00,57.00)(76.00,57.00)(76.00,76.00)
\put(85.50,62.50){\makebox(0,0)[b]{1}}

\path(95.00,38.00)(114.00,38.00)(114.00,19.00)(95.00,19.00)(95.00,38.00)
\put(104.50,24.50){\makebox(0,0)[b]{1}}

\path(95.00,76.00)(114.00,76.00)(114.00,57.00)(95.00,57.00)(95.00,76.00)
\put(104.50,62.50){\makebox(0,0)[b]{1}}

\path(114.00,38.00)(133.00,38.00)(133.00,19.00)(114.00,19.00)(114.00,38.00)
\put(123.50,24.50){\makebox(0,0)[b]{1}}

\path(114.00,57.00)(133.00,57.00)(133.00,38.00)(114.00,38.00)(114.00,57.00)
\put(123.50,47.50){\makebox(0,0){$\cdot$}}

\path(114.00,76.00)(133.00,76.00)(133.00,57.00)(114.00,57.00)(114.00,76.00)
\put(123.50,62.50){\makebox(0,0)[b]{1}}

\path(133.00,38.00)(152.00,38.00)(152.00,19.00)(133.00,19.00)(133.00,38.00)
\put(142.50,24.50){\makebox(0,0)[b]{1}}

\path(133.00,57.00)(152.00,57.00)(152.00,38.00)(133.00,38.00)(133.00,57.00)
\put(142.50,43.50){\makebox(0,0)[b]{1}}

\path(133.00,76.00)(152.00,76.00)(152.00,57.00)(133.00,57.00)(133.00,76.00)
\put(142.50,62.50){\makebox(0,0)[b]{1}}

\path(152.00,38.00)(171.00,38.00)(171.00,19.00)(152.00,19.00)(152.00,38.00)
\put(161.50,24.50){\makebox(0,0)[b]{1}}

\path(152.00,76.00)(171.00,76.00)(171.00,57.00)(152.00,57.00)(152.00,76.00)
\put(161.50,62.50){\makebox(0,0)[b]{1}}

\path(171.00,38.00)(190.00,38.00)(190.00,19.00)(171.00,19.00)(171.00,38.00)
\put(180.50,24.50){\makebox(0,0)[b]{1}}

\path(171.00,57.00)(190.00,57.00)(190.00,38.00)(171.00,38.00)(171.00,57.00)
\put(180.50,47.50){\makebox(0,0){$\cdot$}}

\path(171.00,76.00)(190.00,76.00)(190.00,57.00)(171.00,57.00)(171.00,76.00)
\put(180.50,62.50){\makebox(0,0)[b]{1}}

\path(190.00,38.00)(209.00,38.00)(209.00,19.00)(190.00,19.00)(190.00,38.00)
\put(199.50,24.50){\makebox(0,0)[b]{1}}

\path(190.00,57.00)(209.00,57.00)(209.00,38.00)(190.00,38.00)(190.00,57.00)
\put(199.50,43.50){\makebox(0,0)[b]{1}}

\path(190.00,76.00)(209.00,76.00)(209.00,57.00)(190.00,57.00)(190.00,76.00)
\put(199.50,62.50){\makebox(0,0)[b]{1}}

\path(209.00,38.00)(228.00,38.00)(228.00,19.00)(209.00,19.00)(209.00,38.00)
\put(218.50,24.50){\makebox(0,0)[b]{1}}

\path(209.00,76.00)(228.00,76.00)(228.00,57.00)(209.00,57.00)(209.00,76.00)
\put(218.50,62.50){\makebox(0,0)[b]{1}}

\path(228.00,38.00)(247.00,38.00)(247.00,19.00)(228.00,19.00)(228.00,38.00)
\put(237.50,24.50){\makebox(0,0)[b]{1}}

\path(228.00,57.00)(247.00,57.00)(247.00,38.00)(228.00,38.00)(228.00,57.00)
\put(237.50,47.50){\makebox(0,0){$\cdot$}}

\path(228.00,76.00)(247.00,76.00)(247.00,57.00)(228.00,57.00)(228.00,76.00)
\put(237.50,62.50){\makebox(0,0)[b]{1}}

\path(247.00,38.00)(266.00,38.00)(266.00,19.00)(247.00,19.00)(247.00,38.00)
\put(256.50,24.50){\makebox(0,0)[b]{1}}

\path(247.00,57.00)(266.00,57.00)(266.00,38.00)(247.00,38.00)(247.00,57.00)
\put(256.50,43.50){\makebox(0,0)[b]{1}}

\path(247.00,76.00)(266.00,76.00)(266.00,57.00)(247.00,57.00)(247.00,76.00)
\put(256.50,62.50){\makebox(0,0)[b]{1}}

\path(285.00,19.00)(266.00,19.00)(266.00,38.00)(285.00,38.00)
\put(275.50,28.50){\makebox(0,0){$\cdot$}}
\put(281.83,28.50){\makebox(0,0){$\cdot$}}
\put(288.17,28.50){\makebox(0,0){$\cdot$}}

\path(285.00,38.00)(266.00,38.00)(266.00,57.00)(285.00,57.00)
\put(275.50,47.50){\makebox(0,0){$\cdot$}}
\put(281.83,47.50){\makebox(0,0){$\cdot$}}
\put(288.17,47.50){\makebox(0,0){$\cdot$}}

\path(285.00,57.00)(266.00,57.00)(266.00,76.00)(285.00,76.00)
\put(275.50,66.50){\makebox(0,0){$\cdot$}}
\put(281.83,66.50){\makebox(0,0){$\cdot$}}
\put(288.17,66.50){\makebox(0,0){$\cdot$}}

\texture{cccccccc 0 0 0 cccccccc 0 0 0
         cccccccc 0 0 0 cccccccc 0 0 0
         cccccccc 0 0 0 cccccccc 0 0 0
         cccccccc 0 0 0 cccccccc 0 0 0}
\whiten

\path(19.00,38.00)(38.00,38.00)(38.00,19.00)(19.00,19.00)(19.00,38.00)
\put(28.50,28.50){\makebox(0,0){\circle*{9.50}}}

\path(19.00,57.00)(38.00,57.00)(38.00,38.00)(19.00,38.00)(19.00,57.00)
\put(28.50,47.50){\makebox(0,0){\circle*{9.50}}}

\path(19.00,76.00)(38.00,76.00)(38.00,57.00)(19.00,57.00)(19.00,76.00)
\put(28.50,66.50){\makebox(0,0){\circle*{9.50}}}

\path(38.00,57.00)(57.00,57.00)(57.00,38.00)(38.00,38.00)(38.00,57.00)
\put(47.50,47.50){\makebox(0,0){\circle*{9.50}}}

\path(95.00,57.00)(114.00,57.00)(114.00,38.00)(95.00,38.00)(95.00,57.00)
\put(104.50,47.50){\makebox(0,0){\circle*{9.50}}}

\path(152.00,57.00)(171.00,57.00)(171.00,38.00)(152.00,38.00)(152.00,57.00)
\put(161.50,47.50){\makebox(0,0){\circle*{9.50}}}

\path(209.00,57.00)(228.00,57.00)(228.00,38.00)(209.00,38.00)(209.00,57.00)
\put(218.50,47.50){\makebox(0,0){\circle*{9.50}}}

\put(84.5,85.5){\makebox(0,0)[r]{$0$}}
\put(85.5,85.5){\vector(1,0){30}}
\mathversion{normal}

\end{picture}
\end{center} \vspace{-20pt}
\caption{A wire that conducts zero} \label{zero}
\end{figure}
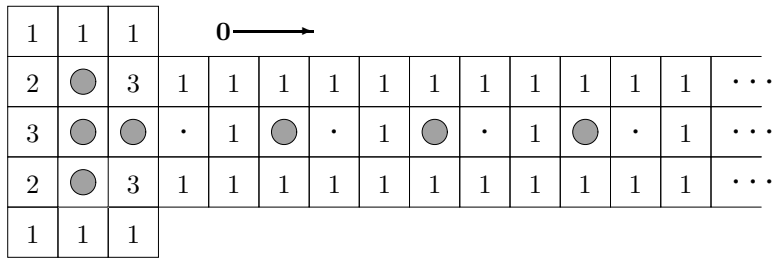

In figure \ref{wire}, we do not know what the compartments with $x$ and $x'$
contain, but we do know that either exactly all compartments with $x$ contain
a mine (figure \ref{one}, corresponding to $x = 1$) or exactly all compartments 
with $x'$ do (figure \ref{zero}, corresponding to $x = 0$). 
Which of both cases are possible follows from the global
structure of the Minesweeper board. To make polynomial equations from the
polynomial expressions, it suffices to force these expressions to have a 
given value in $\{0,1\}$. This can be done by forcing a wire that conducts a 
certain expression to conduct a given value. Figure \ref{force} shows how 
to force a wire to conduct one and zero respectively.

\begin{figure}[!htp]
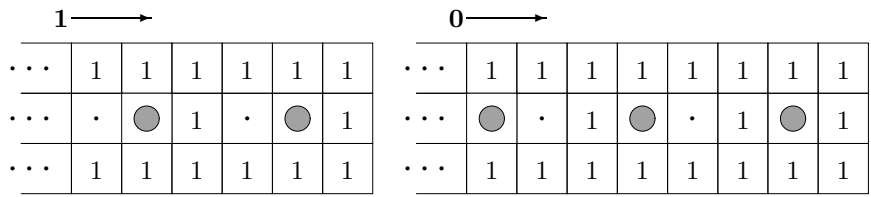

\begin{center}

\end{center} \vspace*{-20pt}
\caption{A wire that is forced to conduct a given value} \label{force}
\end{figure}

\section{Small computational components}

Next, we need to have curves to bend wires (figure \ref{curve}) and
splitters to duplicate wires (figure \ref{splitter}). 

\begin{figure}[!htp]
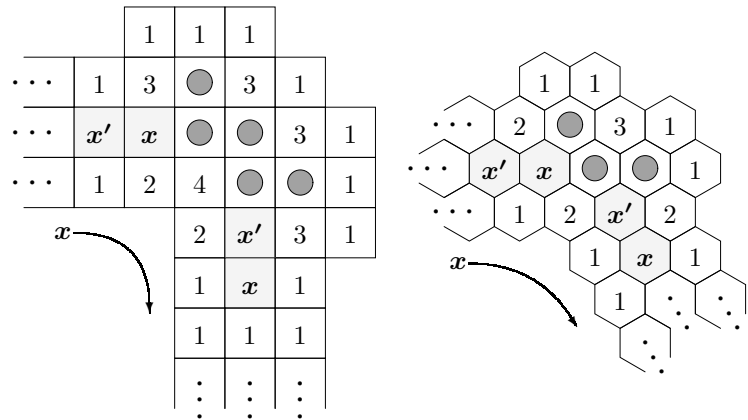

\begin{center}
%
\end{center} \vspace*{-20pt}
\caption{A curve of a wire with $x$} \label{curve}
\end{figure}

\begin{figure}[!htp]
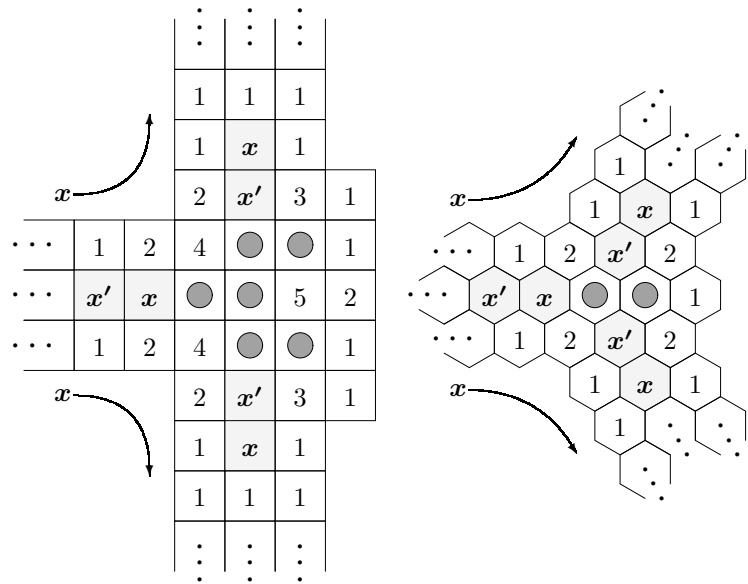

\begin{center}
%
\end{center} \vspace*{-20pt}
\caption{Merge two curves to get a splitter} \label{splitter}
\end{figure}

Something that we might need as well are so-called
phase-shifters (figure \ref{phase}). The phase-shifter with square
compartments is stolen from R. Kaye \cite{kaye}.

\begin{figure}[!htp]
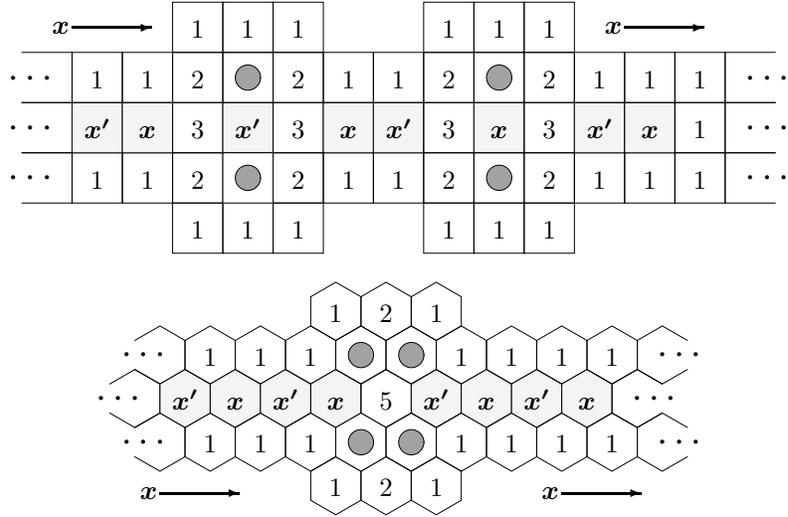

\begin{center}
%
\end{center} \vspace*{-20pt}
\caption{A phase-shifter} \label{phase}
\end{figure}

You might think that the phase-shifter with hexagonal
compartments is a so-called invertor as well. In that case, I do not agree
with you. A real invertor is shown in figure \ref{invertor}.

\begin{figure}[!htp]
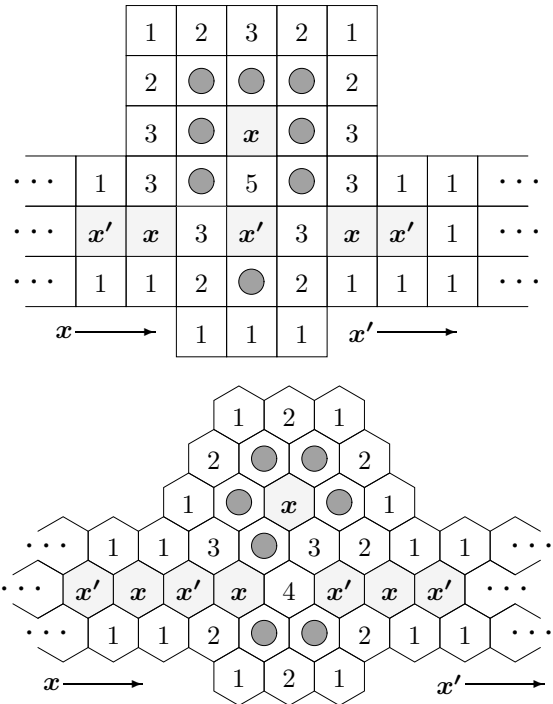

\begin{center}
%
\end{center} \vspace*{-20pt}
\caption{An invertor} \label{invertor}
\end{figure}

In case you might not have noticed already, the hexagonal phase-shifter
is not a good invertor because the number of mines might given information
about the circuit then, possibly screwing up all our efforts with it.

\section{Larger computational components}

\begin{figure}[!htp]
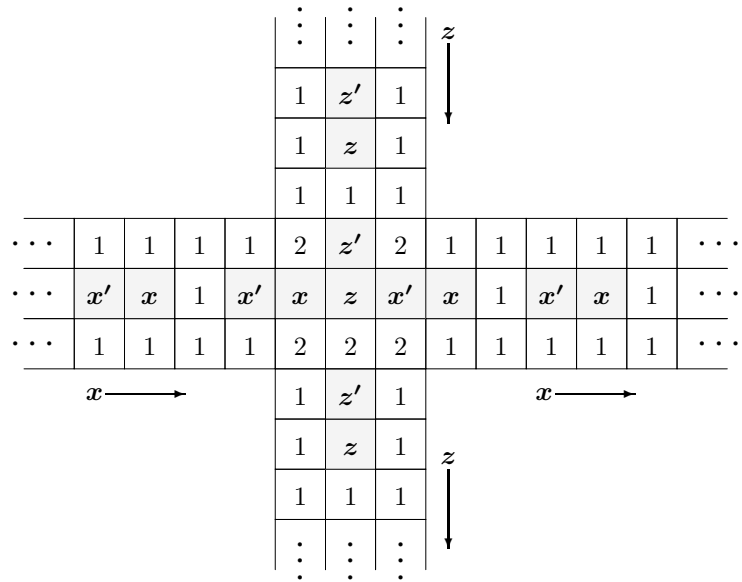

\begin{center}
%
\end{center} \vspace*{-20pt}
\caption{A crossover by R. Kaye} \label{crossover}
\end{figure}

\noindent
A crossover for Minesweeper with square compartments is shown in
figure \ref{crossover}. It is taken from Richards Kaye's slides, see
\cite{slides}. You might wonder why there is no
crossover given with hexagonal compartments. The answer is that I did not
found one by direct construction. But a crossover can also be made from 
three splitters and the same number of adders, where the adders act modulo $2$,
see \cite{kaye} or figure \ref{crossadders}.

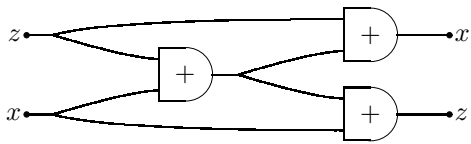
\begin{figure}[!htp]
\begin{center}
\begin{picture}(174,50)(-7,-10)

\put(0,30){\circle*{2}}
\put(-2,30){\makebox(0,0)[r]{$z$}}
\put(0,0){\circle*{2}}
\put(-2,0){\makebox(0,0)[r]{$x$}}

\put(0,30){\line(1,0){10}}
\put(0,0){\line(1,0){10}}
\qbezier(10,30)(30,36)(120,36)
\qbezier(10,30)(40,21)(50,21)
\qbezier(10,0)(40,9)(50,9)
\qbezier(10,0)(30,-6)(120,-6)

\put(50,5){\line(0,1){20}}
\put(50,15){\oval(40,20)[r]}
\put(60,15){\makebox(0,0){$+$}}
\put(70,15){\line(1,0){10}}
\qbezier(80,15)(110,24)(120,24)
\qbezier(80,15)(110,6)(120,6)

\put(120,20){\line(0,1){20}}
\put(120,30){\oval(40,20)[r]}
\put(130,30){\makebox(0,0){$+$}}
\put(140,30){\line(1,0){20}}
\put(120,-10){\line(0,1){20}}
\put(120,0){\oval(40,20)[r]}
\put(130,0){\makebox(0,0){$+$}}
\put(140,0){\line(1,0){20}}

\put(160,30){\circle*{2}}
\put(160,0){\circle*{2}}
\put(162,30){\makebox(0,0)[l]{$x$}}
\put(162,,0){\makebox(0,0)[l]{$z$}}

\end{picture}
\end{center} \vspace{-20pt}
\caption{A crossover circuit}
\label{crossadders}
\end{figure}

But before we have such a crossover, we must first make a hexagonal adder.
Figure \ref{addmul} shows an adder-multiplier combi, where the adder
acts modulo $2$. If you only want to use 
the adder, you just cut off the wire of the multiplier output. A wire cut-off
is the same as the start of a variable in figure \ref{wire}.

\begin{figure}[p]
\begin{center}
\begin{picture}(285.00,253.33)(0.00,-3.17)

\mathversion{bold}
\filltype{shade}
\texture{40004000 0 0 0 00400040 0 0 0
         40004000 0 0 0 00400040 0 0 0
         40004000 0 0 0 00400040 0 0 0
         40004000 0 0 0 00400040 0 0 0}
\whiten

\path(0.00,114.00)(19.00,114.00)(19.00,95.00)(0.00,95.00)(0.00,114.00)
\put(9.50,100.50){\makebox(0,0)[b]{1}}

\path(0.00,133.00)(19.00,133.00)(19.00,114.00)(0.00,114.00)(0.00,133.00)
\put(9.50,119.50){\makebox(0,0)[b]{1}}

\path(0.00,152.00)(19.00,152.00)(19.00,133.00)(0.00,133.00)(0.00,152.00)
\put(9.50,138.50){\makebox(0,0)[b]{1}}

\path(19.00,0.00)(19.00,19.00)(38.00,19.00)(38.00,0.00)
\put(28.50,9.50){\makebox(0,0){$\cdot$}}
\put(28.50,3.17){\makebox(0,0){$\cdot$}}
\put(28.50,-3.17){\makebox(0,0){$\cdot$}}

\path(19.00,38.00)(38.00,38.00)(38.00,19.00)(19.00,19.00)(19.00,38.00)
\put(28.50,24.50){\makebox(0,0)[b]{1}}

\path(19.00,57.00)(38.00,57.00)(38.00,38.00)(19.00,38.00)(19.00,57.00)
\put(28.50,43.50){\makebox(0,0)[b]{1}}

\path(19.00,76.00)(38.00,76.00)(38.00,57.00)(19.00,57.00)(19.00,76.00)
\put(28.50,62.50){\makebox(0,0)[b]{1}}

\path(19.00,95.00)(38.00,95.00)(38.00,76.00)(19.00,76.00)(19.00,95.00)
\put(28.50,81.50){\makebox(0,0)[b]{1}}

\path(19.00,114.00)(38.00,114.00)(38.00,95.00)(19.00,95.00)(19.00,114.00)
\put(28.50,100.50){\makebox(0,0)[b]{3}}

\path(19.00,152.00)(38.00,152.00)(38.00,133.00)(19.00,133.00)(19.00,152.00)
\put(28.50,138.50){\makebox(0,0)[b]{3}}

\path(19.00,171.00)(38.00,171.00)(38.00,152.00)(19.00,152.00)(19.00,171.00)
\put(28.50,157.50){\makebox(0,0)[b]{1}}

\path(19.00,190.00)(38.00,190.00)(38.00,171.00)(19.00,171.00)(19.00,190.00)
\put(28.50,176.50){\makebox(0,0)[b]{1}}

\path(19.00,209.00)(38.00,209.00)(38.00,190.00)(19.00,190.00)(19.00,209.00)
\put(28.50,195.50){\makebox(0,0)[b]{1}}

\path(19.00,228.00)(38.00,228.00)(38.00,209.00)(19.00,209.00)(19.00,228.00)
\put(28.50,214.50){\makebox(0,0)[b]{1}}

\path(38.00,247.00)(38.00,228.00)(19.00,228.00)(19.00,247.00)
\put(28.50,237.50){\makebox(0,0){$\cdot$}}
\put(28.50,243.83){\makebox(0,0){$\cdot$}}
\put(28.50,250.17){\makebox(0,0){$\cdot$}}

\path(38.00,0.00)(38.00,19.00)(57.00,19.00)(57.00,0.00)
\put(47.50,9.50){\makebox(0,0){$\cdot$}}
\put(47.50,3.17){\makebox(0,0){$\cdot$}}
\put(47.50,-3.17){\makebox(0,0){$\cdot$}}

\shade
\path(38.00,38.00)(57.00,38.00)(57.00,19.00)(38.00,19.00)(38.00,38.00)
\put(47.50,24.50){\makebox(0,0)[b]{$z'$}}

\shade
\path(38.00,57.00)(57.00,57.00)(57.00,38.00)(38.00,38.00)(38.00,57.00)
\put(47.50,43.50){\makebox(0,0)[b]{$z$}}

\path(38.00,76.00)(57.00,76.00)(57.00,57.00)(38.00,57.00)(38.00,76.00)
\put(47.50,62.50){\makebox(0,0)[b]{2}}

\shade
\path(38.00,95.00)(57.00,95.00)(57.00,76.00)(38.00,76.00)(38.00,95.00)
\put(47.50,81.50){\makebox(0,0)[b]{$z'$}}

\shade
\path(38.00,114.00)(57.00,114.00)(57.00,95.00)(38.00,95.00)(38.00,114.00)
\put(47.50,100.50){\makebox(0,0)[b]{$z$}}

\shade
\path(38.00,152.00)(57.00,152.00)(57.00,133.00)(38.00,133.00)(38.00,152.00)
\put(47.50,138.50){\makebox(0,0)[b]{$x$}}

\shade
\path(38.00,171.00)(57.00,171.00)(57.00,152.00)(38.00,152.00)(38.00,171.00)
\put(47.50,157.50){\makebox(0,0)[b]{$x'$}}

\path(38.00,190.00)(57.00,190.00)(57.00,171.00)(38.00,171.00)(38.00,190.00)
\put(47.50,176.50){\makebox(0,0)[b]{2}}

\shade
\path(38.00,209.00)(57.00,209.00)(57.00,190.00)(38.00,190.00)(38.00,209.00)
\put(47.50,195.50){\makebox(0,0)[b]{$x$}}

\shade
\path(38.00,228.00)(57.00,228.00)(57.00,209.00)(38.00,209.00)(38.00,228.00)
\put(47.50,214.50){\makebox(0,0)[b]{$x'$}}

\path(57.00,247.00)(57.00,228.00)(38.00,228.00)(38.00,247.00)
\put(47.50,237.50){\makebox(0,0){$\cdot$}}
\put(47.50,243.83){\makebox(0,0){$\cdot$}}
\put(47.50,250.17){\makebox(0,0){$\cdot$}}

\path(57.00,0.00)(57.00,19.00)(76.00,19.00)(76.00,0.00)
\put(66.50,9.50){\makebox(0,0){$\cdot$}}
\put(66.50,3.17){\makebox(0,0){$\cdot$}}
\put(66.50,-3.17){\makebox(0,0){$\cdot$}}

\path(57.00,38.00)(76.00,38.00)(76.00,19.00)(57.00,19.00)(57.00,38.00)
\put(66.50,24.50){\makebox(0,0)[b]{1}}

\path(57.00,57.00)(76.00,57.00)(76.00,38.00)(57.00,38.00)(57.00,57.00)
\put(66.50,43.50){\makebox(0,0)[b]{1}}

\path(57.00,76.00)(76.00,76.00)(76.00,57.00)(57.00,57.00)(57.00,76.00)
\put(66.50,62.50){\makebox(0,0)[b]{2}}

\path(57.00,190.00)(76.00,190.00)(76.00,171.00)(57.00,171.00)(57.00,190.00)
\put(66.50,176.50){\makebox(0,0)[b]{2}}

\path(57.00,209.00)(76.00,209.00)(76.00,190.00)(57.00,190.00)(57.00,209.00)
\put(66.50,195.50){\makebox(0,0)[b]{1}}

\path(57.00,228.00)(76.00,228.00)(76.00,209.00)(57.00,209.00)(57.00,228.00)
\put(66.50,214.50){\makebox(0,0)[b]{1}}

\path(76.00,247.00)(76.00,228.00)(57.00,228.00)(57.00,247.00)
\put(66.50,237.50){\makebox(0,0){$\cdot$}}
\put(66.50,243.83){\makebox(0,0){$\cdot$}}
\put(66.50,250.17){\makebox(0,0){$\cdot$}}

\path(76.00,76.00)(95.00,76.00)(95.00,57.00)(76.00,57.00)(76.00,76.00)
\put(85.50,62.50){\makebox(0,0)[b]{1}}

\path(76.00,95.00)(95.00,95.00)(95.00,76.00)(76.00,76.00)(76.00,95.00)
\put(85.50,81.50){\makebox(0,0)[b]{3}}

\shade
\path(76.00,114.00)(95.00,114.00)(95.00,95.00)(76.00,95.00)(76.00,114.00)
\put(85.50,100.50){\makebox(0,0)[b]{$v'$}}

\path(76.00,133.00)(95.00,133.00)(95.00,114.00)(76.00,114.00)(76.00,133.00)
\put(85.50,119.50){\makebox(0,0)[b]{4}}

\shade
\path(76.00,152.00)(95.00,152.00)(95.00,133.00)(76.00,133.00)(76.00,152.00)
\put(85.50,138.50){\makebox(0,0)[b]{$a'$}}

\path(76.00,171.00)(95.00,171.00)(95.00,152.00)(76.00,152.00)(76.00,171.00)
\put(85.50,157.50){\makebox(0,0)[b]{3}}

\path(76.00,190.00)(95.00,190.00)(95.00,171.00)(76.00,171.00)(76.00,190.00)
\put(85.50,176.50){\makebox(0,0)[b]{1}}

\path(95.00,95.00)(114.00,95.00)(114.00,76.00)(95.00,76.00)(95.00,95.00)
\put(104.50,81.50){\makebox(0,0)[b]{1}}

\shade
\path(95.00,114.00)(114.00,114.00)(114.00,95.00)(95.00,95.00)(95.00,114.00)
\put(104.50,100.50){\makebox(0,0)[b]{$v$}}

\path(95.00,133.00)(114.00,133.00)(114.00,114.00)(95.00,114.00)(95.00,133.00)
\put(104.50,119.50){\makebox(0,0)[b]{2}}

\shade
\path(95.00,152.00)(114.00,152.00)(114.00,133.00)(95.00,133.00)(95.00,152.00)
\put(104.50,138.50){\makebox(0,0)[b]{$a$}}

\path(95.00,171.00)(114.00,171.00)(114.00,152.00)(95.00,152.00)(95.00,171.00)
\put(104.50,157.50){\makebox(0,0)[b]{1}}

\path(114.00,95.00)(133.00,95.00)(133.00,76.00)(114.00,76.00)(114.00,95.00)
\put(123.50,81.50){\makebox(0,0)[b]{1}}

\path(114.00,114.00)(133.00,114.00)(133.00,95.00)(114.00,95.00)(114.00,114.00)
\put(123.50,100.50){\makebox(0,0)[b]{1}}

\path(114.00,133.00)(133.00,133.00)(133.00,114.00)(114.00,114.00)(114.00,133.00)
\put(123.50,119.50){\makebox(0,0)[b]{2}}

\path(114.00,152.00)(133.00,152.00)(133.00,133.00)(114.00,133.00)(114.00,152.00)
\put(123.50,138.50){\makebox(0,0)[b]{1}}

\path(114.00,171.00)(133.00,171.00)(133.00,152.00)(114.00,152.00)(114.00,171.00)
\put(123.50,157.50){\makebox(0,0)[b]{1}}

\path(133.00,95.00)(152.00,95.00)(152.00,76.00)(133.00,76.00)(133.00,95.00)
\put(142.50,81.50){\makebox(0,0)[b]{1}}

\shade
\path(133.00,114.00)(152.00,114.00)(152.00,95.00)(133.00,95.00)(133.00,114.00)
\put(142.50,100.50){\makebox(0,0)[b]{$v'$}}

\path(133.00,133.00)(152.00,133.00)(152.00,114.00)(133.00,114.00)(133.00,133.00)
\put(142.50,119.50){\makebox(0,0)[b]{2}}

\shade
\path(133.00,152.00)(152.00,152.00)(152.00,133.00)(133.00,133.00)(133.00,152.00)
\put(142.50,138.50){\makebox(0,0)[b]{$a'$}}

\path(133.00,171.00)(152.00,171.00)(152.00,152.00)(133.00,152.00)(133.00,171.00)
\put(142.50,157.50){\makebox(0,0)[b]{1}}

\path(152.00,95.00)(171.00,95.00)(171.00,76.00)(152.00,76.00)(152.00,95.00)
\put(161.50,81.50){\makebox(0,0)[b]{1}}

\shade
\path(152.00,114.00)(171.00,114.00)(171.00,95.00)(152.00,95.00)(152.00,114.00)
\put(161.50,100.50){\makebox(0,0)[b]{$v$}}

\path(152.00,133.00)(171.00,133.00)(171.00,114.00)(152.00,114.00)(152.00,133.00)
\put(161.50,119.50){\makebox(0,0)[b]{4}}

\shade
\path(152.00,152.00)(171.00,152.00)(171.00,133.00)(152.00,133.00)(152.00,152.00)
\put(161.50,138.50){\makebox(0,0)[b]{$a$}}

\path(152.00,171.00)(171.00,171.00)(171.00,152.00)(152.00,152.00)(152.00,171.00)
\put(161.50,157.50){\makebox(0,0)[b]{2}}

\path(171.00,95.00)(190.00,95.00)(190.00,76.00)(171.00,76.00)(171.00,95.00)
\put(180.50,81.50){\makebox(0,0)[b]{1}}

\path(171.00,114.00)(190.00,114.00)(190.00,95.00)(171.00,95.00)(171.00,114.00)
\put(180.50,100.50){\makebox(0,0)[b]{2}}

\path(171.00,171.00)(190.00,171.00)(190.00,152.00)(171.00,152.00)(171.00,171.00)
\put(180.50,157.50){\makebox(0,0)[b]{4}}

\path(171.00,190.00)(190.00,190.00)(190.00,171.00)(171.00,171.00)(171.00,190.00)
\put(180.50,176.50){\makebox(0,0)[b]{2}}

\path(171.00,209.00)(190.00,209.00)(190.00,190.00)(171.00,190.00)(171.00,209.00)
\put(180.50,195.50){\makebox(0,0)[b]{1}}

\path(171.00,228.00)(190.00,228.00)(190.00,209.00)(171.00,209.00)(171.00,228.00)
\put(180.50,214.50){\makebox(0,0)[b]{1}}

\path(190.00,247.00)(190.00,228.00)(171.00,228.00)(171.00,247.00)
\put(180.50,237.50){\makebox(0,0){$\cdot$}}
\put(180.50,243.83){\makebox(0,0){$\cdot$}}
\put(180.50,250.17){\makebox(0,0){$\cdot$}}

\path(190.00,95.00)(209.00,95.00)(209.00,76.00)(190.00,76.00)(190.00,95.00)
\put(199.50,81.50){\makebox(0,0)[b]{1}}

\shade
\path(190.00,114.00)(209.00,114.00)(209.00,95.00)(190.00,95.00)(190.00,114.00)
\put(199.50,100.50){\makebox(0,0)[b]{$v'$}}

\path(190.00,133.00)(209.00,133.00)(209.00,114.00)(190.00,114.00)(190.00,133.00)
\put(199.50,119.50){\makebox(0,0)[b]{6}}

\shade
\path(190.00,190.00)(209.00,190.00)(209.00,171.00)(190.00,171.00)(190.00,190.00)
\put(199.50,176.50){\makebox(0,0)[b]{$a'$}}

\shade
\path(190.00,209.00)(209.00,209.00)(209.00,190.00)(190.00,190.00)(190.00,209.00)
\put(199.50,195.50){\makebox(0,0)[b]{$a$}}

\path(190.00,228.00)(209.00,228.00)(209.00,209.00)(190.00,209.00)(190.00,228.00)
\put(199.50,214.50){\makebox(0,0)[b]{1}}

\path(209.00,247.00)(209.00,228.00)(190.00,228.00)(190.00,247.00)
\put(199.50,237.50){\makebox(0,0){$\cdot$}}
\put(199.50,243.83){\makebox(0,0){$\cdot$}}
\put(199.50,250.17){\makebox(0,0){$\cdot$}}

\path(209.00,19.00)(209.00,38.00)(228.00,38.00)(228.00,19.00)
\put(218.50,28.50){\makebox(0,0){$\cdot$}}
\put(218.50,22.17){\makebox(0,0){$\cdot$}}
\put(218.50,15.83){\makebox(0,0){$\cdot$}}

\path(209.00,57.00)(228.00,57.00)(228.00,38.00)(209.00,38.00)(209.00,57.00)
\put(218.50,43.50){\makebox(0,0)[b]{1}}

\path(209.00,76.00)(228.00,76.00)(228.00,57.00)(209.00,57.00)(209.00,76.00)
\put(218.50,62.50){\makebox(0,0)[b]{1}}

\path(209.00,95.00)(228.00,95.00)(228.00,76.00)(209.00,76.00)(209.00,95.00)
\put(218.50,81.50){\makebox(0,0)[b]{2}}

\path(209.00,171.00)(228.00,171.00)(228.00,152.00)(209.00,152.00)(209.00,171.00)
\put(218.50,157.50){\makebox(0,0)[b]{5}}

\path(209.00,190.00)(228.00,190.00)(228.00,171.00)(209.00,171.00)(209.00,190.00)
\put(218.50,176.50){\makebox(0,0)[b]{3}}

\path(209.00,209.00)(228.00,209.00)(228.00,190.00)(209.00,190.00)(209.00,209.00)
\put(218.50,195.50){\makebox(0,0)[b]{1}}

\path(209.00,228.00)(228.00,228.00)(228.00,209.00)(209.00,209.00)(209.00,228.00)
\put(218.50,214.50){\makebox(0,0)[b]{1}}

\path(228.00,247.00)(228.00,228.00)(209.00,228.00)(209.00,247.00)
\put(218.50,237.50){\makebox(0,0){$\cdot$}}
\put(218.50,243.83){\makebox(0,0){$\cdot$}}
\put(218.50,250.17){\makebox(0,0){$\cdot$}}

\path(228.00,19.00)(228.00,38.00)(247.00,38.00)(247.00,19.00)
\put(237.50,28.50){\makebox(0,0){$\cdot$}}
\put(237.50,22.17){\makebox(0,0){$\cdot$}}
\put(237.50,15.83){\makebox(0,0){$\cdot$}}

\path(228.00,57.00)(247.00,57.00)(247.00,38.00)(228.00,38.00)(228.00,57.00)
\put(237.50,43.50){\makebox(0,0)[b]{1}}

\shade
\path(228.00,76.00)(247.00,76.00)(247.00,57.00)(228.00,57.00)(228.00,76.00)
\put(237.50,62.50){\makebox(0,0)[b]{$c$}}

\shade
\path(228.00,152.00)(247.00,152.00)(247.00,133.00)(228.00,133.00)(228.00,152.00)
\put(237.50,138.50){\makebox(0,0)[b]{$a$}}

\path(228.00,190.00)(247.00,190.00)(247.00,171.00)(228.00,171.00)(228.00,190.00)
\put(237.50,176.50){\makebox(0,0)[b]{2}}

\path(247.00,19.00)(247.00,38.00)(266.00,38.00)(266.00,19.00)
\put(256.50,28.50){\makebox(0,0){$\cdot$}}
\put(256.50,22.17){\makebox(0,0){$\cdot$}}
\put(256.50,15.83){\makebox(0,0){$\cdot$}}

\path(247.00,57.00)(266.00,57.00)(266.00,38.00)(247.00,38.00)(247.00,57.00)
\put(256.50,43.50){\makebox(0,0)[b]{1}}

\path(247.00,76.00)(266.00,76.00)(266.00,57.00)(247.00,57.00)(247.00,76.00)
\put(256.50,62.50){\makebox(0,0)[b]{1}}

\path(247.00,95.00)(266.00,95.00)(266.00,76.00)(247.00,76.00)(247.00,95.00)
\put(256.50,81.50){\makebox(0,0)[b]{2}}

\path(247.00,133.00)(266.00,133.00)(266.00,114.00)(247.00,114.00)(247.00,133.00)
\put(256.50,119.50){\makebox(0,0)[b]{2}}

\path(247.00,152.00)(266.00,152.00)(266.00,133.00)(247.00,133.00)(247.00,152.00)
\put(256.50,138.50){\makebox(0,0)[b]{3}}

\path(247.00,190.00)(266.00,190.00)(266.00,171.00)(247.00,171.00)(247.00,190.00)
\put(256.50,176.50){\makebox(0,0)[b]{2}}

\path(266.00,95.00)(285.00,95.00)(285.00,76.00)(266.00,76.00)(266.00,95.00)
\put(275.50,81.50){\makebox(0,0)[b]{1}}

\path(266.00,114.00)(285.00,114.00)(285.00,95.00)(266.00,95.00)(266.00,114.00)
\put(275.50,100.50){\makebox(0,0)[b]{1}}

\path(266.00,133.00)(285.00,133.00)(285.00,114.00)(266.00,114.00)(266.00,133.00)
\put(275.50,119.50){\makebox(0,0)[b]{1}}

\path(266.00,152.00)(285.00,152.00)(285.00,133.00)(266.00,133.00)(266.00,152.00)
\put(275.50,138.50){\makebox(0,0)[b]{1}}

\path(266.00,171.00)(285.00,171.00)(285.00,152.00)(266.00,152.00)(266.00,171.00)
\put(275.50,157.50){\makebox(0,0)[b]{1}}

\path(266.00,190.00)(285.00,190.00)(285.00,171.00)(266.00,171.00)(266.00,190.00)
\put(275.50,176.50){\makebox(0,0)[b]{1}}

\texture{c0c0c0c0 0 0 0 0c0c0c0c 0 0 0
          c0c0c0c0 0 0 0 0c0c0c0c 0 0 0
          c0c0c0c0 0 0 0 0c0c0c0c 0 0 0
          c0c0c0c0 0 0 0 0c0c0c0c 0 0 0}

\shade
\path(57.00,133.00)(76.00,133.00)(76.00,114.00)(57.00,114.00)(57.00,133.00)
\put(66.50,119.50){\makebox(0,0)[b]{$5$}}

\shade
\path(228.00,114.00)(247.00,114.00)(247.00,95.00)(228.00,95.00)(228.00,114.00)
\put(237.50,100.50){\makebox(0,0)[b]{$4$}}

\texture{cccccccc 0 0 0 cccccccc 0 0 0
         cccccccc 0 0 0 cccccccc 0 0 0
         cccccccc 0 0 0 cccccccc 0 0 0
         cccccccc 0 0 0 cccccccc 0 0 0}
\whiten

\path(19.00,133.00)(38.00,133.00)(38.00,114.00)(19.00,114.00)(19.00,133.00)
\put(28.50,123.50){\makebox(0,0){\circle*{9.50}}}

\path(38.00,133.00)(57.00,133.00)(57.00,114.00)(38.00,114.00)(38.00,133.00)
\put(47.50,123.50){\makebox(0,0){\circle*{9.50}}}

\path(57.00,95.00)(76.00,95.00)(76.00,76.00)(57.00,76.00)(57.00,95.00)
\put(66.50,85.50){\makebox(0,0){\circle*{9.50}}}

\path(57.00,114.00)(76.00,114.00)(76.00,95.00)(57.00,95.00)(57.00,114.00)
\put(66.50,104.50){\makebox(0,0){\circle*{9.50}}}

\path(57.00,152.00)(76.00,152.00)(76.00,133.00)(57.00,133.00)(57.00,152.00)
\put(66.50,142.50){\makebox(0,0){\circle*{9.50}}}

\path(57.00,171.00)(76.00,171.00)(76.00,152.00)(57.00,152.00)(57.00,171.00)
\put(66.50,161.50){\makebox(0,0){\circle*{9.50}}}

\path(171.00,133.00)(190.00,133.00)(190.00,114.00)(171.00,114.00)(171.00,133.00)
\put(180.50,123.50){\makebox(0,0){\circle*{9.50}}}

\path(171.00,152.00)(190.00,152.00)(190.00,133.00)(171.00,133.00)(171.00,152.00)
\put(180.50,142.50){\makebox(0,0){\circle*{9.50}}}

\path(190.00,152.00)(209.00,152.00)(209.00,133.00)(190.00,133.00)(190.00,152.00)
\put(199.50,142.50){\makebox(0,0){\circle*{9.50}}}

\path(190.00,171.00)(209.00,171.00)(209.00,152.00)(190.00,152.00)(190.00,171.00)
\put(199.50,161.50){\makebox(0,0){\circle*{9.50}}}

\path(209.00,133.00)(228.00,133.00)(228.00,114.00)(209.00,114.00)(209.00,133.00)
\put(218.50,123.50){\makebox(0,0){\circle*{9.50}}}

\path(209.00,152.00)(228.00,152.00)(228.00,133.00)(209.00,133.00)(209.00,152.00)
\put(218.50,142.50){\makebox(0,0){\circle*{9.50}}}

\path(228.00,171.00)(247.00,171.00)(247.00,152.00)(228.00,152.00)(228.00,171.00)
\put(237.50,161.50){\makebox(0,0){\circle*{9.50}}}

\path(247.00,114.00)(266.00,114.00)(266.00,95.00)(247.00,95.00)(247.00,114.00)
\put(256.50,104.50){\makebox(0,0){\circle*{9.50}}}

\path(247.00,171.00)(266.00,171.00)(266.00,152.00)(247.00,152.00)(247.00,171.00)
\put(256.50,161.50){\makebox(0,0){\circle*{9.50}}}

\texture{40004000 0 0 0 00400040 0 0 0
          40004000 0 0 0 00400040 0 0 0
          40004000 0 0 0 00400040 0 0 0
          40004000 0 0 0 00400040 0 0 0}

\Thicklines
\shade
\path(209.00,114.00)(228.00,114.00)(228.00,95.00)(209.00,95.00)(209.00,114.00)
\put(218.50,100.50){\makebox(0,0)[b]{$v$}}
\thinlines

\Thicklines
\shade
\path(228.00,95.00)(247.00,95.00)(247.00,76.00)(228.00,76.00)(228.00,95.00)
\put(237.50,81.50){\makebox(0,0)[b]{$c'$}}
\thinlines

\Thicklines
\shade
\path(228.00,133.00)(247.00,133.00)(247.00,114.00)(228.00,114.00)(228.00,133.00)
\put(237.50,119.50){\makebox(0,0)[b]{$a'$}}
\thinlines

\put(9.5,36){\makebox(0,0)[t]{$z$}}
\put(9.5,38){\vector(0,1){30}}
\put(9.5,211){\makebox(0,0)[b]{$x$}}
\put(9.5,209){\vector(0,-1){30}}
\dashline{2}(95,0)(95,57)
\dashline{2}(95,190)(95,247)
\dashline{2}(152,0)(152,76)
\dashline{2}(152,171)(152,247)
\put(237.5,207){\makebox(0,0)[t]{$a$}}
\put(237.5,209){\vector(0,1){30}}
\put(275.5,59){\makebox(0,0)[b]{$c$}}
\put(275.5,57){\vector(0,-1){30}}
\mathversion{normal}

\end{picture} \\[10pt]
\begin{picture}(272.62,236.50)(-6.62,0.00)

\mathversion{bold}
\filltype{shade}
\texture{40004000 0 0 0 00400040 0 0 0
         40004000 0 0 0 00400040 0 0 0
         40004000 0 0 0 00400040 0 0 0
         40004000 0 0 0 00400040 0 0 0}
\whiten

\path(9.50,49.50)(9.50,38.50)(19.00,33.00)(28.50,38.50)
     (28.50,49.50)(19.00,55.00)(9.50,49.50)
\put(19.00,40.00){\makebox(0,0)[b]{1}}

\path(9.50,16.50)(19.00,22.00)(19.00,33.00)
     (9.50,38.50)(0.00,33.00)(0.00,22.00)
\put(9.50,27.50){\makebox(0,0){$\cdot$}}
\put(6.33,22.00){\makebox(0,0){$\cdot$}}
\put(3.17,16.50){\makebox(0,0){$\cdot$}}

\path(0.00,165.00)(0.00,154.00)(9.50,148.50)
     (19.00,154.00)(19.00,165.00)(9.50,170.50)
\put(9.50,159.50){\makebox(0,0){$\cdot$}}
\put(6.33,165.00){\makebox(0,0){$\cdot$}}
\put(3.17,170.50){\makebox(0,0){$\cdot$}}

\path(19.00,0.00)(28.50,5.50)(28.50,16.50)
     (19.00,22.00)(9.50,16.50)(9.50,5.50)
\put(19.00,11.00){\makebox(0,0){$\cdot$}}
\put(15.83,5.50){\makebox(0,0){$\cdot$}}
\put(12.67,0.00){\makebox(0,0){$\cdot$}}

\path(9.50,148.50)(9.50,137.50)(19.00,132.00)(28.50,137.50)
     (28.50,148.50)(19.00,154.00)(9.50,148.50)
\put(19.00,139.00){\makebox(0,0)[b]{1}}

\path(9.50,181.50)(9.50,170.50)(19.00,165.00)
     (28.50,170.50)(28.50,181.50)(19.00,187.00)
\put(19.00,176.00){\makebox(0,0){$\cdot$}}
\put(15.83,181.50){\makebox(0,0){$\cdot$}}
\put(12.67,187.00){\makebox(0,0){$\cdot$}}

\shade
\path(19.00,33.00)(19.00,22.00)(28.50,16.50)(38.00,22.00)
     (38.00,33.00)(28.50,38.50)(19.00,33.00)
\put(28.50,23.50){\makebox(0,0)[b]{$z'$}}

\path(19.00,66.00)(19.00,55.00)(28.50,49.50)(38.00,55.00)
     (38.00,66.00)(28.50,71.50)(19.00,66.00)
\put(28.50,56.50){\makebox(0,0)[b]{1}}

\path(19.00,99.00)(19.00,88.00)(28.50,82.50)(38.00,88.00)
     (38.00,99.00)(28.50,104.50)(19.00,99.00)
\put(28.50,89.50){\makebox(0,0)[b]{1}}

\path(19.00,132.00)(19.00,121.00)(28.50,115.50)(38.00,121.00)
     (38.00,132.00)(28.50,137.50)(19.00,132.00)
\put(28.50,122.50){\makebox(0,0)[b]{1}}

\shade
\path(19.00,165.00)(19.00,154.00)(28.50,148.50)(38.00,154.00)
     (38.00,165.00)(28.50,170.50)(19.00,165.00)
\put(28.50,155.50){\makebox(0,0)[b]{$x'$}}

\path(38.00,0.00)(47.50,5.50)(47.50,16.50)
     (38.00,22.00)(28.50,16.50)(28.50,5.50)
\put(38.00,11.00){\makebox(0,0){$\cdot$}}
\put(34.83,5.50){\makebox(0,0){$\cdot$}}
\put(31.67,0.00){\makebox(0,0){$\cdot$}}

\shade
\path(28.50,49.50)(28.50,38.50)(38.00,33.00)(47.50,38.50)
     (47.50,49.50)(38.00,55.00)(28.50,49.50)
\put(38.00,40.00){\makebox(0,0)[b]{$z$}}

\path(28.50,82.50)(28.50,71.50)(38.00,66.00)(47.50,71.50)
     (47.50,82.50)(38.00,88.00)(28.50,82.50)
\put(38.00,73.00){\makebox(0,0)[b]{2}}

\path(28.50,115.50)(28.50,104.50)(38.00,99.00)(47.50,104.50)
     (47.50,115.50)(38.00,121.00)(28.50,115.50)
\put(38.00,106.00){\makebox(0,0)[b]{2}}

\shade
\path(28.50,148.50)(28.50,137.50)(38.00,132.00)(47.50,137.50)
     (47.50,148.50)(38.00,154.00)(28.50,148.50)
\put(38.00,139.00){\makebox(0,0)[b]{$x$}}

\path(28.50,181.50)(28.50,170.50)(38.00,165.00)
     (47.50,170.50)(47.50,181.50)(38.00,187.00)
\put(38.00,176.00){\makebox(0,0){$\cdot$}}
\put(34.83,181.50){\makebox(0,0){$\cdot$}}
\put(31.67,187.00){\makebox(0,0){$\cdot$}}

\path(38.00,33.00)(38.00,22.00)(47.50,16.50)(57.00,22.00)
     (57.00,33.00)(47.50,38.50)(38.00,33.00)
\put(47.50,23.50){\makebox(0,0)[b]{1}}

\shade
\path(38.00,66.00)(38.00,55.00)(47.50,49.50)(57.00,55.00)
     (57.00,66.00)(47.50,71.50)(38.00,66.00)
\put(47.50,56.50){\makebox(0,0)[b]{$z'$}}

\shade
\path(38.00,132.00)(38.00,121.00)(47.50,115.50)(57.00,121.00)
     (57.00,132.00)(47.50,137.50)(38.00,132.00)
\put(47.50,122.50){\makebox(0,0)[b]{$x'$}}

\path(38.00,165.00)(38.00,154.00)(47.50,148.50)(57.00,154.00)
     (57.00,165.00)(47.50,170.50)(38.00,165.00)
\put(47.50,155.50){\makebox(0,0)[b]{1}}

\path(47.50,49.50)(47.50,38.50)(57.00,33.00)(66.50,38.50)
     (66.50,49.50)(57.00,55.00)(47.50,49.50)
\put(57.00,40.00){\makebox(0,0)[b]{1}}

\shade
\path(47.50,82.50)(47.50,71.50)(57.00,66.00)(66.50,71.50)
     (66.50,82.50)(57.00,88.00)(47.50,82.50)
\put(57.00,73.00){\makebox(0,0)[b]{$z$}}

\path(47.50,148.50)(47.50,137.50)(57.00,132.00)(66.50,137.50)
     (66.50,148.50)(57.00,154.00)(47.50,148.50)
\put(57.00,139.00){\makebox(0,0)[b]{2}}

\path(57.00,66.00)(57.00,55.00)(66.50,49.50)(76.00,55.00)
     (76.00,66.00)(66.50,71.50)(57.00,66.00)
\put(66.50,56.50){\makebox(0,0)[b]{2}}

\path(66.50,148.50)(66.50,137.50)(76.00,132.00)(85.50,137.50)
     (85.50,148.50)(76.00,154.00)(66.50,148.50)
\put(76.00,139.00){\makebox(0,0)[b]{1}}

\path(76.00,66.00)(76.00,55.00)(85.50,49.50)(95.00,55.00)
     (95.00,66.00)(85.50,71.50)(76.00,66.00)
\put(85.50,56.50){\makebox(0,0)[b]{1}}

\shade
\path(76.00,99.00)(76.00,88.00)(85.50,82.50)(95.00,88.00)
     (95.00,99.00)(85.50,104.50)(76.00,99.00)
\put(85.50,89.50){\makebox(0,0)[b]{$v'$}}

\path(76.00,132.00)(76.00,121.00)(85.50,115.50)(95.00,121.00)
     (95.00,132.00)(85.50,137.50)(76.00,132.00)
\put(85.50,122.50){\makebox(0,0)[b]{2}}

\path(76.00,165.00)(76.00,154.00)(85.50,148.50)(95.00,154.00)
     (95.00,165.00)(85.50,170.50)(76.00,165.00)
\put(85.50,155.50){\makebox(0,0)[b]{1}}

\path(85.50,82.50)(85.50,71.50)(95.00,66.00)(104.50,71.50)
     (104.50,82.50)(95.00,88.00)(85.50,82.50)
\put(95.00,73.00){\makebox(0,0)[b]{2}}

\shade
\path(85.50,115.50)(85.50,104.50)(95.00,99.00)(104.50,104.50)
     (104.50,115.50)(95.00,121.00)(85.50,115.50)
\put(95.00,106.00){\makebox(0,0)[b]{$a$}}

\path(85.50,148.50)(85.50,137.50)(95.00,132.00)(104.50,137.50)
     (104.50,148.50)(95.00,154.00)(85.50,148.50)
\put(95.00,139.00){\makebox(0,0)[b]{1}}

\path(85.50,181.50)(85.50,170.50)(95.00,165.00)(104.50,170.50)
     (104.50,181.50)(95.00,187.00)(85.50,181.50)
\put(95.00,172.00){\makebox(0,0)[b]{1}}

\path(95.00,132.00)(95.00,121.00)(104.50,115.50)(114.00,121.00)
     (114.00,132.00)(104.50,137.50)(95.00,132.00)
\put(104.50,122.50){\makebox(0,0)[b]{2}}

\path(104.50,82.50)(104.50,71.50)(114.00,66.00)(123.50,71.50)
     (123.50,82.50)(114.00,88.00)(104.50,82.50)
\put(114.00,73.00){\makebox(0,0)[b]{1}}

\path(104.50,148.50)(104.50,137.50)(114.00,132.00)(123.50,137.50)
     (123.50,148.50)(114.00,154.00)(104.50,148.50)
\put(114.00,139.00){\makebox(0,0)[b]{2}}

\path(104.50,181.50)(104.50,170.50)(114.00,165.00)(123.50,170.50)
     (123.50,181.50)(114.00,187.00)(104.50,181.50)
\put(114.00,172.00){\makebox(0,0)[b]{2}}

\shade
\path(114.00,99.00)(114.00,88.00)(123.50,82.50)(133.00,88.00)
     (133.00,99.00)(123.50,104.50)(114.00,99.00)
\put(123.50,89.50){\makebox(0,0)[b]{$v'$}}

\shade
\path(114.00,132.00)(114.00,121.00)(123.50,115.50)(133.00,121.00)
     (133.00,132.00)(123.50,137.50)(114.00,132.00)
\put(123.50,122.50){\makebox(0,0)[b]{$a'$}}

\shade
\path(114.00,165.00)(114.00,154.00)(123.50,148.50)(133.00,154.00)
     (133.00,165.00)(123.50,170.50)(114.00,165.00)
\put(123.50,155.50){\makebox(0,0)[b]{$a$}}

\path(114.00,198.00)(114.00,187.00)(123.50,181.50)(133.00,187.00)
     (133.00,198.00)(123.50,203.50)(114.00,198.00)
\put(123.50,188.50){\makebox(0,0)[b]{1}}

\path(123.50,82.50)(123.50,71.50)(133.00,66.00)(142.50,71.50)
     (142.50,82.50)(133.00,88.00)(123.50,82.50)
\put(133.00,73.00){\makebox(0,0)[b]{1}}

\path(123.50,148.50)(123.50,137.50)(133.00,132.00)(142.50,137.50)
     (142.50,148.50)(133.00,154.00)(123.50,148.50)
\put(133.00,139.00){\makebox(0,0)[b]{1}}

\shade
\path(123.50,181.50)(123.50,170.50)(133.00,165.00)(142.50,170.50)
     (142.50,181.50)(133.00,187.00)(123.50,181.50)
\put(133.00,172.00){\makebox(0,0)[b]{$a'$}}

\path(123.50,214.50)(123.50,203.50)(133.00,198.00)(142.50,203.50)
     (142.50,214.50)(133.00,220.00)(123.50,214.50)
\put(133.00,205.00){\makebox(0,0)[b]{1}}

\path(133.00,132.00)(133.00,121.00)(142.50,115.50)(152.00,121.00)
     (152.00,132.00)(142.50,137.50)(133.00,132.00)
\put(142.50,122.50){\makebox(0,0)[b]{2}}

\path(133.00,165.00)(133.00,154.00)(142.50,148.50)(152.00,154.00)
     (152.00,165.00)(142.50,170.50)(133.00,165.00)
\put(142.50,155.50){\makebox(0,0)[b]{1}}

\shade
\path(133.00,198.00)(133.00,187.00)(142.50,181.50)(152.00,187.00)
     (152.00,198.00)(142.50,203.50)(133.00,198.00)
\put(142.50,188.50){\makebox(0,0)[b]{$a$}}

\path(142.50,236.50)(133.00,231.00)(133.00,220.00)
     (142.50,214.50)(152.00,220.00)(152.00,231.00)
\put(142.50,225.50){\makebox(0,0){$\cdot$}}
\put(145.67,231.00){\makebox(0,0){$\cdot$}}
\put(148.83,236.50){\makebox(0,0){$\cdot$}}

\path(142.50,82.50)(142.50,71.50)(152.00,66.00)(161.50,71.50)
     (161.50,82.50)(152.00,88.00)(142.50,82.50)
\put(152.00,73.00){\makebox(0,0)[b]{1}}

\shade
\path(142.50,115.50)(142.50,104.50)(152.00,99.00)(161.50,104.50)
     (161.50,115.50)(152.00,121.00)(142.50,115.50)
\put(152.00,106.00){\makebox(0,0)[b]{$a$}}

\path(142.50,181.50)(142.50,170.50)(152.00,165.00)(161.50,170.50)
     (161.50,181.50)(152.00,187.00)(142.50,181.50)
\put(152.00,172.00){\makebox(0,0)[b]{1}}

\path(152.00,220.00)(142.50,214.50)(142.50,203.50)
     (152.00,198.00)(161.50,203.50)(161.50,214.50)
\put(152.00,209.00){\makebox(0,0){$\cdot$}}
\put(155.17,214.50){\makebox(0,0){$\cdot$}}
\put(158.33,220.00){\makebox(0,0){$\cdot$}}

\shade
\path(152.00,99.00)(152.00,88.00)(161.50,82.50)(171.00,88.00)
     (171.00,99.00)(161.50,104.50)(152.00,99.00)
\put(161.50,89.50){\makebox(0,0)[b]{$v'$}}

\path(152.00,132.00)(152.00,121.00)(161.50,115.50)(171.00,121.00)
     (171.00,132.00)(161.50,137.50)(152.00,132.00)
\put(161.50,122.50){\makebox(0,0)[b]{2}}

\path(152.00,198.00)(152.00,187.00)(161.50,181.50)(171.00,187.00)
     (171.00,198.00)(161.50,203.50)(152.00,198.00)
\put(161.50,188.50){\makebox(0,0)[b]{1}}

\path(161.50,82.50)(161.50,71.50)(171.00,66.00)(180.50,71.50)
     (180.50,82.50)(171.00,88.00)(161.50,82.50)
\put(171.00,73.00){\makebox(0,0)[b]{2}}

\path(161.50,148.50)(161.50,137.50)(171.00,132.00)(180.50,137.50)
     (180.50,148.50)(171.00,154.00)(161.50,148.50)
\put(171.00,139.00){\makebox(0,0)[b]{1}}

\path(171.00,220.00)(161.50,214.50)(161.50,203.50)
     (171.00,198.00)(180.50,203.50)(180.50,214.50)
\put(171.00,209.00){\makebox(0,0){$\cdot$}}
\put(174.17,214.50){\makebox(0,0){$\cdot$}}
\put(177.33,220.00){\makebox(0,0){$\cdot$}}

\path(171.00,66.00)(171.00,55.00)(180.50,49.50)(190.00,55.00)
     (190.00,66.00)(180.50,71.50)(171.00,66.00)
\put(180.50,56.50){\makebox(0,0)[b]{1}}

\path(180.50,148.50)(180.50,137.50)(190.00,132.00)(199.50,137.50)
     (199.50,148.50)(190.00,154.00)(180.50,148.50)
\put(190.00,139.00){\makebox(0,0)[b]{1}}

\path(190.00,66.00)(190.00,55.00)(199.50,49.50)(209.00,55.00)
     (209.00,66.00)(199.50,71.50)(190.00,66.00)
\put(199.50,56.50){\makebox(0,0)[b]{2}}

\path(190.00,132.00)(190.00,121.00)(199.50,115.50)(209.00,121.00)
     (209.00,132.00)(199.50,137.50)(190.00,132.00)
\put(199.50,122.50){\makebox(0,0)[b]{2}}

\path(199.50,49.50)(199.50,38.50)(209.00,33.00)(218.50,38.50)
     (218.50,49.50)(209.00,55.00)(199.50,49.50)
\put(209.00,40.00){\makebox(0,0)[b]{1}}

\shade
\path(199.50,82.50)(199.50,71.50)(209.00,66.00)(218.50,71.50)
     (218.50,82.50)(209.00,88.00)(199.50,82.50)
\put(209.00,73.00){\makebox(0,0)[b]{$c$}}

\path(209.00,33.00)(209.00,22.00)(218.50,16.50)(228.00,22.00)
     (228.00,33.00)(218.50,38.50)(209.00,33.00)
\put(218.50,23.50){\makebox(0,0)[b]{1}}

\shade
\path(209.00,66.00)(209.00,55.00)(218.50,49.50)(228.00,55.00)
     (228.00,66.00)(218.50,71.50)(209.00,66.00)
\put(218.50,56.50){\makebox(0,0)[b]{$c'$}}

\path(209.00,99.00)(209.00,88.00)(218.50,82.50)(228.00,88.00)
     (228.00,99.00)(218.50,104.50)(209.00,99.00)
\put(218.50,89.50){\makebox(0,0)[b]{2}}

\path(209.00,132.00)(209.00,121.00)(218.50,115.50)(228.00,121.00)
     (228.00,132.00)(218.50,137.50)(209.00,132.00)
\put(218.50,122.50){\makebox(0,0)[b]{1}}

\path(237.50,5.50)(237.50,16.50)(228.00,22.00)
     (218.50,16.50)(218.50,5.50)(228.00,0.00)
\put(228.00,11.00){\makebox(0,0){$\cdot$}}
\put(231.17,5.50){\makebox(0,0){$\cdot$}}
\put(234.33,0.00){\makebox(0,0){$\cdot$}}

\shade
\path(218.50,49.50)(218.50,38.50)(228.00,33.00)(237.50,38.50)
     (237.50,49.50)(228.00,55.00)(218.50,49.50)
\put(228.00,40.00){\makebox(0,0)[b]{$c$}}

\path(218.50,82.50)(218.50,71.50)(228.00,66.00)(237.50,71.50)
     (237.50,82.50)(228.00,88.00)(218.50,82.50)
\put(228.00,73.00){\makebox(0,0)[b]{1}}

\path(218.50,115.50)(218.50,104.50)(228.00,99.00)(237.50,104.50)
     (237.50,115.50)(228.00,121.00)(218.50,115.50)
\put(228.00,106.00){\makebox(0,0)[b]{1}}

\path(247.00,22.00)(247.00,33.00)(237.50,38.50)
     (228.00,33.00)(228.00,22.00)(237.50,16.50)
\put(237.50,27.50){\makebox(0,0){$\cdot$}}
\put(240.67,22.00){\makebox(0,0){$\cdot$}}
\put(243.83,16.50){\makebox(0,0){$\cdot$}}

\path(228.00,66.00)(228.00,55.00)(237.50,49.50)(247.00,55.00)
     (247.00,66.00)(237.50,71.50)(228.00,66.00)
\put(237.50,56.50){\makebox(0,0)[b]{1}}

\path(237.50,49.50)(237.50,38.50)(247.00,33.00)(256.50,38.50)
     (256.50,49.50)(247.00,55.00)(237.50,49.50)
\put(247.00,40.00){\makebox(0,0)[b]{1}}

\path(266.00,22.00)(266.00,33.00)(256.50,38.50)
     (247.00,33.00)(247.00,22.00)(256.50,16.50)
\put(256.50,27.50){\makebox(0,0){$\cdot$}}
\put(259.67,22.00){\makebox(0,0){$\cdot$}}
\put(262.83,16.50){\makebox(0,0){$\cdot$}}

\texture{c0c0c0c0 0 0 0 0c0c0c0c 0 0 0
          c0c0c0c0 0 0 0 0c0c0c0c 0 0 0
          c0c0c0c0 0 0 0 0c0c0c0c 0 0 0
          c0c0c0c0 0 0 0 0c0c0c0c 0 0 0}

\shade
\path(57.00,99.00)(57.00,88.00)(66.50,82.50)(76.00,88.00)
     (76.00,99.00)(66.50,104.50)(57.00,99.00)
\put(66.50,89.50){\makebox(0,0)[b]{$4$}}

\shade
\path(180.50,115.50)(180.50,104.50)(190.00,99.00)(199.50,104.50)
     (199.50,115.50)(190.00,121.00)(180.50,115.50)
\put(190.00,106.00){\makebox(0,0)[b]{$4$}}

\texture{cccccccc 0 0 0 cccccccc 0 0 0
         cccccccc 0 0 0 cccccccc 0 0 0
         cccccccc 0 0 0 cccccccc 0 0 0
         cccccccc 0 0 0 cccccccc 0 0 0}
\whiten

\path(38.00,99.00)(38.00,88.00)(47.50,82.50)(57.00,88.00)
     (57.00,99.00)(47.50,104.50)(38.00,99.00)
\put(47.50,93.50){\makebox(0,0){\circle*{9.50}}}

\path(57.00,132.00)(57.00,121.00)(66.50,115.50)(76.00,121.00)
     (76.00,132.00)(66.50,137.50)(57.00,132.00)
\put(66.50,126.50){\makebox(0,0){\circle*{9.50}}}

\path(66.50,82.50)(66.50,71.50)(76.00,66.00)(85.50,71.50)
     (85.50,82.50)(76.00,88.00)(66.50,82.50)
\put(76.00,77.00){\makebox(0,0){\circle*{9.50}}}

\path(95.00,165.00)(95.00,154.00)(104.50,148.50)(114.00,154.00)
     (114.00,165.00)(104.50,170.50)(95.00,165.00)
\put(104.50,159.50){\makebox(0,0){\circle*{9.50}}}

\path(104.50,115.50)(104.50,104.50)(114.00,99.00)(123.50,104.50)
     (123.50,115.50)(114.00,121.00)(104.50,115.50)
\put(114.00,110.00){\makebox(0,0){\circle*{9.50}}}

\path(123.50,115.50)(123.50,104.50)(133.00,99.00)(142.50,104.50)
     (142.50,115.50)(133.00,121.00)(123.50,115.50)
\put(133.00,110.00){\makebox(0,0){\circle*{9.50}}}

\path(171.00,132.00)(171.00,121.00)(180.50,115.50)(190.00,121.00)
     (190.00,132.00)(180.50,137.50)(171.00,132.00)
\put(180.50,126.50){\makebox(0,0){\circle*{9.50}}}

\path(180.50,82.50)(180.50,71.50)(190.00,66.00)(199.50,71.50)
     (199.50,82.50)(190.00,88.00)(180.50,82.50)
\put(190.00,77.00){\makebox(0,0){\circle*{9.50}}}

\path(199.50,115.50)(199.50,104.50)(209.00,99.00)(218.50,104.50)
     (218.50,115.50)(209.00,121.00)(199.50,115.50)
\put(209.00,110.00){\makebox(0,0){\circle*{9.50}}}

\texture{40004000 0 0 0 00400040 0 0 0
          40004000 0 0 0 00400040 0 0 0
          40004000 0 0 0 00400040 0 0 0
          40004000 0 0 0 00400040 0 0 0}

\Thicklines
\shade
\path(47.50,115.50)(47.50,104.50)(57.00,99.00)(66.50,104.50)
     (66.50,115.50)(57.00,121.00)(47.50,115.50)
\put(57.00,106.00){\makebox(0,0)[b]{$x$}}
\thinlines

\Thicklines
\shade
\path(66.50,115.50)(66.50,104.50)(76.00,99.00)(85.50,104.50)
     (85.50,115.50)(76.00,121.00)(66.50,115.50)
\put(76.00,106.00){\makebox(0,0)[b]{$a'$}}
\thinlines

\Thicklines
\shade
\path(95.00,99.00)(95.00,88.00)(104.50,82.50)(114.00,88.00)
     (114.00,99.00)(104.50,104.50)(95.00,99.00)
\put(104.50,89.50){\makebox(0,0)[b]{$v$}}
\thinlines

\Thicklines
\shade
\path(133.00,99.00)(133.00,88.00)(142.50,82.50)(152.00,88.00)
     (152.00,99.00)(142.50,104.50)(133.00,99.00)
\put(142.50,89.50){\makebox(0,0)[b]{$v$}}
\thinlines

\Thicklines
\shade
\path(161.50,115.50)(161.50,104.50)(171.00,99.00)(180.50,104.50)
     (180.50,115.50)(171.00,121.00)(161.50,115.50)
\put(171.00,106.00){\makebox(0,0)[b]{$a'$}}
\thinlines

\Thicklines
\shade
\path(171.00,99.00)(171.00,88.00)(180.50,82.50)(190.00,88.00)
     (190.00,99.00)(180.50,104.50)(171.00,99.00)
\put(180.50,89.50){\makebox(0,0)[b]{$v$}}
\thinlines

\Thicklines
\shade
\path(190.00,99.00)(190.00,88.00)(199.50,82.50)(209.00,88.00)
     (209.00,99.00)(199.50,104.50)(190.00,99.00)
\put(199.50,89.50){\makebox(0,0)[b]{$c'$}}
\thinlines

\put(-1.62,43.5){\makebox(0,0)[rt]{$z$}}
\put(-2.62,45.5){\vector(1,2){13}}
\put(-1.62,143.5){\makebox(0,0)[rb]{$x$}}
\put(-2.62,141.5){\vector(1,-2){13}}
\put(169.86,164.75){\makebox(0,0)[rt]{$a$}}
\put(168.86,166.75){\vector(1,2){13}}
\put(245.86,88.75){\makebox(0,0)[rb]{$c$}}
\put(244.86,86.75){\vector(1,-2){13}}
\mathversion{normal}

\end{picture}
\end{center} \vspace*{-20pt}
\caption{A component that outputs $a = x \cdot z$ and $c = x + z$} 
\label{addmul}
\end{figure}

Since the adder-multiplier combi is the most complicated component by far,
it needs some explanation. Let us concentrate on the one for normal 
minesweeper first. The part between the dashed lines is optional and its
benefit will be discussed later. 

The dark shaded compartment on the left hand side with the number $4$
enforces the equation $x + z + a' + v' = 2$. Since also
$a + v + a' + v' = 2$, the sets $\{x,z\}$ and $\{a,v\}$ are 
enforced to be equal. So the dark shaded compartment on the left hand side is in 
fact the heart of a {\em shaker}\/: its outputs $a$ and $v$ are a nondeterministic
permutation of $x$ and $z$.

We show that $a = \min(x,z)$ and $v = \max(x,z)$. Suppose that this is not 
the case. Then $a = 1$ and $v = 0$. So $a' = v = 0$. It follows that
the dark shaded compartment on the right hand side with the number $4$ 
is surrounded by three mines at most. This contradicts the number $4$ in it, 
so $a = \min(x,z)$ and $v = \max(x,z)$. Thus the dark shaded compartment 
on the right hand side is the heart of a {\em tester-adder combi}\/: its 
inputs are tested and at least one of them must be equal to one.

Since $x \cdot z = \min(x,z)$, the output $a$ equals $x \cdot z$.
Next, the heart on the right hand side enforces
the equation $a' + v + c' = 2$. Since $v = 0$ implies $a' = 1$ and
$a' = 0$ implies $v = 1$, $1 \le a' + v \le 2$. Thus $c'$
can be chosen such that $a' + v + c' = 2$.

Modulo $2$, the heart on the right hand side gives the following information:
$$
0 \equiv a' + v + c' \equiv a + v + c \equiv x + z + c \pmod 2
$$
It follows that $c \equiv x + z$ modulo $2$. 

But we are not yet done now with the adder-multiplier combi for normal 
minesweeper. This is because for the square compartments which are bold,
revealing these compartments might give new information about the board
at first glance. But that only seems so. If you e.g.\@ reveal the bold
compartment with $v$, then $v = 0$ and you already know before revealing 
it that $a' = c' = 1$.

The adder-multiplier combi for hexagonal minesweeper works essentially
the same as that for normal minesweeper. Notice that the heart on the left
hand side is the heart of a somewhat buggy shaker this time, since you can 
reveal $x$ when $a = 1$ by way of the bold compartment with $a'$. But since 
$a = 1$ implies $x = 1$, the shaker subcomponent is correct within its 
context.

$v$ is not an output of the adder-multiplier combi, but one can reconstruct
$v$ by adding $a$ and $c$ with another adder-multiplier combi or just only
a modified version of the tester-adder combi. One can make a tester-adder 
combi without `fragile' literal compartments for both normal and hexagonal 
Minesweeper.

\section{Triangular Minesweeper}

\begin{figure}[!htp]
\begin{center}
\begin{picture}(158.33,157.67)(-3.17,0.00)

\filltype{shade}
\texture{40004000 0 0 0 00400040 0 0 0
         40004000 0 0 0 00400040 0 0 0
         40004000 0 0 0 00400040 0 0 0
         40004000 0 0 0 00400040 0 0 0}
\whiten

\path(19.00,82.50)(19.00,71.50)(28.50,66.00)(38.00,71.50)
     (38.00,82.50)(28.50,88.00)(19.00,82.50)
\put(28.50,73.00){\makebox(0,0)[b]{2}}

\put(9.50,77.00){\makebox(0,0){$\cdot$}}
\put(3.17,77.00){\makebox(0,0){$\cdot$}}
\put(-3.17,77.00){\makebox(0,0){$\cdot$}}

\put(19.00,60.50){\makebox(0,0){$\cdot$}}
\put(12.67,56.83){\makebox(0,0){$\cdot$}}
\put(6.33,53.17){\makebox(0,0){$\cdot$}}
\path(19.00,49.50)(28.50,55.00)
\path(19.00,71.50)(9.50,66.00)

\put(19.00,93.50){\makebox(0,0){$\cdot$}}
\put(12.67,97.17){\makebox(0,0){$\cdot$}}
\put(6.33,100.83){\makebox(0,0){$\cdot$}}
\path(19.00,82.50)(9.50,88.00)
\path(19.00,104.50)(28.50,99.00)

\put(28.50,44.00){\makebox(0,0){$\cdot$}}
\put(22.17,40.33){\makebox(0,0){$\cdot$}}
\put(15.83,36.67){\makebox(0,0){$\cdot$}}
\path(28.50,33.00)(38.00,38.50)
\path(28.50,55.00)(19.00,49.50)

\put(28.50,110.00){\makebox(0,0){$\cdot$}}
\put(22.17,113.67){\makebox(0,0){$\cdot$}}
\put(15.83,117.33){\makebox(0,0){$\cdot$}}
\path(28.50,99.00)(19.00,104.50)
\path(28.50,121.00)(38.00,115.50)

\put(38.00,27.50){\makebox(0,0){$\cdot$}}
\put(31.67,23.83){\makebox(0,0){$\cdot$}}
\put(25.33,20.17){\makebox(0,0){$\cdot$}}
\path(38.00,16.50)(47.50,22.00)
\path(38.00,38.50)(28.50,33.00)

\path(28.50,66.00)(28.50,55.00)(38.00,49.50)(47.50,55.00)
     (47.50,66.00)(38.00,71.50)(28.50,66.00)
\put(38.00,56.50){\makebox(0,0)[b]{1}}

\path(28.50,99.00)(28.50,88.00)(38.00,82.50)(47.50,88.00)
     (47.50,99.00)(38.00,104.50)(28.50,99.00)
\put(38.00,89.50){\makebox(0,0)[b]{1}}

\put(38.00,126.50){\makebox(0,0){$\cdot$}}
\put(31.67,130.17){\makebox(0,0){$\cdot$}}
\put(25.33,133.83){\makebox(0,0){$\cdot$}}
\path(38.00,115.50)(28.50,121.00)
\path(38.00,137.50)(47.50,132.00)

\put(47.50,11.00){\makebox(0,0){$\cdot$}}
\put(44.33,5.50){\makebox(0,0){$\cdot$}}
\put(41.17,0.00){\makebox(0,0){$\cdot$}}

\path(38.00,49.50)(38.00,38.50)(47.50,33.00)(57.00,38.50)
     (57.00,49.50)(47.50,55.00)(38.00,49.50)
\put(47.50,40.00){\makebox(0,0)[b]{ }}

\path(38.00,115.50)(38.00,104.50)(47.50,99.00)(57.00,104.50)
     (57.00,115.50)(47.50,121.00)(38.00,115.50)
\put(47.50,106.00){\makebox(0,0)[b]{ }}

\put(47.50,143.00){\makebox(0,0){$\cdot$}}
\put(44.33,148.50){\makebox(0,0){$\cdot$}}
\put(41.17,154.00){\makebox(0,0){$\cdot$}}

\path(47.50,33.00)(47.50,22.00)(57.00,16.50)(66.50,22.00)
     (66.50,33.00)(57.00,38.50)(47.50,33.00)
\put(57.00,23.50){\makebox(0,0)[b]{ }}

\path(47.50,66.00)(47.50,55.00)(57.00,49.50)(66.50,55.00)
     (66.50,66.00)(57.00,71.50)(47.50,66.00)
\put(57.00,56.50){\makebox(0,0)[b]{1}}

\path(47.50,99.00)(47.50,88.00)(57.00,82.50)(66.50,88.00)
     (66.50,99.00)(57.00,104.50)(47.50,99.00)
\put(57.00,89.50){\makebox(0,0)[b]{1}}

\path(47.50,132.00)(47.50,121.00)(57.00,115.50)(66.50,121.00)
     (66.50,132.00)(57.00,137.50)(47.50,132.00)
\put(57.00,122.50){\makebox(0,0)[b]{ }}

\put(66.50,11.00){\makebox(0,0){$\cdot$}}
\put(66.50,3.67){\makebox(0,0){$\cdot$}}
\put(66.50,-3.67){\makebox(0,0){$\cdot$}}
\path(57.00,5.50)(57.00,16.50)
\path(76.00,5.50)(76.00,16.50)

\path(57.00,49.50)(57.00,38.50)(66.50,33.00)(76.00,38.50)
     (76.00,49.50)(66.50,55.00)(57.00,49.50)
\put(66.50,40.00){\makebox(0,0)[b]{ }}

\path(57.00,82.50)(57.00,71.50)(66.50,66.00)(76.00,71.50)
     (76.00,82.50)(66.50,88.00)(57.00,82.50)
\put(66.50,73.00){\makebox(0,0)[b]{2}}

\path(57.00,115.50)(57.00,104.50)(66.50,99.00)(76.00,104.50)
     (76.00,115.50)(66.50,121.00)(57.00,115.50)
\put(66.50,106.00){\makebox(0,0)[b]{ }}

\put(66.50,143.00){\makebox(0,0){$\cdot$}}
\put(66.50,150.33){\makebox(0,0){$\cdot$}}
\put(66.50,157.67){\makebox(0,0){$\cdot$}}
\path(57.00,137.50)(57.00,148.50)
\path(76.00,137.50)(76.00,148.50)

\path(66.50,33.00)(66.50,22.00)(76.00,16.50)(85.50,22.00)
     (85.50,33.00)(76.00,38.50)(66.50,33.00)
\put(76.00,23.50){\makebox(0,0)[b]{ }}

\path(66.50,66.00)(66.50,55.00)(76.00,49.50)(85.50,55.00)
     (85.50,66.00)(76.00,71.50)(66.50,66.00)
\put(76.00,56.50){\makebox(0,0)[b]{1}}

\path(66.50,99.00)(66.50,88.00)(76.00,82.50)(85.50,88.00)
     (85.50,99.00)(76.00,104.50)(66.50,99.00)
\put(76.00,89.50){\makebox(0,0)[b]{1}}

\path(66.50,132.00)(66.50,121.00)(76.00,115.50)(85.50,121.00)
     (85.50,132.00)(76.00,137.50)(66.50,132.00)
\put(76.00,122.50){\makebox(0,0)[b]{ }}

\put(85.50,11.00){\makebox(0,0){$\cdot$}}
\put(85.50,3.67){\makebox(0,0){$\cdot$}}
\put(85.50,-3.67){\makebox(0,0){$\cdot$}}
\path(76.00,5.50)(76.00,16.50)
\path(95.00,5.50)(95.00,16.50)

\path(76.00,49.50)(76.00,38.50)(85.50,33.00)(95.00,38.50)
     (95.00,49.50)(85.50,55.00)(76.00,49.50)
\put(85.50,40.00){\makebox(0,0)[b]{ }}

\path(76.00,115.50)(76.00,104.50)(85.50,99.00)(95.00,104.50)
     (95.00,115.50)(85.50,121.00)(76.00,115.50)
\put(85.50,106.00){\makebox(0,0)[b]{ }}

\put(85.50,143.00){\makebox(0,0){$\cdot$}}
\put(85.50,150.33){\makebox(0,0){$\cdot$}}
\put(85.50,157.67){\makebox(0,0){$\cdot$}}
\path(76.00,137.50)(76.00,148.50)
\path(95.00,137.50)(95.00,148.50)

\path(85.50,33.00)(85.50,22.00)(95.00,16.50)(104.50,22.00)
     (104.50,33.00)(95.00,38.50)(85.50,33.00)
\put(95.00,23.50){\makebox(0,0)[b]{ }}

\path(85.50,66.00)(85.50,55.00)(95.00,49.50)(104.50,55.00)
     (104.50,66.00)(95.00,71.50)(85.50,66.00)
\put(95.00,56.50){\makebox(0,0)[b]{1}}

\path(85.50,99.00)(85.50,88.00)(95.00,82.50)(104.50,88.00)
     (104.50,99.00)(95.00,104.50)(85.50,99.00)
\put(95.00,89.50){\makebox(0,0)[b]{1}}

\path(85.50,132.00)(85.50,121.00)(95.00,115.50)(104.50,121.00)
     (104.50,132.00)(95.00,137.50)(85.50,132.00)
\put(95.00,122.50){\makebox(0,0)[b]{ }}

\put(104.50,11.00){\makebox(0,0){$\cdot$}}
\put(107.67,5.50){\makebox(0,0){$\cdot$}}
\put(110.83,0.00){\makebox(0,0){$\cdot$}}

\path(95.00,49.50)(95.00,38.50)(104.50,33.00)(114.00,38.50)
     (114.00,49.50)(104.50,55.00)(95.00,49.50)
\put(104.50,40.00){\makebox(0,0)[b]{ }}

\path(95.00,82.50)(95.00,71.50)(104.50,66.00)(114.00,71.50)
     (114.00,82.50)(104.50,88.00)(95.00,82.50)
\put(104.50,73.00){\makebox(0,0)[b]{2}}

\path(95.00,115.50)(95.00,104.50)(104.50,99.00)(114.00,104.50)
     (114.00,115.50)(104.50,121.00)(95.00,115.50)
\put(104.50,106.00){\makebox(0,0)[b]{ }}

\put(104.50,143.00){\makebox(0,0){$\cdot$}}
\put(107.67,148.50){\makebox(0,0){$\cdot$}}
\put(110.83,154.00){\makebox(0,0){$\cdot$}}

\put(114.00,27.50){\makebox(0,0){$\cdot$}}
\put(120.33,23.83){\makebox(0,0){$\cdot$}}
\put(126.67,20.17){\makebox(0,0){$\cdot$}}
\path(114.00,16.50)(104.50,22.00)
\path(114.00,38.50)(123.50,33.00)

\path(104.50,66.00)(104.50,55.00)(114.00,49.50)(123.50,55.00)
     (123.50,66.00)(114.00,71.50)(104.50,66.00)
\put(114.00,56.50){\makebox(0,0)[b]{1}}

\path(104.50,99.00)(104.50,88.00)(114.00,82.50)(123.50,88.00)
     (123.50,99.00)(114.00,104.50)(104.50,99.00)
\put(114.00,89.50){\makebox(0,0)[b]{1}}

\put(114.00,126.50){\makebox(0,0){$\cdot$}}
\put(120.33,130.17){\makebox(0,0){$\cdot$}}
\put(126.67,133.83){\makebox(0,0){$\cdot$}}
\path(114.00,115.50)(123.50,121.00)
\path(114.00,137.50)(104.50,132.00)

\put(123.50,44.00){\makebox(0,0){$\cdot$}}
\put(129.83,40.33){\makebox(0,0){$\cdot$}}
\put(136.17,36.67){\makebox(0,0){$\cdot$}}
\path(123.50,33.00)(114.00,38.50)
\path(123.50,55.00)(133.00,49.50)

\put(123.50,110.00){\makebox(0,0){$\cdot$}}
\put(129.83,113.67){\makebox(0,0){$\cdot$}}
\put(136.17,117.33){\makebox(0,0){$\cdot$}}
\path(123.50,99.00)(133.00,104.50)
\path(123.50,121.00)(114.00,115.50)

\put(133.00,60.50){\makebox(0,0){$\cdot$}}
\put(139.33,56.83){\makebox(0,0){$\cdot$}}
\put(145.67,53.17){\makebox(0,0){$\cdot$}}
\path(133.00,49.50)(123.50,55.00)
\path(133.00,71.50)(142.50,66.00)

\put(133.00,93.50){\makebox(0,0){$\cdot$}}
\put(139.33,97.17){\makebox(0,0){$\cdot$}}
\put(145.67,100.83){\makebox(0,0){$\cdot$}}
\path(133.00,82.50)(142.50,88.00)
\path(133.00,104.50)(123.50,99.00)

\put(142.50,77.00){\makebox(0,0){$\cdot$}}
\put(148.83,77.00){\makebox(0,0){$\cdot$}}
\put(155.17,77.00){\makebox(0,0){$\cdot$}}

\texture{cccccccc 0 0 0 cccccccc 0 0 0
         cccccccc 0 0 0 cccccccc 0 0 0
         cccccccc 0 0 0 cccccccc 0 0 0
         cccccccc 0 0 0 cccccccc 0 0 0}
\whiten

\path(38.00,82.50)(38.00,71.50)(47.50,66.00)(57.00,71.50)
     (57.00,82.50)(47.50,88.00)(38.00,82.50)
\put(47.50,77.00){\makebox(0,0){\circle*{9.50}}}

\path(76.00,82.50)(76.00,71.50)(85.50,66.00)(95.00,71.50)
     (95.00,82.50)(85.50,88.00)(76.00,82.50)
\put(85.50,77.00){\makebox(0,0){\circle*{9.50}}}

\path(114.00,82.50)(114.00,71.50)(123.50,66.00)(133.00,71.50)
     (133.00,82.50)(123.50,88.00)(114.00,82.50)
\put(123.50,77.00){\makebox(0,0){\circle*{9.50}}}

\end{picture}
\quad
\nolinebreak ~ \nolinebreak ~ \nolinebreak
\begin{picture}(158.33,157.67)(-3.17,0.00)

\filltype{shade}
\texture{c0c0c0c0 0 0 0 0c0c0c0c 0 0 0
         c0c0c0c0 0 0 0 0c0c0c0c 0 0 0
         c0c0c0c0 0 0 0 0c0c0c0c 0 0 0
         c0c0c0c0 0 0 0 0c0c0c0c 0 0 0}
\whiten

\path(19.00,77.00)(28.50,93.50)(38.00,77.00)(19.00,77.00)
\path(19.00,77.00)(28.50,60.50)(38.00,77.00)(19.00,77.00)
\put(28.50,79.70){\makebox(0,0)[b]{\footnotesize 8}}
\put(28.50,71.50){\circle*{6.65}}

\put(9.50,77.00){\makebox(0,0){$\cdot$}}
\put(3.17,77.00){\makebox(0,0){$\cdot$}}
\put(-3.17,77.00){\makebox(0,0){$\cdot$}}

\put(19.00,60.50){\makebox(0,0){$\cdot$}}
\put(12.67,56.83){\makebox(0,0){$\cdot$}}
\put(6.33,53.17){\makebox(0,0){$\cdot$}}

\put(19.00,93.50){\makebox(0,0){$\cdot$}}
\put(12.67,97.17){\makebox(0,0){$\cdot$}}
\put(6.33,100.83){\makebox(0,0){$\cdot$}}

\put(28.50,44.00){\makebox(0,0){$\cdot$}}
\put(22.17,40.33){\makebox(0,0){$\cdot$}}
\put(15.83,36.67){\makebox(0,0){$\cdot$}}

\put(28.50,110.00){\makebox(0,0){$\cdot$}}
\put(22.17,113.67){\makebox(0,0){$\cdot$}}
\put(15.83,117.33){\makebox(0,0){$\cdot$}}

\put(38.00,27.50){\makebox(0,0){$\cdot$}}
\put(31.67,23.83){\makebox(0,0){$\cdot$}}
\put(25.33,20.17){\makebox(0,0){$\cdot$}}

\path(28.50,60.50)(38.00,77.00)(47.50,60.50)(28.50,60.50)
\path(28.50,60.50)(38.00,44.00)(47.50,60.50)(28.50,60.50)
\put(38.00,63.20){\makebox(0,0)[b]{\footnotesize 7}}
\put(38.00,55.00){\circle*{6.65}}

\path(28.50,93.50)(38.00,110.00)(47.50,93.50)(28.50,93.50)
\path(28.50,93.50)(38.00,77.00)(47.50,93.50)(28.50,93.50)
\put(38.00,96.20){\makebox(0,0)[b]{\footnotesize 7}}
\put(38.00,88.00){\circle*{6.65}}

\put(38.00,126.50){\makebox(0,0){$\cdot$}}
\put(31.67,130.17){\makebox(0,0){$\cdot$}}
\put(25.33,133.83){\makebox(0,0){$\cdot$}}

\put(47.50,11.00){\makebox(0,0){$\cdot$}}
\put(44.33,5.50){\makebox(0,0){$\cdot$}}
\put(41.17,0.00){\makebox(0,0){$\cdot$}}

\path(38.00,44.00)(47.50,60.50)(57.00,44.00)(38.00,44.00)
\path(38.00,44.00)(47.50,27.50)(57.00,44.00)(38.00,44.00)
\put(47.50,46.70){\makebox(0,0)[b]{\footnotesize 6}}
\put(47.50,38.50){\circle*{6.65}}

\path(38.00,110.00)(47.50,126.50)(57.00,110.00)(38.00,110.00)
\path(38.00,110.00)(47.50,93.50)(57.00,110.00)(38.00,110.00)
\put(47.50,112.70){\makebox(0,0)[b]{\footnotesize 6}}
\put(47.50,104.50){\circle*{6.65}}

\put(47.50,143.00){\makebox(0,0){$\cdot$}}
\put(44.33,148.50){\makebox(0,0){$\cdot$}}
\put(41.17,154.00){\makebox(0,0){$\cdot$}}

\path(47.50,27.50)(57.00,44.00)(66.50,27.50)(47.50,27.50)
\path(47.50,27.50)(57.00,11.00)(66.50,27.50)(47.50,27.50)
\put(57.00,30.20){\makebox(0,0)[b]{\footnotesize 6}}
\put(57.00,22.00){\circle*{6.65}}

\path(47.50,60.50)(57.00,77.00)(66.50,60.50)(47.50,60.50)
\path(47.50,60.50)(57.00,44.00)(66.50,60.50)(47.50,60.50)
\put(57.00,63.20){\makebox(0,0)[b]{\footnotesize 7}}
\put(57.00,55.00){\circle*{6.65}}

\path(47.50,93.50)(57.00,110.00)(66.50,93.50)(47.50,93.50)
\path(47.50,93.50)(57.00,77.00)(66.50,93.50)(47.50,93.50)
\put(57.00,96.20){\makebox(0,0)[b]{\footnotesize 7}}
\put(57.00,88.00){\circle*{6.65}}

\path(47.50,126.50)(57.00,143.00)(66.50,126.50)(47.50,126.50)
\path(47.50,126.50)(57.00,110.00)(66.50,126.50)(47.50,126.50)
\put(57.00,129.20){\makebox(0,0)[b]{\footnotesize 6}}
\put(57.00,121.00){\circle*{6.65}}

\put(66.50,11.00){\makebox(0,0){$\cdot$}}
\put(66.50,3.67){\makebox(0,0){$\cdot$}}
\put(66.50,-3.67){\makebox(0,0){$\cdot$}}

\path(57.00,44.00)(66.50,60.50)(76.00,44.00)(57.00,44.00)
\path(57.00,44.00)(66.50,27.50)(76.00,44.00)(57.00,44.00)
\put(66.50,46.70){\makebox(0,0)[b]{\footnotesize 6}}
\put(66.50,38.50){\circle*{6.65}}

\path(57.00,77.00)(66.50,93.50)(76.00,77.00)(57.00,77.00)
\path(57.00,77.00)(66.50,60.50)(76.00,77.00)(57.00,77.00)
\put(66.50,79.70){\makebox(0,0)[b]{\footnotesize 8}}
\put(66.50,71.50){\circle*{6.65}}

\path(57.00,110.00)(66.50,126.50)(76.00,110.00)(57.00,110.00)
\path(57.00,110.00)(66.50,93.50)(76.00,110.00)(57.00,110.00)
\put(66.50,112.70){\makebox(0,0)[b]{\footnotesize 6}}
\put(66.50,104.50){\circle*{6.65}}

\put(66.50,143.00){\makebox(0,0){$\cdot$}}
\put(66.50,150.33){\makebox(0,0){$\cdot$}}
\put(66.50,157.67){\makebox(0,0){$\cdot$}}

\path(66.50,27.50)(76.00,44.00)(85.50,27.50)(66.50,27.50)
\path(66.50,27.50)(76.00,11.00)(85.50,27.50)(66.50,27.50)
\put(76.00,30.20){\makebox(0,0)[b]{\footnotesize 6}}
\put(76.00,22.00){\circle*{6.65}}

\path(66.50,60.50)(76.00,77.00)(85.50,60.50)(66.50,60.50)
\path(66.50,60.50)(76.00,44.00)(85.50,60.50)(66.50,60.50)
\put(76.00,63.20){\makebox(0,0)[b]{\footnotesize 7}}
\put(76.00,55.00){\circle*{6.65}}

\path(66.50,93.50)(76.00,110.00)(85.50,93.50)(66.50,93.50)
\path(66.50,93.50)(76.00,77.00)(85.50,93.50)(66.50,93.50)
\put(76.00,96.20){\makebox(0,0)[b]{\footnotesize 7}}
\put(76.00,88.00){\circle*{6.65}}

\path(66.50,126.50)(76.00,143.00)(85.50,126.50)(66.50,126.50)
\path(66.50,126.50)(76.00,110.00)(85.50,126.50)(66.50,126.50)
\put(76.00,129.20){\makebox(0,0)[b]{\footnotesize 6}}
\put(76.00,121.00){\circle*{6.65}}

\put(85.50,11.00){\makebox(0,0){$\cdot$}}
\put(85.50,3.67){\makebox(0,0){$\cdot$}}
\put(85.50,-3.67){\makebox(0,0){$\cdot$}}

\path(76.00,44.00)(85.50,60.50)(95.00,44.00)(76.00,44.00)
\path(76.00,44.00)(85.50,27.50)(95.00,44.00)(76.00,44.00)
\put(85.50,46.70){\makebox(0,0)[b]{\footnotesize 6}}
\put(85.50,38.50){\circle*{6.65}}

\path(76.00,110.00)(85.50,126.50)(95.00,110.00)(76.00,110.00)
\path(76.00,110.00)(85.50,93.50)(95.00,110.00)(76.00,110.00)
\put(85.50,112.70){\makebox(0,0)[b]{\footnotesize 6}}
\put(85.50,104.50){\circle*{6.65}}

\put(85.50,143.00){\makebox(0,0){$\cdot$}}
\put(85.50,150.33){\makebox(0,0){$\cdot$}}
\put(85.50,157.67){\makebox(0,0){$\cdot$}}

\path(85.50,27.50)(95.00,44.00)(104.50,27.50)(85.50,27.50)
\path(85.50,27.50)(95.00,11.00)(104.50,27.50)(85.50,27.50)
\put(95.00,30.20){\makebox(0,0)[b]{\footnotesize 6}}
\put(95.00,22.00){\circle*{6.65}}

\path(85.50,60.50)(95.00,77.00)(104.50,60.50)(85.50,60.50)
\path(85.50,60.50)(95.00,44.00)(104.50,60.50)(85.50,60.50)
\put(95.00,63.20){\makebox(0,0)[b]{\footnotesize 7}}
\put(95.00,55.00){\circle*{6.65}}

\path(85.50,93.50)(95.00,110.00)(104.50,93.50)(85.50,93.50)
\path(85.50,93.50)(95.00,77.00)(104.50,93.50)(85.50,93.50)
\put(95.00,96.20){\makebox(0,0)[b]{\footnotesize 7}}
\put(95.00,88.00){\circle*{6.65}}

\path(85.50,126.50)(95.00,143.00)(104.50,126.50)(85.50,126.50)
\path(85.50,126.50)(95.00,110.00)(104.50,126.50)(85.50,126.50)
\put(95.00,129.20){\makebox(0,0)[b]{\footnotesize 6}}
\put(95.00,121.00){\circle*{6.65}}

\put(104.50,11.00){\makebox(0,0){$\cdot$}}
\put(107.67,5.50){\makebox(0,0){$\cdot$}}
\put(110.83,0.00){\makebox(0,0){$\cdot$}}

\path(95.00,44.00)(104.50,60.50)(114.00,44.00)(95.00,44.00)
\path(95.00,44.00)(104.50,27.50)(114.00,44.00)(95.00,44.00)
\put(104.50,46.70){\makebox(0,0)[b]{\footnotesize 6}}
\put(104.50,38.50){\circle*{6.65}}

\path(95.00,77.00)(104.50,93.50)(114.00,77.00)(95.00,77.00)
\path(95.00,77.00)(104.50,60.50)(114.00,77.00)(95.00,77.00)
\put(104.50,79.70){\makebox(0,0)[b]{\footnotesize 8}}
\put(104.50,71.50){\circle*{6.65}}

\path(95.00,110.00)(104.50,126.50)(114.00,110.00)(95.00,110.00)
\path(95.00,110.00)(104.50,93.50)(114.00,110.00)(95.00,110.00)
\put(104.50,112.70){\makebox(0,0)[b]{\footnotesize 6}}
\put(104.50,104.50){\circle*{6.65}}

\put(104.50,143.00){\makebox(0,0){$\cdot$}}
\put(107.67,148.50){\makebox(0,0){$\cdot$}}
\put(110.83,154.00){\makebox(0,0){$\cdot$}}

\put(114.00,27.50){\makebox(0,0){$\cdot$}}
\put(120.33,23.83){\makebox(0,0){$\cdot$}}
\put(126.67,20.17){\makebox(0,0){$\cdot$}}

\path(104.50,60.50)(114.00,77.00)(123.50,60.50)(104.50,60.50)
\path(104.50,60.50)(114.00,44.00)(123.50,60.50)(104.50,60.50)
\put(114.00,63.20){\makebox(0,0)[b]{\footnotesize 7}}
\put(114.00,55.00){\circle*{6.65}}

\path(104.50,93.50)(114.00,110.00)(123.50,93.50)(104.50,93.50)
\path(104.50,93.50)(114.00,77.00)(123.50,93.50)(104.50,93.50)
\put(114.00,96.20){\makebox(0,0)[b]{\footnotesize 7}}
\put(114.00,88.00){\circle*{6.65}}

\put(114.00,126.50){\makebox(0,0){$\cdot$}}
\put(120.33,130.17){\makebox(0,0){$\cdot$}}
\put(126.67,133.83){\makebox(0,0){$\cdot$}}

\put(123.50,44.00){\makebox(0,0){$\cdot$}}
\put(129.83,40.33){\makebox(0,0){$\cdot$}}
\put(136.17,36.67){\makebox(0,0){$\cdot$}}

\put(123.50,110.00){\makebox(0,0){$\cdot$}}
\put(129.83,113.67){\makebox(0,0){$\cdot$}}
\put(136.17,117.33){\makebox(0,0){$\cdot$}}

\put(133.00,60.50){\makebox(0,0){$\cdot$}}
\put(139.33,56.83){\makebox(0,0){$\cdot$}}
\put(145.67,53.17){\makebox(0,0){$\cdot$}}

\put(133.00,93.50){\makebox(0,0){$\cdot$}}
\put(139.33,97.17){\makebox(0,0){$\cdot$}}
\put(145.67,100.83){\makebox(0,0){$\cdot$}}

\put(142.50,77.00){\makebox(0,0){$\cdot$}}
\put(148.83,77.00){\makebox(0,0){$\cdot$}}
\put(155.17,77.00){\makebox(0,0){$\cdot$}}

\path(38.00,77.00)(47.50,93.50)(57.00,77.00)(38.00,77.00)
\path(38.00,77.00)(47.50,60.50)(57.00,77.00)(38.00,77.00)
\put(47.50,79.70){\makebox(0,0)[b]{\footnotesize  }}
\put(47.50,71.50){\circle*{6.65}}

\path(76.00,77.00)(85.50,93.50)(95.00,77.00)(76.00,77.00)
\path(76.00,77.00)(85.50,60.50)(95.00,77.00)(76.00,77.00)
\put(85.50,79.70){\makebox(0,0)[b]{\footnotesize  }}
\put(85.50,71.50){\circle*{6.65}}

\path(114.00,77.00)(123.50,93.50)(133.00,77.00)(114.00,77.00)
\path(114.00,77.00)(123.50,60.50)(133.00,77.00)(114.00,77.00)
\put(123.50,79.70){\makebox(0,0)[b]{\footnotesize  }}
\put(123.50,71.50){\circle*{6.65}}

\texture{cccccccc 0 0 0 cccccccc 0 0 0
         cccccccc 0 0 0 cccccccc 0 0 0
         cccccccc 0 0 0 cccccccc 0 0 0
         cccccccc 0 0 0 cccccccc 0 0 0}
\whiten

\put(47.50,82.50){\makebox(0,0){\circle*{6.65}}}

\put(85.50,82.50){\makebox(0,0){\circle*{6.65}}}

\put(123.50,82.50){\makebox(0,0){\circle*{6.65}}}

\end{picture}
\end{center} \vspace*{-20pt}
\caption{Hexagonal Minesweeper solvability reduces to triangular Minesweeper solvability}
\label{hex2tri}
\end{figure}

\noindent
One can reduce hexagonal minesweeper to triangular minesweeper as follows.
In figure \ref{hex2tri}, each hexagonal compartment on the left hand side
is replaced by two triangular compartments: one with a light mine that
points downwards and another that points upwards. The triangle that points
upwards contains a dark mine if the corresponding hexagonal compartment 
on the left hand does. Otherwise, it contains the number of the corresponding
hexagon plus $6$. The `plus $6$' counts (for) the light mines surrounding
the triangular compartment. So apart for some minor border issues that are ignored
here, we have a reduction from hexagonal minesweeper to triangular minesweeper
here. It follows that triangular minesweeper is NP-complete.

\section{A stronger solvability result for normal Minesweeper}

Since Richard Kaye already proved that Minesweeper solvability is NP-com\-plete, 
it would be nice to improve on that. One way to do so is to observe that Minesweeper 
solvability is ASP-complete. It is true that all three variants of Minesweeper discussed 
above are indeed shown to be ASP-complete, while Richard Kaye's proof is not an ASP-proof
(if the inputs $u$ and $v$ of his AND-gate are both zero, then $r$ and $s$
can be interchanged). 

But that is not what I want to discuss here. No, we are going to show that determining
the solvability of a minesweeper board of which only one square is uncovered initially
is NP-complete. Of course, the uncovered square can only be surrounded by zero
mines, since otherwise it would be impossible to uncover any other square, in
which case you will not get any further.

Ok, say that there is one uncovered square with no mines surrounding it.
Then you can uncover all surrounding squares, and for each such square that has
no mines surrounding it either you can uncover the surrounding squares as well,
etc. All programs for minesweeper do this automatically.

So we get an area of uncovered squares consisting of connected squares with
no mines around them, from now on called a {\em whitespace component}. 
Furthermore, the border of the first whitespace component is uncovered as 
well, but the border squares do have mines around them. 

In order to show that determining
the solvability of a minesweeper board of which only one square 
is uncovered initially is NP-complete, it suffices to be able to do the 
following by local reasoning:
\begin{itemize}

\item reason through wires to uncover all whitespace components,

\item get to know all components except for the values of their variables.

\end{itemize}
Figure \ref{phasereason} shows how to get to know wires and to reason 
through them.

\begin{figure}[!htp]
\begin{center}
\begin{picture}(329.33,79.17)(-3.17,-41.17)

\mathversion{bold}
\filltype{shade}
\texture{40004000 0 0 0 00400040 0 0 0
         40004000 0 0 0 00400040 0 0 0
         40004000 0 0 0 00400040 0 0 0
         40004000 0 0 0 00400040 0 0 0}
\whiten

\path(19.00,38.00)(38.00,38.00)(38.00,19.00)(19.00,19.00)(19.00,38.00)
\put(28.50,24.50){\makebox(0,0)[b]{1}}

\path(0.00,0.00)(19.00,0.00)(19.00,-19.00)(0.00,-19.00)
\put(9.50,-9.50){\makebox(0,0){$\cdot$}}
\put(3.17,-9.50){\makebox(0,0){$\cdot$}}
\put(-3.17,-9.50){\makebox(0,0){$\cdot$}}

\path(0.00,19.00)(19.00,19.00)(19.00,0.00)(0.00,0.00)
\put(9.50,9.50){\makebox(0,0){$\cdot$}}
\put(3.17,9.50){\makebox(0,0){$\cdot$}}
\put(-3.17,9.50){\makebox(0,0){$\cdot$}}

\path(0.00,38.00)(19.00,38.00)(19.00,19.00)(0.00,19.00)
\put(9.50,28.50){\makebox(0,0){$\cdot$}}
\put(3.17,28.50){\makebox(0,0){$\cdot$}}
\put(-3.17,28.50){\makebox(0,0){$\cdot$}}

\path(19.00,-38.00)(19.00,-19.00)(38.00,-19.00)(38.00,-38.00)
\put(28.50,-28.50){\makebox(0,0){$\cdot$}}
\put(28.50,-34.83){\makebox(0,0){$\cdot$}}
\put(28.50,-41.17){\makebox(0,0){$\cdot$}}

\path(19.00,0.00)(38.00,0.00)(38.00,-19.00)(19.00,-19.00)(19.00,0.00)
\put(28.50,-9.50){\makebox(0,0){?}}

\path(19.00,19.00)(38.00,19.00)(38.00,0.00)(19.00,0.00)(19.00,19.00)
\put(28.50,9.50){\makebox(0,0){?}}

\path(38.00,-38.00)(38.00,-19.00)(57.00,-19.00)(57.00,-38.00)
\put(47.50,-28.50){\makebox(0,0){$\cdot$}}
\put(47.50,-34.83){\makebox(0,0){$\cdot$}}
\put(47.50,-41.17){\makebox(0,0){$\cdot$}}

\path(38.00,0.00)(57.00,0.00)(57.00,-19.00)(38.00,-19.00)(38.00,0.00)
\put(47.50,-9.50){\makebox(0,0){?}}

\path(38.00,19.00)(57.00,19.00)(57.00,0.00)(38.00,0.00)(38.00,19.00)
\put(47.50,9.50){\makebox(0,0){?}}

\path(38.00,38.00)(57.00,38.00)(57.00,19.00)(38.00,19.00)(38.00,38.00)
\put(47.50,24.50){\makebox(0,0)[b]{1}}

\path(57.00,-38.00)(57.00,-19.00)(76.00,-19.00)(76.00,-38.00)
\put(66.50,-28.50){\makebox(0,0){$\cdot$}}
\put(66.50,-34.83){\makebox(0,0){$\cdot$}}
\put(66.50,-41.17){\makebox(0,0){$\cdot$}}

\path(57.00,0.00)(76.00,0.00)(76.00,-19.00)(57.00,-19.00)(57.00,0.00)
\put(66.50,-9.50){\makebox(0,0){?}}

\path(57.00,19.00)(76.00,19.00)(76.00,0.00)(57.00,0.00)(57.00,19.00)
\put(66.50,9.50){\makebox(0,0){?}}

\path(57.00,38.00)(76.00,38.00)(76.00,19.00)(57.00,19.00)(57.00,38.00)
\put(66.50,24.50){\makebox(0,0)[b]{1}}

\path(76.00,-38.00)(76.00,-19.00)(95.00,-19.00)(95.00,-38.00)
\put(85.50,-28.50){\makebox(0,0){$\cdot$}}
\put(85.50,-34.83){\makebox(0,0){$\cdot$}}
\put(85.50,-41.17){\makebox(0,0){$\cdot$}}

\path(76.00,0.00)(95.00,0.00)(95.00,-19.00)(76.00,-19.00)(76.00,0.00)
\put(85.50,-9.50){\makebox(0,0){?}}

\path(76.00,19.00)(95.00,19.00)(95.00,0.00)(76.00,0.00)(76.00,19.00)
\put(85.50,9.50){\makebox(0,0){?}}

\path(76.00,38.00)(95.00,38.00)(95.00,19.00)(76.00,19.00)(76.00,38.00)
\put(85.50,24.50){\makebox(0,0)[b]{1}}

\path(95.00,-38.00)(95.00,-19.00)(114.00,-19.00)(114.00,-38.00)
\put(104.50,-28.50){\makebox(0,0){$\cdot$}}
\put(104.50,-34.83){\makebox(0,0){$\cdot$}}
\put(104.50,-41.17){\makebox(0,0){$\cdot$}}

\path(95.00,0.00)(114.00,0.00)(114.00,-19.00)(95.00,-19.00)(95.00,0.00)
\put(104.50,-9.50){\makebox(0,0){?}}

\path(95.00,19.00)(114.00,19.00)(114.00,0.00)(95.00,0.00)(95.00,19.00)
\put(104.50,9.50){\makebox(0,0){?}}

\path(95.00,38.00)(114.00,38.00)(114.00,19.00)(95.00,19.00)(95.00,38.00)
\put(104.50,24.50){\makebox(0,0)[b]{1}}

\path(114.00,-38.00)(114.00,-19.00)(133.00,-19.00)(133.00,-38.00)
\put(123.50,-28.50){\makebox(0,0){$\cdot$}}
\put(123.50,-34.83){\makebox(0,0){$\cdot$}}
\put(123.50,-41.17){\makebox(0,0){$\cdot$}}

\path(114.00,0.00)(133.00,0.00)(133.00,-19.00)(114.00,-19.00)(114.00,0.00)
\put(123.50,-9.50){\makebox(0,0){?}}

\path(114.00,19.00)(133.00,19.00)(133.00,0.00)(114.00,0.00)(114.00,19.00)
\put(123.50,9.50){\makebox(0,0){?}}

\path(114.00,38.00)(133.00,38.00)(133.00,19.00)(114.00,19.00)(114.00,38.00)
\put(123.50,24.50){\makebox(0,0)[b]{1}}

\path(133.00,-38.00)(133.00,-19.00)(152.00,-19.00)(152.00,-38.00)
\put(142.50,-28.50){\makebox(0,0){$\cdot$}}
\put(142.50,-34.83){\makebox(0,0){$\cdot$}}
\put(142.50,-41.17){\makebox(0,0){$\cdot$}}

\path(133.00,0.00)(152.00,0.00)(152.00,-19.00)(133.00,-19.00)(133.00,0.00)
\put(142.50,-9.50){\makebox(0,0){?}}

\path(133.00,19.00)(152.00,19.00)(152.00,0.00)(133.00,0.00)(133.00,19.00)
\put(142.50,9.50){\makebox(0,0){?}}

\path(133.00,38.00)(152.00,38.00)(152.00,19.00)(133.00,19.00)(133.00,38.00)
\put(142.50,24.50){\makebox(0,0)[b]{1}}

\path(152.00,-38.00)(152.00,-19.00)(171.00,-19.00)(171.00,-38.00)
\put(161.50,-28.50){\makebox(0,0){$\cdot$}}
\put(161.50,-34.83){\makebox(0,0){$\cdot$}}
\put(161.50,-41.17){\makebox(0,0){$\cdot$}}

\path(152.00,0.00)(171.00,0.00)(171.00,-19.00)(152.00,-19.00)(152.00,0.00)
\put(161.50,-9.50){\makebox(0,0){?}}

\path(152.00,19.00)(171.00,19.00)(171.00,0.00)(152.00,0.00)(152.00,19.00)
\put(161.50,9.50){\makebox(0,0){?}}

\path(152.00,38.00)(171.00,38.00)(171.00,19.00)(152.00,19.00)(152.00,38.00)
\put(161.50,24.50){\makebox(0,0)[b]{1}}

\path(171.00,-38.00)(171.00,-19.00)(190.00,-19.00)(190.00,-38.00)
\put(180.50,-28.50){\makebox(0,0){$\cdot$}}
\put(180.50,-34.83){\makebox(0,0){$\cdot$}}
\put(180.50,-41.17){\makebox(0,0){$\cdot$}}

\path(171.00,0.00)(190.00,0.00)(190.00,-19.00)(171.00,-19.00)(171.00,0.00)
\put(180.50,-9.50){\makebox(0,0){?}}

\path(171.00,19.00)(190.00,19.00)(190.00,0.00)(171.00,0.00)(171.00,19.00)
\put(180.50,9.50){\makebox(0,0){?}}

\path(171.00,38.00)(190.00,38.00)(190.00,19.00)(171.00,19.00)(171.00,38.00)
\put(180.50,24.50){\makebox(0,0)[b]{1}}

\path(190.00,-38.00)(190.00,-19.00)(209.00,-19.00)(209.00,-38.00)
\put(199.50,-28.50){\makebox(0,0){$\cdot$}}
\put(199.50,-34.83){\makebox(0,0){$\cdot$}}
\put(199.50,-41.17){\makebox(0,0){$\cdot$}}

\path(190.00,0.00)(209.00,0.00)(209.00,-19.00)(190.00,-19.00)(190.00,0.00)
\put(199.50,-9.50){\makebox(0,0){?}}

\path(190.00,19.00)(209.00,19.00)(209.00,0.00)(190.00,0.00)(190.00,19.00)
\put(199.50,9.50){\makebox(0,0){?}}

\path(190.00,38.00)(209.00,38.00)(209.00,19.00)(190.00,19.00)(190.00,38.00)
\put(199.50,24.50){\makebox(0,0)[b]{1}}

\path(209.00,-38.00)(209.00,-19.00)(228.00,-19.00)(228.00,-38.00)
\put(218.50,-28.50){\makebox(0,0){$\cdot$}}
\put(218.50,-34.83){\makebox(0,0){$\cdot$}}
\put(218.50,-41.17){\makebox(0,0){$\cdot$}}

\path(209.00,0.00)(228.00,0.00)(228.00,-19.00)(209.00,-19.00)(209.00,0.00)
\put(218.50,-9.50){\makebox(0,0){?}}

\path(209.00,19.00)(228.00,19.00)(228.00,0.00)(209.00,0.00)(209.00,19.00)
\put(218.50,9.50){\makebox(0,0){?}}

\path(209.00,38.00)(228.00,38.00)(228.00,19.00)(209.00,19.00)(209.00,38.00)
\put(218.50,24.50){\makebox(0,0)[b]{1}}

\path(228.00,-38.00)(228.00,-19.00)(247.00,-19.00)(247.00,-38.00)
\put(237.50,-28.50){\makebox(0,0){$\cdot$}}
\put(237.50,-34.83){\makebox(0,0){$\cdot$}}
\put(237.50,-41.17){\makebox(0,0){$\cdot$}}

\path(228.00,0.00)(247.00,0.00)(247.00,-19.00)(228.00,-19.00)(228.00,0.00)
\put(237.50,-9.50){\makebox(0,0){?}}

\path(228.00,19.00)(247.00,19.00)(247.00,0.00)(228.00,0.00)(228.00,19.00)
\put(237.50,9.50){\makebox(0,0){?}}

\path(228.00,38.00)(247.00,38.00)(247.00,19.00)(228.00,19.00)(228.00,38.00)
\put(237.50,24.50){\makebox(0,0)[b]{1}}

\path(247.00,-38.00)(247.00,-19.00)(266.00,-19.00)(266.00,-38.00)
\put(256.50,-28.50){\makebox(0,0){$\cdot$}}
\put(256.50,-34.83){\makebox(0,0){$\cdot$}}
\put(256.50,-41.17){\makebox(0,0){$\cdot$}}

\path(247.00,0.00)(266.00,0.00)(266.00,-19.00)(247.00,-19.00)(247.00,0.00)
\put(256.50,-9.50){\makebox(0,0){?}}

\path(247.00,19.00)(266.00,19.00)(266.00,0.00)(247.00,0.00)(247.00,19.00)
\put(256.50,9.50){\makebox(0,0){?}}

\path(247.00,38.00)(266.00,38.00)(266.00,19.00)(247.00,19.00)(247.00,38.00)
\put(256.50,24.50){\makebox(0,0)[b]{1}}

\path(266.00,-38.00)(266.00,-19.00)(285.00,-19.00)(285.00,-38.00)
\put(275.50,-28.50){\makebox(0,0){$\cdot$}}
\put(275.50,-34.83){\makebox(0,0){$\cdot$}}
\put(275.50,-41.17){\makebox(0,0){$\cdot$}}

\path(266.00,0.00)(285.00,0.00)(285.00,-19.00)(266.00,-19.00)(266.00,0.00)
\put(275.50,-9.50){\makebox(0,0){?}}

\path(266.00,19.00)(285.00,19.00)(285.00,0.00)(266.00,0.00)(266.00,19.00)
\put(275.50,9.50){\makebox(0,0){?}}

\path(266.00,38.00)(285.00,38.00)(285.00,19.00)(266.00,19.00)(266.00,38.00)
\put(275.50,24.50){\makebox(0,0)[b]{1}}

\path(285.00,-38.00)(285.00,-19.00)(304.00,-19.00)(304.00,-38.00)
\put(294.50,-28.50){\makebox(0,0){$\cdot$}}
\put(294.50,-34.83){\makebox(0,0){$\cdot$}}
\put(294.50,-41.17){\makebox(0,0){$\cdot$}}

\path(285.00,0.00)(304.00,0.00)(304.00,-19.00)(285.00,-19.00)(285.00,0.00)
\put(294.50,-9.50){\makebox(0,0){?}}

\path(285.00,19.00)(304.00,19.00)(304.00,0.00)(285.00,0.00)(285.00,19.00)
\put(294.50,9.50){\makebox(0,0){?}}

\path(285.00,38.00)(304.00,38.00)(304.00,19.00)(285.00,19.00)(285.00,38.00)
\put(294.50,24.50){\makebox(0,0)[b]{1}}

\path(323.00,0.00)(304.00,0.00)(304.00,19.00)(323.00,19.00)
\put(313.50,9.50){\makebox(0,0){$\cdot$}}
\put(319.83,9.50){\makebox(0,0){$\cdot$}}
\put(326.17,9.50){\makebox(0,0){$\cdot$}}

\path(323.00,-19.00)(304.00,-19.00)(304.00,0.00)(323.00,0.00)
\put(313.50,-9.50){\makebox(0,0){$\cdot$}}
\put(319.83,-9.50){\makebox(0,0){$\cdot$}}
\put(326.17,-9.50){\makebox(0,0){$\cdot$}}

\path(323.00,19.00)(304.00,19.00)(304.00,38.00)(323.00,38.00)
\put(313.50,28.50){\makebox(0,0){$\cdot$}}
\put(319.83,28.50){\makebox(0,0){$\cdot$}}
\put(326.17,28.50){\makebox(0,0){$\cdot$}}

\mathversion{normal}

\end{picture}
\\[10pt]
\begin{picture}(329.33,79.17)(-3.17,-41.17)

\mathversion{bold}
\filltype{shade}
\texture{40004000 0 0 0 00400040 0 0 0
         40004000 0 0 0 00400040 0 0 0
         40004000 0 0 0 00400040 0 0 0
         40004000 0 0 0 00400040 0 0 0}
\whiten

\path(19.00,38.00)(38.00,38.00)(38.00,19.00)(19.00,19.00)(19.00,38.00)
\put(28.50,24.50){\makebox(0,0)[b]{1}}

\path(0.00,0.00)(19.00,0.00)(19.00,-19.00)(0.00,-19.00)
\put(9.50,-9.50){\makebox(0,0){$\cdot$}}
\put(3.17,-9.50){\makebox(0,0){$\cdot$}}
\put(-3.17,-9.50){\makebox(0,0){$\cdot$}}

\path(0.00,19.00)(19.00,19.00)(19.00,0.00)(0.00,0.00)
\put(9.50,9.50){\makebox(0,0){$\cdot$}}
\put(3.17,9.50){\makebox(0,0){$\cdot$}}
\put(-3.17,9.50){\makebox(0,0){$\cdot$}}

\path(0.00,38.00)(19.00,38.00)(19.00,19.00)(0.00,19.00)
\put(9.50,28.50){\makebox(0,0){$\cdot$}}
\put(3.17,28.50){\makebox(0,0){$\cdot$}}
\put(-3.17,28.50){\makebox(0,0){$\cdot$}}

\path(19.00,-38.00)(19.00,-19.00)(38.00,-19.00)(38.00,-38.00)
\put(28.50,-28.50){\makebox(0,0){$\cdot$}}
\put(28.50,-34.83){\makebox(0,0){$\cdot$}}
\put(28.50,-41.17){\makebox(0,0){$\cdot$}}

\path(19.00,0.00)(38.00,0.00)(38.00,-19.00)(19.00,-19.00)(19.00,0.00)
\put(28.50,-9.50){\makebox(0,0){?}}

\path(19.00,19.00)(38.00,19.00)(38.00,0.00)(19.00,0.00)(19.00,19.00)
\put(28.50,9.50){\makebox(0,0){?}}

\path(38.00,-38.00)(38.00,-19.00)(57.00,-19.00)(57.00,-38.00)
\put(47.50,-28.50){\makebox(0,0){$\cdot$}}
\put(47.50,-34.83){\makebox(0,0){$\cdot$}}
\put(47.50,-41.17){\makebox(0,0){$\cdot$}}

\path(38.00,0.00)(57.00,0.00)(57.00,-19.00)(38.00,-19.00)(38.00,0.00)
\put(47.50,-9.50){\makebox(0,0){?}}

\path(38.00,19.00)(57.00,19.00)(57.00,0.00)(38.00,0.00)(38.00,19.00)
\put(47.50,9.50){\makebox(0,0){?}}

\path(38.00,38.00)(57.00,38.00)(57.00,19.00)(38.00,19.00)(38.00,38.00)
\put(47.50,24.50){\makebox(0,0)[b]{1}}

\path(57.00,-38.00)(57.00,-19.00)(76.00,-19.00)(76.00,-38.00)
\put(66.50,-28.50){\makebox(0,0){$\cdot$}}
\put(66.50,-34.83){\makebox(0,0){$\cdot$}}
\put(66.50,-41.17){\makebox(0,0){$\cdot$}}

\path(57.00,0.00)(76.00,0.00)(76.00,-19.00)(57.00,-19.00)(57.00,0.00)
\put(66.50,-9.50){\makebox(0,0){?}}

\path(57.00,19.00)(76.00,19.00)(76.00,0.00)(57.00,0.00)(57.00,19.00)
\put(66.50,9.50){\makebox(0,0){$\cdot$}}

\path(57.00,38.00)(76.00,38.00)(76.00,19.00)(57.00,19.00)(57.00,38.00)
\put(66.50,24.50){\makebox(0,0)[b]{1}}

\path(76.00,-38.00)(76.00,-19.00)(95.00,-19.00)(95.00,-38.00)
\put(85.50,-28.50){\makebox(0,0){$\cdot$}}
\put(85.50,-34.83){\makebox(0,0){$\cdot$}}
\put(85.50,-41.17){\makebox(0,0){$\cdot$}}

\path(76.00,0.00)(95.00,0.00)(95.00,-19.00)(76.00,-19.00)(76.00,0.00)
\put(85.50,-9.50){\makebox(0,0){?}}

\path(76.00,19.00)(95.00,19.00)(95.00,0.00)(76.00,0.00)(76.00,19.00)
\put(85.50,9.50){\makebox(0,0){?}}

\path(76.00,38.00)(95.00,38.00)(95.00,19.00)(76.00,19.00)(76.00,38.00)
\put(85.50,24.50){\makebox(0,0)[b]{1}}

\path(95.00,-38.00)(95.00,-19.00)(114.00,-19.00)(114.00,-38.00)
\put(104.50,-28.50){\makebox(0,0){$\cdot$}}
\put(104.50,-34.83){\makebox(0,0){$\cdot$}}
\put(104.50,-41.17){\makebox(0,0){$\cdot$}}

\path(95.00,0.00)(114.00,0.00)(114.00,-19.00)(95.00,-19.00)(95.00,0.00)
\put(104.50,-9.50){\makebox(0,0){?}}

\path(95.00,19.00)(114.00,19.00)(114.00,0.00)(95.00,0.00)(95.00,19.00)
\put(104.50,9.50){\makebox(0,0){?}}

\path(95.00,38.00)(114.00,38.00)(114.00,19.00)(95.00,19.00)(95.00,38.00)
\put(104.50,24.50){\makebox(0,0)[b]{2}}

\path(114.00,-38.00)(114.00,-19.00)(133.00,-19.00)(133.00,-38.00)
\put(123.50,-28.50){\makebox(0,0){$\cdot$}}
\put(123.50,-34.83){\makebox(0,0){$\cdot$}}
\put(123.50,-41.17){\makebox(0,0){$\cdot$}}

\path(114.00,0.00)(133.00,0.00)(133.00,-19.00)(114.00,-19.00)(114.00,0.00)
\put(123.50,-9.50){\makebox(0,0){?}}

\texture{cccccccc 0 0 0 cccccccc 0 0 0
         cccccccc 0 0 0 cccccccc 0 0 0
         cccccccc 0 0 0 cccccccc 0 0 0
         cccccccc 0 0 0 cccccccc 0 0 0}
\whiten

\path(114.00,19.00)(133.00,19.00)(133.00,0.00)(114.00,0.00)(114.00,19.00)
\put(123.50,9.50){\makebox(0,0){\circle*{9.50}}}

\path(114.00,38.00)(133.00,38.00)(133.00,19.00)(114.00,19.00)(114.00,38.00)
\put(123.50,24.50){\makebox(0,0)[b]{2}}

\path(133.00,-38.00)(133.00,-19.00)(152.00,-19.00)(152.00,-38.00)
\put(142.50,-28.50){\makebox(0,0){$\cdot$}}
\put(142.50,-34.83){\makebox(0,0){$\cdot$}}
\put(142.50,-41.17){\makebox(0,0){$\cdot$}}

\path(133.00,0.00)(152.00,0.00)(152.00,-19.00)(133.00,-19.00)(133.00,0.00)
\put(142.50,-9.50){\makebox(0,0){?}}

\path(133.00,19.00)(152.00,19.00)(152.00,0.00)(133.00,0.00)(133.00,19.00)
\put(142.50,9.50){\makebox(0,0){?}}

\path(133.00,38.00)(152.00,38.00)(152.00,19.00)(133.00,19.00)(133.00,38.00)
\put(142.50,24.50){\makebox(0,0)[b]{2}}

\path(152.00,-19.00)(171.00,-19.00)(171.00,-38.00)(152.00,-38.00)(152.00,-19.00)
\put(161.50,-28.50){\makebox(0,0){$\cdot$}}

\path(152.00,0.00)(171.00,0.00)(171.00,-19.00)(152.00,-19.00)(152.00,0.00)
\put(161.50,-9.50){\makebox(0,0){$\cdot$}}

\path(152.00,19.00)(171.00,19.00)(171.00,0.00)(152.00,0.00)(152.00,19.00)
\put(161.50,9.50){\makebox(0,0){?}}

\path(152.00,38.00)(171.00,38.00)(171.00,19.00)(152.00,19.00)(152.00,38.00)
\put(161.50,24.50){\makebox(0,0)[b]{1}}

\path(171.00,-19.00)(190.00,-19.00)(190.00,-38.00)(171.00,-38.00)(171.00,-19.00)
\put(180.50,-28.50){\makebox(0,0){$\cdot$}}

\path(171.00,0.00)(190.00,0.00)(190.00,-19.00)(171.00,-19.00)(171.00,0.00)
\put(180.50,-13.50){\makebox(0,0)[b]{1}}

\path(171.00,19.00)(190.00,19.00)(190.00,0.00)(171.00,0.00)(171.00,19.00)
\put(180.50,5.50){\makebox(0,0)[b]{1}}

\path(171.00,38.00)(190.00,38.00)(190.00,19.00)(171.00,19.00)(171.00,38.00)
\put(180.50,24.50){\makebox(0,0)[b]{1}}

\path(190.00,-19.00)(209.00,-19.00)(209.00,-38.00)(190.00,-38.00)(190.00,-19.00)
\put(199.50,-28.50){\makebox(0,0){$\cdot$}}

\path(190.00,0.00)(209.00,0.00)(209.00,-19.00)(190.00,-19.00)(190.00,0.00)
\put(199.50,-9.50){\makebox(0,0){$\cdot$}}

\path(190.00,19.00)(209.00,19.00)(209.00,0.00)(190.00,0.00)(190.00,19.00)
\put(199.50,9.50){\makebox(0,0){?}}

\path(190.00,38.00)(209.00,38.00)(209.00,19.00)(190.00,19.00)(190.00,38.00)
\put(199.50,24.50){\makebox(0,0)[b]{1}}

\path(209.00,-38.00)(209.00,-19.00)(228.00,-19.00)(228.00,-38.00)
\put(218.50,-28.50){\makebox(0,0){$\cdot$}}
\put(218.50,-34.83){\makebox(0,0){$\cdot$}}
\put(218.50,-41.17){\makebox(0,0){$\cdot$}}

\path(209.00,0.00)(228.00,0.00)(228.00,-19.00)(209.00,-19.00)(209.00,0.00)
\put(218.50,-9.50){\makebox(0,0){?}}

\path(209.00,19.00)(228.00,19.00)(228.00,0.00)(209.00,0.00)(209.00,19.00)
\put(218.50,9.50){\makebox(0,0){?}}

\path(209.00,38.00)(228.00,38.00)(228.00,19.00)(209.00,19.00)(209.00,38.00)
\put(218.50,24.50){\makebox(0,0)[b]{1}}

\path(228.00,-38.00)(228.00,-19.00)(247.00,-19.00)(247.00,-38.00)
\put(237.50,-28.50){\makebox(0,0){$\cdot$}}
\put(237.50,-34.83){\makebox(0,0){$\cdot$}}
\put(237.50,-41.17){\makebox(0,0){$\cdot$}}

\path(228.00,0.00)(247.00,0.00)(247.00,-19.00)(228.00,-19.00)(228.00,0.00)
\put(237.50,-9.50){\makebox(0,0){?}}

\path(228.00,19.00)(247.00,19.00)(247.00,0.00)(228.00,0.00)(228.00,19.00)
\put(237.50,9.50){\makebox(0,0){$\cdot$}}

\path(228.00,38.00)(247.00,38.00)(247.00,19.00)(228.00,19.00)(228.00,38.00)
\put(237.50,24.50){\makebox(0,0)[b]{1}}

\path(247.00,-38.00)(247.00,-19.00)(266.00,-19.00)(266.00,-38.00)
\put(256.50,-28.50){\makebox(0,0){$\cdot$}}
\put(256.50,-34.83){\makebox(0,0){$\cdot$}}
\put(256.50,-41.17){\makebox(0,0){$\cdot$}}

\path(247.00,0.00)(266.00,0.00)(266.00,-19.00)(247.00,-19.00)(247.00,0.00)
\put(256.50,-9.50){\makebox(0,0){?}}

\path(247.00,19.00)(266.00,19.00)(266.00,0.00)(247.00,0.00)(247.00,19.00)
\put(256.50,9.50){\makebox(0,0){?}}

\path(247.00,38.00)(266.00,38.00)(266.00,19.00)(247.00,19.00)(247.00,38.00)
\put(256.50,24.50){\makebox(0,0)[b]{1}}

\path(266.00,-38.00)(266.00,-19.00)(285.00,-19.00)(285.00,-38.00)
\put(275.50,-28.50){\makebox(0,0){$\cdot$}}
\put(275.50,-34.83){\makebox(0,0){$\cdot$}}
\put(275.50,-41.17){\makebox(0,0){$\cdot$}}

\path(266.00,0.00)(285.00,0.00)(285.00,-19.00)(266.00,-19.00)(266.00,0.00)
\put(275.50,-9.50){\makebox(0,0){?}}

\path(266.00,19.00)(285.00,19.00)(285.00,0.00)(266.00,0.00)(266.00,19.00)
\put(275.50,9.50){\makebox(0,0){?}}

\path(266.00,38.00)(285.00,38.00)(285.00,19.00)(266.00,19.00)(266.00,38.00)
\put(275.50,24.50){\makebox(0,0)[b]{1}}

\path(285.00,-38.00)(285.00,-19.00)(304.00,-19.00)(304.00,-38.00)
\put(294.50,-28.50){\makebox(0,0){$\cdot$}}
\put(294.50,-34.83){\makebox(0,0){$\cdot$}}
\put(294.50,-41.17){\makebox(0,0){$\cdot$}}

\path(285.00,0.00)(304.00,0.00)(304.00,-19.00)(285.00,-19.00)(285.00,0.00)
\put(294.50,-9.50){\makebox(0,0){?}}

\path(285.00,19.00)(304.00,19.00)(304.00,0.00)(285.00,0.00)(285.00,19.00)
\put(294.50,9.50){\makebox(0,0){$\cdot$}}

\path(285.00,38.00)(304.00,38.00)(304.00,19.00)(285.00,19.00)(285.00,38.00)
\put(294.50,24.50){\makebox(0,0)[b]{1}}

\path(323.00,-19.00)(304.00,-19.00)(304.00,0.00)(323.00,0.00)
\put(313.50,-9.50){\makebox(0,0){$\cdot$}}
\put(319.83,-9.50){\makebox(0,0){$\cdot$}}
\put(326.17,-9.50){\makebox(0,0){$\cdot$}}

\path(323.00,0.00)(304.00,0.00)(304.00,19.00)(323.00,19.00)
\put(313.50,9.50){\makebox(0,0){$\cdot$}}
\put(319.83,9.50){\makebox(0,0){$\cdot$}}
\put(326.17,9.50){\makebox(0,0){$\cdot$}}

\path(323.00,19.00)(304.00,19.00)(304.00,38.00)(323.00,38.00)
\put(313.50,28.50){\makebox(0,0){$\cdot$}}
\put(319.83,28.50){\makebox(0,0){$\cdot$}}
\put(326.17,28.50){\makebox(0,0){$\cdot$}}

\mathversion{normal}

\end{picture} 
\end{center} \vspace*{-10pt}
\caption{Mark the phase of a wire by an extra mine an you can reason through it} \label{phasereason}
\end{figure}

All other components of normal Minesweeper that are presented here can be
figured out as well, provided all whitespace components are uncovered and
the phases of all wires are known (see figure \ref{phasereason} as well).
The adder-multiplier combi needs the part between the dashed lines very 
sorely now. 

To increase the probability that one gets as far as this solvability check,
one can easily ensure that at least 99 percent of all squares do 
not have surrounding mines.

\section{Minesweeper is PP-hard}

In order to show that Minesweeper is PP-hard, we reduce from weak MAJ\-SAT.
Weak MAJSAT is the problem of estimating the probability that a circuit
is satisfied by $0$ or $1$, with an error of at most $0.5$, assuming that 
the inputs are random. Hence the error will always be $0.5$ when the probability
of satisfaction is exactly $0.5$, and both $0$ and $1$ are valid estimates
in that case. MAJSAT differs from weak MAJSAT that $0$ must be the output
when the probability of satisfaction is exactly $0.5$, and is known to 
be PP-complete \cite[Problem 11.5.16 (a)]{cc}.

\begin{lemma} \label{wMS}
Weak MAJSAT is PP-complete.
\end{lemma}

\begin{proof}
It is clear that weak MAJSAT is in PP, thus it suffices 
to show that weak MAJSAT is PP-hard. For that purpose,
assume we have a circuit that computes $f(x_1,x_2,\ldots,x_n)$,
where the $x_i$ are the inputs of the circuit. Build
another circuit which computes
$\min\{f(x_1,x_2,\ldots,x_n),\max\{y_1,y_2,\ldots,y_n\}\}$.
The latter circuit
cannot have probability $0.5$ of satisfaction, and the weak MAJ\-SAT
value of it is the MAJSAT value of the original circuit.
\end{proof}

\begin{theorem}
Minesweeper is PP-hard.
\end{theorem}

\begin{proof}
The only thing that we need to do in addition to the NP-completeness
proof of Minesweeper solvability is to wire back the output $s$
of our circuit to the starting points of the inputs, in such a way that
these inputs can be revealed when $s$ is known. This is done in figure 
\ref{ppwires}.

\begin{figure}[!htp]
\begin{center}
\begin{picture}(117.17,215.33)(0.00,-3.17)

\mathversion{bold}
\filltype{shade}
\texture{40004000 0 0 0 00400040 0 0 0
         40004000 0 0 0 00400040 0 0 0
         40004000 0 0 0 00400040 0 0 0
         40004000 0 0 0 00400040 0 0 0}
\whiten

\path(0.00,38.00)(19.00,38.00)(19.00,19.00)(0.00,19.00)(0.00,38.00)
\put(9.50,24.50){\makebox(0,0)[b]{1}}

\path(0.00,0.00)(0.00,19.00)(19.00,19.00)(19.00,0.00)
\put(9.50,9.50){\makebox(0,0){$\cdot$}}
\put(9.50,3.17){\makebox(0,0){$\cdot$}}
\put(9.50,-3.17){\makebox(0,0){$\cdot$}}

\path(0.00,57.00)(19.00,57.00)(19.00,38.00)(0.00,38.00)(0.00,57.00)
\put(9.50,43.50){\makebox(0,0)[b]{1}}

\path(0.00,76.00)(19.00,76.00)(19.00,57.00)(0.00,57.00)(0.00,76.00)
\put(9.50,62.50){\makebox(0,0)[b]{2}}

\path(0.00,95.00)(19.00,95.00)(19.00,76.00)(0.00,76.00)(0.00,95.00)
\put(9.50,81.50){\makebox(0,0)[b]{2}}

\path(0.00,114.00)(19.00,114.00)(19.00,95.00)(0.00,95.00)(0.00,114.00)
\put(9.50,100.50){\makebox(0,0)[b]{2}}

\path(0.00,133.00)(19.00,133.00)(19.00,114.00)(0.00,114.00)(0.00,133.00)
\put(9.50,119.50){\makebox(0,0)[b]{1}}

\path(0.00,152.00)(19.00,152.00)(19.00,133.00)(0.00,133.00)(0.00,152.00)
\put(9.50,138.50){\makebox(0,0)[b]{1}}

\path(0.00,171.00)(19.00,171.00)(19.00,152.00)(0.00,152.00)(0.00,171.00)
\put(9.50,157.50){\makebox(0,0)[b]{1}}

\path(0.00,190.00)(19.00,190.00)(19.00,171.00)(0.00,171.00)(0.00,190.00)
\put(9.50,176.50){\makebox(0,0)[b]{1}}

\path(19.00,209.00)(19.00,190.00)(0.00,190.00)(0.00,209.00)
\put(9.50,199.50){\makebox(0,0){$\cdot$}}
\put(9.50,205.83){\makebox(0,0){$\cdot$}}
\put(9.50,212.17){\makebox(0,0){$\cdot$}}

\path(19.00,0.00)(19.00,19.00)(38.00,19.00)(38.00,0.00)
\put(28.50,9.50){\makebox(0,0){$\cdot$}}
\put(28.50,3.17){\makebox(0,0){$\cdot$}}
\put(28.50,-3.17){\makebox(0,0){$\cdot$}}

\path(19.00,38.00)(38.00,38.00)(38.00,19.00)(19.00,19.00)(19.00,38.00)
\put(28.50,24.50){\makebox(0,0)[b]{1}}

\shade
\path(19.00,57.00)(38.00,57.00)(38.00,38.00)(19.00,38.00)(19.00,57.00)
\put(28.50,43.50){\makebox(0,0)[b]{$s$}}

\shade
\path(19.00,76.00)(38.00,76.00)(38.00,57.00)(19.00,57.00)(19.00,76.00)
\put(28.50,62.50){\makebox(0,0)[b]{$s'$}}

\shade
\path(19.00,114.00)(38.00,114.00)(38.00,95.00)(19.00,95.00)(19.00,114.00)
\put(28.50,100.50){\makebox(0,0)[b]{$s$}}

\shade
\path(19.00,133.00)(38.00,133.00)(38.00,114.00)(19.00,114.00)(19.00,133.00)
\put(28.50,119.50){\makebox(0,0)[b]{$s'$}}

\path(19.00,152.00)(38.00,152.00)(38.00,133.00)(19.00,133.00)(19.00,152.00)
\put(28.50,138.50){\makebox(0,0)[b]{1}}

\shade
\path(19.00,171.00)(38.00,171.00)(38.00,152.00)(19.00,152.00)(19.00,171.00)
\put(28.50,157.50){\makebox(0,0)[b]{$s$}}

\shade
\path(19.00,190.00)(38.00,190.00)(38.00,171.00)(19.00,171.00)(19.00,190.00)
\put(28.50,176.50){\makebox(0,0)[b]{$s'$}}

\path(38.00,209.00)(38.00,190.00)(19.00,190.00)(19.00,209.00)
\put(28.50,199.50){\makebox(0,0){$\cdot$}}
\put(28.50,205.83){\makebox(0,0){$\cdot$}}
\put(28.50,212.17){\makebox(0,0){$\cdot$}}

\path(38.00,0.00)(38.00,19.00)(57.00,19.00)(57.00,0.00)
\put(47.50,9.50){\makebox(0,0){$\cdot$}}
\put(47.50,3.17){\makebox(0,0){$\cdot$}}
\put(47.50,-3.17){\makebox(0,0){$\cdot$}}

\path(38.00,38.00)(57.00,38.00)(57.00,19.00)(38.00,19.00)(38.00,38.00)
\put(47.50,24.50){\makebox(0,0)[b]{1}}

\path(38.00,57.00)(57.00,57.00)(57.00,38.00)(38.00,38.00)(38.00,57.00)
\put(47.50,43.50){\makebox(0,0)[b]{1}}

\path(38.00,76.00)(57.00,76.00)(57.00,57.00)(38.00,57.00)(38.00,76.00)
\put(47.50,62.50){\makebox(0,0)[b]{2}}

\path(38.00,95.00)(57.00,95.00)(57.00,76.00)(38.00,76.00)(38.00,95.00)
\put(47.50,81.50){\makebox(0,0)[b]{3}}

\shade
\path(38.00,114.00)(57.00,114.00)(57.00,95.00)(38.00,95.00)(38.00,114.00)
\put(47.50,100.50){\makebox(0,0)[b]{$x'$}}

\path(38.00,133.00)(57.00,133.00)(57.00,114.00)(38.00,114.00)(38.00,133.00)
\put(47.50,119.50){\makebox(0,0)[b]{2}}

\path(38.00,152.00)(57.00,152.00)(57.00,133.00)(38.00,133.00)(38.00,152.00)
\put(47.50,138.50){\makebox(0,0)[b]{1}}

\path(38.00,171.00)(57.00,171.00)(57.00,152.00)(38.00,152.00)(38.00,171.00)
\put(47.50,157.50){\makebox(0,0)[b]{1}}

\path(38.00,190.00)(57.00,190.00)(57.00,171.00)(38.00,171.00)(38.00,190.00)
\put(47.50,176.50){\makebox(0,0)[b]{1}}

\path(57.00,209.00)(57.00,190.00)(38.00,190.00)(38.00,209.00)
\put(47.50,199.50){\makebox(0,0){$\cdot$}}
\put(47.50,205.83){\makebox(0,0){$\cdot$}}
\put(47.50,212.17){\makebox(0,0){$\cdot$}}

\path(57.00,95.00)(76.00,95.00)(76.00,76.00)(57.00,76.00)(57.00,95.00)
\put(66.50,81.50){\makebox(0,0)[b]{1}}

\shade
\path(57.00,114.00)(76.00,114.00)(76.00,95.00)(57.00,95.00)(57.00,114.00)
\put(66.50,100.50){\makebox(0,0)[b]{$x$}}

\path(57.00,133.00)(76.00,133.00)(76.00,114.00)(57.00,114.00)(57.00,133.00)
\put(66.50,119.50){\makebox(0,0)[b]{1}}

\path(76.00,95.00)(95.00,95.00)(95.00,76.00)(76.00,76.00)(76.00,95.00)
\put(85.50,81.50){\makebox(0,0)[b]{1}}

\path(76.00,114.00)(95.00,114.00)(95.00,95.00)(76.00,95.00)(76.00,114.00)
\put(85.50,100.50){\makebox(0,0)[b]{1}}

\path(76.00,133.00)(95.00,133.00)(95.00,114.00)(76.00,114.00)(76.00,133.00)
\put(85.50,119.50){\makebox(0,0)[b]{1}}

\path(114.00,76.00)(95.00,76.00)(95.00,95.00)(114.00,95.00)
\put(104.50,85.50){\makebox(0,0){$\cdot$}}
\put(110.83,85.50){\makebox(0,0){$\cdot$}}
\put(117.17,85.50){\makebox(0,0){$\cdot$}}

\path(114.00,95.00)(95.00,95.00)(95.00,114.00)(114.00,114.00)
\put(104.50,104.50){\makebox(0,0){$\cdot$}}
\put(110.83,104.50){\makebox(0,0){$\cdot$}}
\put(117.17,104.50){\makebox(0,0){$\cdot$}}

\path(114.00,114.00)(95.00,114.00)(95.00,133.00)(114.00,133.00)
\put(104.50,123.50){\makebox(0,0){$\cdot$}}
\put(110.83,123.50){\makebox(0,0){$\cdot$}}
\put(117.17,123.50){\makebox(0,0){$\cdot$}}

\texture{cccccccc 0 0 0 cccccccc 0 0 0
         cccccccc 0 0 0 cccccccc 0 0 0
         cccccccc 0 0 0 cccccccc 0 0 0
         cccccccc 0 0 0 cccccccc 0 0 0}
\whiten

\path(19.00,95.00)(38.00,95.00)(38.00,76.00)(19.00,76.00)(19.00,95.00)
\put(28.50,85.50){\makebox(0,0){\circle*{9.50}}}

\put(66.5,201){\makebox(0,0)[b]{$s$}}
\put(66.5,199){\vector(0,-1){30}}
\put(66.5,40){\makebox(0,0)[b]{$s$}}
\put(66.5,38){\vector(0,-1){30}}
\put(85,66.5){\makebox(0,0)[r]{$x$}}
\put(86,66.5){\vector(1,0){30}}
\mathversion{normal}

\end{picture}
\begin{picture}(174.17,220.00)(0.00,0.00)

\mathversion{bold}
\filltype{shade}
\texture{40004000 0 0 0 00400040 0 0 0
         40004000 0 0 0 00400040 0 0 0
         40004000 0 0 0 00400040 0 0 0
         40004000 0 0 0 00400040 0 0 0}
\whiten

\path(9.50,49.50)(9.50,38.50)(19.00,33.00)(28.50,38.50)
     (28.50,49.50)(19.00,55.00)(9.50,49.50)
\put(19.00,40.00){\makebox(0,0)[b]{1}}

\path(9.50,16.50)(19.00,22.00)(19.00,33.00)
     (9.50,38.50)(0.00,33.00)(0.00,22.00)
\put(9.50,27.50){\makebox(0,0){$\cdot$}}
\put(6.33,22.00){\makebox(0,0){$\cdot$}}
\put(3.17,16.50){\makebox(0,0){$\cdot$}}

\path(28.50,16.50)(38.00,22.00)(38.00,33.00)
     (28.50,38.50)(19.00,33.00)(19.00,22.00)
\put(28.50,27.50){\makebox(0,0){$\cdot$}}
\put(25.33,22.00){\makebox(0,0){$\cdot$}}
\put(22.17,16.50){\makebox(0,0){$\cdot$}}

\path(19.00,66.00)(19.00,55.00)(28.50,49.50)(38.00,55.00)
     (38.00,66.00)(28.50,71.50)(19.00,66.00)
\put(28.50,56.50){\makebox(0,0)[b]{1}}

\path(38.00,0.00)(47.50,5.50)(47.50,16.50)
     (38.00,22.00)(28.50,16.50)(28.50,5.50)
\put(38.00,11.00){\makebox(0,0){$\cdot$}}
\put(34.83,5.50){\makebox(0,0){$\cdot$}}
\put(31.67,0.00){\makebox(0,0){$\cdot$}}

\shade
\path(28.50,49.50)(28.50,38.50)(38.00,33.00)(47.50,38.50)
     (47.50,49.50)(38.00,55.00)(28.50,49.50)
\put(38.00,40.00){\makebox(0,0)[b]{$s$}}

\path(28.50,82.50)(28.50,71.50)(38.00,66.00)(47.50,71.50)
     (47.50,82.50)(38.00,88.00)(28.50,82.50)
\put(38.00,73.00){\makebox(0,0)[b]{1}}

\path(38.00,33.00)(38.00,22.00)(47.50,16.50)(57.00,22.00)
     (57.00,33.00)(47.50,38.50)(38.00,33.00)
\put(47.50,23.50){\makebox(0,0)[b]{1}}

\shade
\path(38.00,66.00)(38.00,55.00)(47.50,49.50)(57.00,55.00)
     (57.00,66.00)(47.50,71.50)(38.00,66.00)
\put(47.50,56.50){\makebox(0,0)[b]{$s'$}}

\path(38.00,99.00)(38.00,88.00)(47.50,82.50)(57.00,88.00)
     (57.00,99.00)(47.50,104.50)(38.00,99.00)
\put(47.50,89.50){\makebox(0,0)[b]{1}}

\path(47.50,49.50)(47.50,38.50)(57.00,33.00)(66.50,38.50)
     (66.50,49.50)(57.00,55.00)(47.50,49.50)
\put(57.00,40.00){\makebox(0,0)[b]{2}}

\shade
\path(47.50,82.50)(47.50,71.50)(57.00,66.00)(66.50,71.50)
     (66.50,82.50)(57.00,88.00)(47.50,82.50)
\put(57.00,73.00){\makebox(0,0)[b]{$s$}}

\path(47.50,115.50)(47.50,104.50)(57.00,99.00)(66.50,104.50)
     (66.50,115.50)(57.00,121.00)(47.50,115.50)
\put(57.00,106.00){\makebox(0,0)[b]{1}}

\shade
\path(57.00,99.00)(57.00,88.00)(66.50,82.50)(76.00,88.00)
     (76.00,99.00)(66.50,104.50)(57.00,99.00)
\put(66.50,89.50){\makebox(0,0)[b]{$s'$}}

\path(57.00,132.00)(57.00,121.00)(66.50,115.50)(76.00,121.00)
     (76.00,132.00)(66.50,137.50)(57.00,132.00)
\put(66.50,122.50){\makebox(0,0)[b]{1}}

\path(66.50,49.50)(66.50,38.50)(76.00,33.00)(85.50,38.50)
     (85.50,49.50)(76.00,55.00)(66.50,49.50)
\put(76.00,40.00){\makebox(0,0)[b]{1}}

\shade
\path(66.50,82.50)(66.50,71.50)(76.00,66.00)(85.50,71.50)
     (85.50,82.50)(76.00,88.00)(66.50,82.50)
\put(76.00,73.00){\makebox(0,0)[b]{$x'$}}

\shade
\path(66.50,115.50)(66.50,104.50)(76.00,99.00)(85.50,104.50)
     (85.50,115.50)(76.00,121.00)(66.50,115.50)
\put(76.00,106.00){\makebox(0,0)[b]{$s$}}

\path(66.50,148.50)(66.50,137.50)(76.00,132.00)(85.50,137.50)
     (85.50,148.50)(76.00,154.00)(66.50,148.50)
\put(76.00,139.00){\makebox(0,0)[b]{1}}

\path(76.00,66.00)(76.00,55.00)(85.50,49.50)(95.00,55.00)
     (95.00,66.00)(85.50,71.50)(76.00,66.00)
\put(85.50,56.50){\makebox(0,0)[b]{2}}

\path(76.00,99.00)(76.00,88.00)(85.50,82.50)(95.00,88.00)
     (95.00,99.00)(85.50,104.50)(76.00,99.00)
\put(85.50,89.50){\makebox(0,0)[b]{2}}

\shade
\path(76.00,132.00)(76.00,121.00)(85.50,115.50)(95.00,121.00)
     (95.00,132.00)(85.50,137.50)(76.00,132.00)
\put(85.50,122.50){\makebox(0,0)[b]{$s'$}}

\path(76.00,165.00)(76.00,154.00)(85.50,148.50)(95.00,154.00)
     (95.00,165.00)(85.50,170.50)(76.00,165.00)
\put(85.50,155.50){\makebox(0,0)[b]{1}}

\shade
\path(85.50,82.50)(85.50,71.50)(95.00,66.00)(104.50,71.50)
     (104.50,82.50)(95.00,88.00)(85.50,82.50)
\put(95.00,73.00){\makebox(0,0)[b]{$x$}}

\path(85.50,115.50)(85.50,104.50)(95.00,99.00)(104.50,104.50)
     (104.50,115.50)(95.00,121.00)(85.50,115.50)
\put(95.00,106.00){\makebox(0,0)[b]{1}}

\shade
\path(85.50,148.50)(85.50,137.50)(95.00,132.00)(104.50,137.50)
     (104.50,148.50)(95.00,154.00)(85.50,148.50)
\put(95.00,139.00){\makebox(0,0)[b]{$s$}}

\path(85.50,181.50)(85.50,170.50)(95.00,165.00)(104.50,170.50)
     (104.50,181.50)(95.00,187.00)(85.50,181.50)
\put(95.00,172.00){\makebox(0,0)[b]{1}}

\path(95.00,66.00)(95.00,55.00)(104.50,49.50)(114.00,55.00)
     (114.00,66.00)(104.50,71.50)(95.00,66.00)
\put(104.50,56.50){\makebox(0,0)[b]{1}}

\path(95.00,99.00)(95.00,88.00)(104.50,82.50)(114.00,88.00)
     (114.00,99.00)(104.50,104.50)(95.00,99.00)
\put(104.50,89.50){\makebox(0,0)[b]{1}}

\path(95.00,132.00)(95.00,121.00)(104.50,115.50)(114.00,121.00)
     (114.00,132.00)(104.50,137.50)(95.00,132.00)
\put(104.50,122.50){\makebox(0,0)[b]{1}}

\shade
\path(95.00,165.00)(95.00,154.00)(104.50,148.50)(114.00,154.00)
     (114.00,165.00)(104.50,170.50)(95.00,165.00)
\put(104.50,155.50){\makebox(0,0)[b]{$s'$}}

\path(95.00,198.00)(95.00,187.00)(104.50,181.50)(114.00,187.00)
     (114.00,198.00)(104.50,203.50)(95.00,198.00)
\put(104.50,188.50){\makebox(0,0)[b]{1}}

\shade
\path(104.50,82.50)(104.50,71.50)(114.00,66.00)(123.50,71.50)
     (123.50,82.50)(114.00,88.00)(104.50,82.50)
\put(114.00,73.00){\makebox(0,0)[b]{$x'$}}

\path(104.50,148.50)(104.50,137.50)(114.00,132.00)(123.50,137.50)
     (123.50,148.50)(114.00,154.00)(104.50,148.50)
\put(114.00,139.00){\makebox(0,0)[b]{1}}

\shade
\path(104.50,181.50)(104.50,170.50)(114.00,165.00)(123.50,170.50)
     (123.50,181.50)(114.00,187.00)(104.50,181.50)
\put(114.00,172.00){\makebox(0,0)[b]{$s$}}

\path(114.00,220.00)(104.50,214.50)(104.50,203.50)
     (114.00,198.00)(123.50,203.50)(123.50,214.50)
\put(114.00,209.00){\makebox(0,0){$\cdot$}}
\put(117.17,214.50){\makebox(0,0){$\cdot$}}
\put(120.33,220.00){\makebox(0,0){$\cdot$}}

\path(114.00,66.00)(114.00,55.00)(123.50,49.50)(133.00,55.00)
     (133.00,66.00)(123.50,71.50)(114.00,66.00)
\put(123.50,56.50){\makebox(0,0)[b]{1}}

\path(114.00,99.00)(114.00,88.00)(123.50,82.50)(133.00,88.00)
     (133.00,99.00)(123.50,104.50)(114.00,99.00)
\put(123.50,89.50){\makebox(0,0)[b]{1}}

\path(114.00,165.00)(114.00,154.00)(123.50,148.50)(133.00,154.00)
     (133.00,165.00)(123.50,170.50)(114.00,165.00)
\put(123.50,155.50){\makebox(0,0)[b]{1}}

\shade
\path(114.00,198.00)(114.00,187.00)(123.50,181.50)(133.00,187.00)
     (133.00,198.00)(123.50,203.50)(114.00,198.00)
\put(123.50,188.50){\makebox(0,0)[b]{$s'$}}

\shade
\path(123.50,82.50)(123.50,71.50)(133.00,66.00)(142.50,71.50)
     (142.50,82.50)(133.00,88.00)(123.50,82.50)
\put(133.00,73.00){\makebox(0,0)[b]{$x$}}

\path(123.50,181.50)(123.50,170.50)(133.00,165.00)(142.50,170.50)
     (142.50,181.50)(133.00,187.00)(123.50,181.50)
\put(133.00,172.00){\makebox(0,0)[b]{1}}

\path(133.00,220.00)(123.50,214.50)(123.50,203.50)
     (133.00,198.00)(142.50,203.50)(142.50,214.50)
\put(133.00,209.00){\makebox(0,0){$\cdot$}}
\put(136.17,214.50){\makebox(0,0){$\cdot$}}
\put(139.33,220.00){\makebox(0,0){$\cdot$}}

\path(133.00,66.00)(133.00,55.00)(142.50,49.50)(152.00,55.00)
     (152.00,66.00)(142.50,71.50)(133.00,66.00)
\put(142.50,56.50){\makebox(0,0)[b]{1}}

\path(133.00,99.00)(133.00,88.00)(142.50,82.50)(152.00,88.00)
     (152.00,99.00)(142.50,104.50)(133.00,99.00)
\put(142.50,89.50){\makebox(0,0)[b]{1}}

\path(142.50,203.50)(133.00,198.00)(133.00,187.00)
     (142.50,181.50)(152.00,187.00)(152.00,198.00)
\put(142.50,192.50){\makebox(0,0){$\cdot$}}
\put(145.67,198.00){\makebox(0,0){$\cdot$}}
\put(148.83,203.50){\makebox(0,0){$\cdot$}}

\path(161.50,82.50)(152.00,88.00)(142.50,82.50)
     (142.50,71.50)(152.00,66.00)(161.50,71.50)
\put(152.00,77.00){\makebox(0,0){$\cdot$}}
\put(158.33,77.00){\makebox(0,0){$\cdot$}}
\put(164.67,77.00){\makebox(0,0){$\cdot$}}

\path(171.00,66.00)(161.50,71.50)(152.00,66.00)
     (152.00,55.00)(161.50,49.50)(171.00,55.00)
\put(161.50,60.50){\makebox(0,0){$\cdot$}}
\put(167.83,60.50){\makebox(0,0){$\cdot$}}
\put(174.17,60.50){\makebox(0,0){$\cdot$}}

\path(171.00,99.00)(161.50,104.50)(152.00,99.00)
     (152.00,88.00)(161.50,82.50)(171.00,88.00)
\put(161.50,93.50){\makebox(0,0){$\cdot$}}
\put(167.83,93.50){\makebox(0,0){$\cdot$}}
\put(174.17,93.50){\makebox(0,0){$\cdot$}}

\texture{cccccccc 0 0 0 cccccccc 0 0 0
         cccccccc 0 0 0 cccccccc 0 0 0
         cccccccc 0 0 0 cccccccc 0 0 0
         cccccccc 0 0 0 cccccccc 0 0 0}
\whiten

\path(57.00,66.00)(57.00,55.00)(66.50,49.50)(76.00,55.00)
     (76.00,66.00)(66.50,71.50)(57.00,66.00)
\put(66.50,60.50){\makebox(0,0){\circle*{9.50}}}

\put(60,155.5){\makebox(0,0)[lb]{$s$}}
\put(61,153.5){\vector(-1,-2){13}}
\put(110,42){\makebox(0,0)[r]{$x$}}
\put(112,42){\vector(1,0){30}}
\mathversion{normal}

\end{picture}
\end{center} \vspace*{-10pt}
\caption{$x$ can be revealed when $s$ is known} \label{ppwires}
\end{figure}
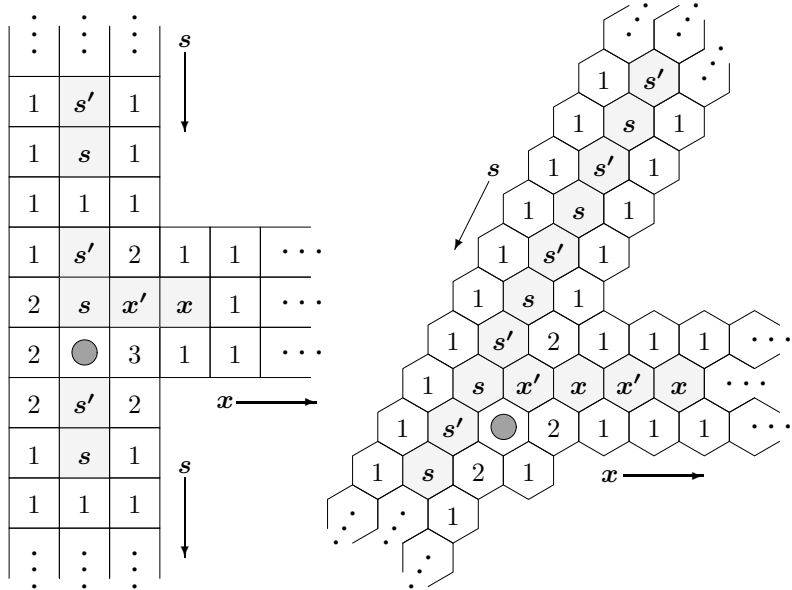

Guessing on other spots than one with $s$ or $s'$ first is useless. One can already
see in advance that such a guess does not give more information than that
a subcircuit has a certain value, in case one does not die. It is equally 
useful just to assume that that subcircuit has a certain value without checking 
it. Thus we have to round the number $\vartheta$ that satisfies
$$
\forrandom x_1 \forrandom x_2 \cdots \forrandom x_n 
\Pr[f(x_1,x_2,\ldots,x_n)]=\vartheta
$$ 
to either $0$ or $1$, with a rounding error of at most $0.5$.
\end{proof}

See \cite[\S 4]{randomq} for the meaning of the random quantifier
$\forrandom$. This quantifier will be used again in the next section.

\section{PSPACE-completeness of Minesweeper}

Now we know that Minesweeper in PP-hard, but the question remains if
Minesweeper is complete for some known complexity class. Since one can show that
Minesweeper is in PSPACE, it seems natural to ask whether Minesweeper is 
PSPACE-complete.

\begin{theorem}
Minesweeper is PSPACE-complete when the probability of revealing all mines 
may be infinitesimal.
\end{theorem}

\begin{proof}
Let $f(x_1,x_2,\ldots,x_n)$ be a boolean fomula and
$E$ be a proper subset of $\{1,2,\ldots,n\}$ of cardinality $e$.
Assume that either $m = 2$ or $m = n + 2 - e$. We define quantifiers
$A_i$ and $R_i$ for all positive $i \le n + 1 + m$, as follows:
\begin{align*}
A_i &= \left\{ \begin{array}{ll}
\exists & \text{if $i \in E$,} \\
\forall & \text{if $i \not\in E$,}
\end{array} \right. &
R_i &= \left\{ \begin{array}{ll}
\exists & \text{if $i \in E$,} \\
\forrandom & \text{if $i \not\in E$.}
\end{array} \right.
\end{align*}
When $m = n + 2 - e$, we reduce from QBF, by determining the validity of
\begin{equation} \label{mn2e}
A_1 x_1 A_2 x_2 \cdots A_n x_n f(x_1,x_2,\ldots,x_n)
\end{equation}
When $m = 2$, we reduce from a combination of QBF and
weak MAJSAT, say weak nonalternating SSAT, by rounding
\begin{equation} \label{m2}
\max\Big\{\vartheta \,\Big|\, R_1 x_1 R_2 x_2 \cdots R_n x_n 
\Pr[f(x_1,x_2,\ldots,x_n)]=\vartheta\Big\}
\end{equation}
to either $0$ or $1$, with a rounding error of at most 0.5. Weak 
nonalternating SSAT is PSPACE-complete as well, since (alternating) 
SSAT \cite[Th.\@ 2]{randomq} can be reduced to it in a similar manner as 
weak MAJSAT was reduced to MAJSAT in the proof of lemma \ref{wMS}.

The probability of removing all mines in our Minesweeper game will lie
between $\frac34 \cdot 2^{-e}$ and $2^{-e}$ inclusive, which is 
(unfortunately) infinitesimal for large $e$.

We lay down a circuit for $s_1 = f(x_1,x_2,\ldots,x_n)$ and circuitry for
$s_2 = \max\{s_1, x_{n+1}\}$ and $s_3 = \max\{x_{(n+1)+1},x_{(n+1)+2},\ldots,
x_{(n+1)+m}\}$. Besides the variables $x_i$ for all positive $i \le n+1 + m$, 
we make variables $z_{ij}$, where $1 \le i < j \le n + 1 + m$, such that
any variable $z_{ij}$ can be revealed when $x_i$ is known, as in figure
\ref{ppwires}. 

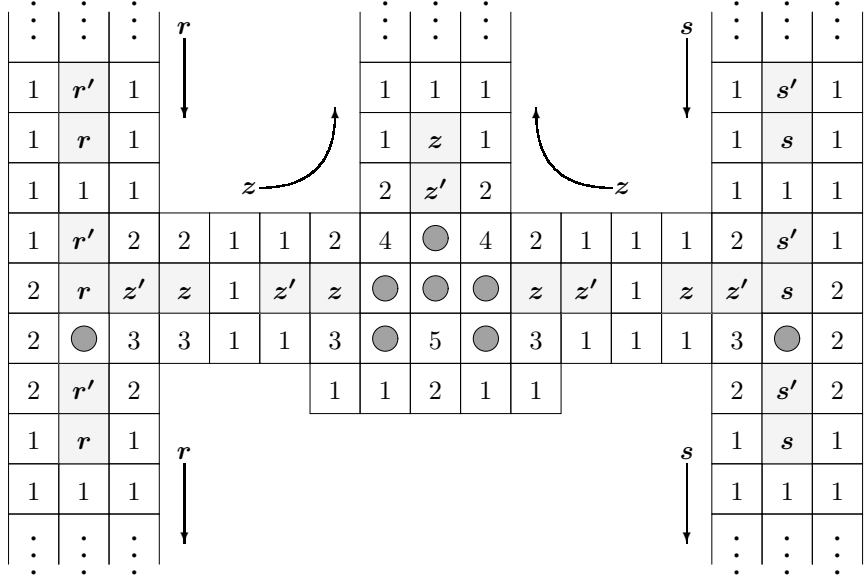
\begin{figure}[!htp]
\begin{center}
\begin{picture}(323.00,215.33)(0.00,-3.17)

\mathversion{bold}
\filltype{shade}
\texture{40004000 0 0 0 00400040 0 0 0
         40004000 0 0 0 00400040 0 0 0
         40004000 0 0 0 00400040 0 0 0
         40004000 0 0 0 00400040 0 0 0}
\whiten

\path(0.00,38.00)(19.00,38.00)(19.00,19.00)(0.00,19.00)(0.00,38.00)
\put(9.50,24.50){\makebox(0,0)[b]{1}}

\path(0.00,0.00)(0.00,19.00)(19.00,19.00)(19.00,0.00)
\put(9.50,9.50){\makebox(0,0){$\cdot$}}
\put(9.50,3.17){\makebox(0,0){$\cdot$}}
\put(9.50,-3.17){\makebox(0,0){$\cdot$}}

\path(0.00,57.00)(19.00,57.00)(19.00,38.00)(0.00,38.00)(0.00,57.00)
\put(9.50,43.50){\makebox(0,0)[b]{1}}

\path(0.00,76.00)(19.00,76.00)(19.00,57.00)(0.00,57.00)(0.00,76.00)
\put(9.50,62.50){\makebox(0,0)[b]{2}}

\path(0.00,95.00)(19.00,95.00)(19.00,76.00)(0.00,76.00)(0.00,95.00)
\put(9.50,81.50){\makebox(0,0)[b]{2}}

\path(0.00,114.00)(19.00,114.00)(19.00,95.00)(0.00,95.00)(0.00,114.00)
\put(9.50,100.50){\makebox(0,0)[b]{2}}

\path(0.00,133.00)(19.00,133.00)(19.00,114.00)(0.00,114.00)(0.00,133.00)
\put(9.50,119.50){\makebox(0,0)[b]{1}}

\path(0.00,152.00)(19.00,152.00)(19.00,133.00)(0.00,133.00)(0.00,152.00)
\put(9.50,138.50){\makebox(0,0)[b]{1}}

\path(0.00,171.00)(19.00,171.00)(19.00,152.00)(0.00,152.00)(0.00,171.00)
\put(9.50,157.50){\makebox(0,0)[b]{1}}

\path(0.00,190.00)(19.00,190.00)(19.00,171.00)(0.00,171.00)(0.00,190.00)
\put(9.50,176.50){\makebox(0,0)[b]{1}}

\path(19.00,209.00)(19.00,190.00)(0.00,190.00)(0.00,209.00)
\put(9.50,199.50){\makebox(0,0){$\cdot$}}
\put(9.50,205.83){\makebox(0,0){$\cdot$}}
\put(9.50,212.17){\makebox(0,0){$\cdot$}}

\path(19.00,0.00)(19.00,19.00)(38.00,19.00)(38.00,0.00)
\put(28.50,9.50){\makebox(0,0){$\cdot$}}
\put(28.50,3.17){\makebox(0,0){$\cdot$}}
\put(28.50,-3.17){\makebox(0,0){$\cdot$}}

\path(19.00,38.00)(38.00,38.00)(38.00,19.00)(19.00,19.00)(19.00,38.00)
\put(28.50,24.50){\makebox(0,0)[b]{1}}

\shade
\path(19.00,57.00)(38.00,57.00)(38.00,38.00)(19.00,38.00)(19.00,57.00)
\put(28.50,43.50){\makebox(0,0)[b]{$r$}}

\shade
\path(19.00,76.00)(38.00,76.00)(38.00,57.00)(19.00,57.00)(19.00,76.00)
\put(28.50,62.50){\makebox(0,0)[b]{$r'$}}

\shade
\path(19.00,114.00)(38.00,114.00)(38.00,95.00)(19.00,95.00)(19.00,114.00)
\put(28.50,100.50){\makebox(0,0)[b]{$r$}}

\shade
\path(19.00,133.00)(38.00,133.00)(38.00,114.00)(19.00,114.00)(19.00,133.00)
\put(28.50,119.50){\makebox(0,0)[b]{$r'$}}

\path(19.00,152.00)(38.00,152.00)(38.00,133.00)(19.00,133.00)(19.00,152.00)
\put(28.50,138.50){\makebox(0,0)[b]{1}}

\shade
\path(19.00,171.00)(38.00,171.00)(38.00,152.00)(19.00,152.00)(19.00,171.00)
\put(28.50,157.50){\makebox(0,0)[b]{$r$}}

\shade
\path(19.00,190.00)(38.00,190.00)(38.00,171.00)(19.00,171.00)(19.00,190.00)
\put(28.50,176.50){\makebox(0,0)[b]{$r'$}}

\path(38.00,209.00)(38.00,190.00)(19.00,190.00)(19.00,209.00)
\put(28.50,199.50){\makebox(0,0){$\cdot$}}
\put(28.50,205.83){\makebox(0,0){$\cdot$}}
\put(28.50,212.17){\makebox(0,0){$\cdot$}}

\path(38.00,0.00)(38.00,19.00)(57.00,19.00)(57.00,0.00)
\put(47.50,9.50){\makebox(0,0){$\cdot$}}
\put(47.50,3.17){\makebox(0,0){$\cdot$}}
\put(47.50,-3.17){\makebox(0,0){$\cdot$}}

\path(38.00,38.00)(57.00,38.00)(57.00,19.00)(38.00,19.00)(38.00,38.00)
\put(47.50,24.50){\makebox(0,0)[b]{1}}

\path(38.00,57.00)(57.00,57.00)(57.00,38.00)(38.00,38.00)(38.00,57.00)
\put(47.50,43.50){\makebox(0,0)[b]{1}}

\path(38.00,76.00)(57.00,76.00)(57.00,57.00)(38.00,57.00)(38.00,76.00)
\put(47.50,62.50){\makebox(0,0)[b]{2}}

\path(38.00,95.00)(57.00,95.00)(57.00,76.00)(38.00,76.00)(38.00,95.00)
\put(47.50,81.50){\makebox(0,0)[b]{3}}

\shade
\path(38.00,114.00)(57.00,114.00)(57.00,95.00)(38.00,95.00)(38.00,114.00)
\put(47.50,100.50){\makebox(0,0)[b]{$z'$}}

\path(38.00,133.00)(57.00,133.00)(57.00,114.00)(38.00,114.00)(38.00,133.00)
\put(47.50,119.50){\makebox(0,0)[b]{2}}

\path(38.00,152.00)(57.00,152.00)(57.00,133.00)(38.00,133.00)(38.00,152.00)
\put(47.50,138.50){\makebox(0,0)[b]{1}}

\path(38.00,171.00)(57.00,171.00)(57.00,152.00)(38.00,152.00)(38.00,171.00)
\put(47.50,157.50){\makebox(0,0)[b]{1}}

\path(38.00,190.00)(57.00,190.00)(57.00,171.00)(38.00,171.00)(38.00,190.00)
\put(47.50,176.50){\makebox(0,0)[b]{1}}

\path(57.00,209.00)(57.00,190.00)(38.00,190.00)(38.00,209.00)
\put(47.50,199.50){\makebox(0,0){$\cdot$}}
\put(47.50,205.83){\makebox(0,0){$\cdot$}}
\put(47.50,212.17){\makebox(0,0){$\cdot$}}

\path(57.00,95.00)(76.00,95.00)(76.00,76.00)(57.00,76.00)(57.00,95.00)
\put(66.50,81.50){\makebox(0,0)[b]{3}}

\shade
\path(57.00,114.00)(76.00,114.00)(76.00,95.00)(57.00,95.00)(57.00,114.00)
\put(66.50,100.50){\makebox(0,0)[b]{$z$}}

\path(57.00,133.00)(76.00,133.00)(76.00,114.00)(57.00,114.00)(57.00,133.00)
\put(66.50,119.50){\makebox(0,0)[b]{2}}

\path(76.00,95.00)(95.00,95.00)(95.00,76.00)(76.00,76.00)(76.00,95.00)
\put(85.50,81.50){\makebox(0,0)[b]{1}}

\path(76.00,114.00)(95.00,114.00)(95.00,95.00)(76.00,95.00)(76.00,114.00)
\put(85.50,100.50){\makebox(0,0)[b]{1}}

\path(76.00,133.00)(95.00,133.00)(95.00,114.00)(76.00,114.00)(76.00,133.00)
\put(85.50,119.50){\makebox(0,0)[b]{1}}

\path(95.00,95.00)(114.00,95.00)(114.00,76.00)(95.00,76.00)(95.00,95.00)
\put(104.50,81.50){\makebox(0,0)[b]{1}}

\shade
\path(95.00,114.00)(114.00,114.00)(114.00,95.00)(95.00,95.00)(95.00,114.00)
\put(104.50,100.50){\makebox(0,0)[b]{$z'$}}

\path(95.00,133.00)(114.00,133.00)(114.00,114.00)(95.00,114.00)(95.00,133.00)
\put(104.50,119.50){\makebox(0,0)[b]{1}}

\path(114.00,76.00)(133.00,76.00)(133.00,57.00)(114.00,57.00)(114.00,76.00)
\put(123.50,62.50){\makebox(0,0)[b]{1}}

\path(114.00,95.00)(133.00,95.00)(133.00,76.00)(114.00,76.00)(114.00,95.00)
\put(123.50,81.50){\makebox(0,0)[b]{3}}

\shade
\path(114.00,114.00)(133.00,114.00)(133.00,95.00)(114.00,95.00)(114.00,114.00)
\put(123.50,100.50){\makebox(0,0)[b]{$z$}}

\path(114.00,133.00)(133.00,133.00)(133.00,114.00)(114.00,114.00)(114.00,133.00)
\put(123.50,119.50){\makebox(0,0)[b]{2}}

\path(133.00,76.00)(152.00,76.00)(152.00,57.00)(133.00,57.00)(133.00,76.00)
\put(142.50,62.50){\makebox(0,0)[b]{1}}

\path(133.00,133.00)(152.00,133.00)(152.00,114.00)(133.00,114.00)(133.00,133.00)
\put(142.50,119.50){\makebox(0,0)[b]{4}}

\path(133.00,152.00)(152.00,152.00)(152.00,133.00)(133.00,133.00)(133.00,152.00)
\put(142.50,138.50){\makebox(0,0)[b]{2}}

\path(133.00,171.00)(152.00,171.00)(152.00,152.00)(133.00,152.00)(133.00,171.00)
\put(142.50,157.50){\makebox(0,0)[b]{1}}

\path(133.00,190.00)(152.00,190.00)(152.00,171.00)(133.00,171.00)(133.00,190.00)
\put(142.50,176.50){\makebox(0,0)[b]{1}}

\path(152.00,209.00)(152.00,190.00)(133.00,190.00)(133.00,209.00)
\put(142.50,199.50){\makebox(0,0){$\cdot$}}
\put(142.50,205.83){\makebox(0,0){$\cdot$}}
\put(142.50,212.17){\makebox(0,0){$\cdot$}}

\path(152.00,76.00)(171.00,76.00)(171.00,57.00)(152.00,57.00)(152.00,76.00)
\put(161.50,62.50){\makebox(0,0)[b]{2}}

\path(152.00,95.00)(171.00,95.00)(171.00,76.00)(152.00,76.00)(152.00,95.00)
\put(161.50,81.50){\makebox(0,0)[b]{5}}

\shade
\path(152.00,152.00)(171.00,152.00)(171.00,133.00)(152.00,133.00)(152.00,152.00)
\put(161.50,138.50){\makebox(0,0)[b]{$z'$}}

\shade
\path(152.00,171.00)(171.00,171.00)(171.00,152.00)(152.00,152.00)(152.00,171.00)
\put(161.50,157.50){\makebox(0,0)[b]{$z$}}

\path(152.00,190.00)(171.00,190.00)(171.00,171.00)(152.00,171.00)(152.00,190.00)
\put(161.50,176.50){\makebox(0,0)[b]{1}}

\path(171.00,209.00)(171.00,190.00)(152.00,190.00)(152.00,209.00)
\put(161.50,199.50){\makebox(0,0){$\cdot$}}
\put(161.50,205.83){\makebox(0,0){$\cdot$}}
\put(161.50,212.17){\makebox(0,0){$\cdot$}}

\path(171.00,76.00)(190.00,76.00)(190.00,57.00)(171.00,57.00)(171.00,76.00)
\put(180.50,62.50){\makebox(0,0)[b]{1}}

\path(171.00,133.00)(190.00,133.00)(190.00,114.00)(171.00,114.00)(171.00,133.00)
\put(180.50,119.50){\makebox(0,0)[b]{4}}

\path(171.00,152.00)(190.00,152.00)(190.00,133.00)(171.00,133.00)(171.00,152.00)
\put(180.50,138.50){\makebox(0,0)[b]{2}}

\path(171.00,171.00)(190.00,171.00)(190.00,152.00)(171.00,152.00)(171.00,171.00)
\put(180.50,157.50){\makebox(0,0)[b]{1}}

\path(171.00,190.00)(190.00,190.00)(190.00,171.00)(171.00,171.00)(171.00,190.00)
\put(180.50,176.50){\makebox(0,0)[b]{1}}

\path(190.00,209.00)(190.00,190.00)(171.00,190.00)(171.00,209.00)
\put(180.50,199.50){\makebox(0,0){$\cdot$}}
\put(180.50,205.83){\makebox(0,0){$\cdot$}}
\put(180.50,212.17){\makebox(0,0){$\cdot$}}

\path(190.00,76.00)(209.00,76.00)(209.00,57.00)(190.00,57.00)(190.00,76.00)
\put(199.50,62.50){\makebox(0,0)[b]{1}}

\path(190.00,95.00)(209.00,95.00)(209.00,76.00)(190.00,76.00)(190.00,95.00)
\put(199.50,81.50){\makebox(0,0)[b]{3}}

\shade
\path(190.00,114.00)(209.00,114.00)(209.00,95.00)(190.00,95.00)(190.00,114.00)
\put(199.50,100.50){\makebox(0,0)[b]{$z$}}

\path(190.00,133.00)(209.00,133.00)(209.00,114.00)(190.00,114.00)(190.00,133.00)
\put(199.50,119.50){\makebox(0,0)[b]{2}}

\path(209.00,95.00)(228.00,95.00)(228.00,76.00)(209.00,76.00)(209.00,95.00)
\put(218.50,81.50){\makebox(0,0)[b]{1}}

\shade
\path(209.00,114.00)(228.00,114.00)(228.00,95.00)(209.00,95.00)(209.00,114.00)
\put(218.50,100.50){\makebox(0,0)[b]{$z'$}}

\path(209.00,133.00)(228.00,133.00)(228.00,114.00)(209.00,114.00)(209.00,133.00)
\put(218.50,119.50){\makebox(0,0)[b]{1}}

\path(228.00,95.00)(247.00,95.00)(247.00,76.00)(228.00,76.00)(228.00,95.00)
\put(237.50,81.50){\makebox(0,0)[b]{1}}

\path(228.00,114.00)(247.00,114.00)(247.00,95.00)(228.00,95.00)(228.00,114.00)
\put(237.50,100.50){\makebox(0,0)[b]{1}}

\path(228.00,133.00)(247.00,133.00)(247.00,114.00)(228.00,114.00)(228.00,133.00)
\put(237.50,119.50){\makebox(0,0)[b]{1}}

\path(247.00,95.00)(266.00,95.00)(266.00,76.00)(247.00,76.00)(247.00,95.00)
\put(256.50,81.50){\makebox(0,0)[b]{1}}

\shade
\path(247.00,114.00)(266.00,114.00)(266.00,95.00)(247.00,95.00)(247.00,114.00)
\put(256.50,100.50){\makebox(0,0)[b]{$z$}}

\path(247.00,133.00)(266.00,133.00)(266.00,114.00)(247.00,114.00)(247.00,133.00)
\put(256.50,119.50){\makebox(0,0)[b]{1}}

\path(266.00,0.00)(266.00,19.00)(285.00,19.00)(285.00,0.00)
\put(275.50,9.50){\makebox(0,0){$\cdot$}}
\put(275.50,3.17){\makebox(0,0){$\cdot$}}
\put(275.50,-3.17){\makebox(0,0){$\cdot$}}

\path(266.00,38.00)(285.00,38.00)(285.00,19.00)(266.00,19.00)(266.00,38.00)
\put(275.50,24.50){\makebox(0,0)[b]{1}}

\path(266.00,57.00)(285.00,57.00)(285.00,38.00)(266.00,38.00)(266.00,57.00)
\put(275.50,43.50){\makebox(0,0)[b]{1}}

\path(266.00,76.00)(285.00,76.00)(285.00,57.00)(266.00,57.00)(266.00,76.00)
\put(275.50,62.50){\makebox(0,0)[b]{2}}

\path(266.00,95.00)(285.00,95.00)(285.00,76.00)(266.00,76.00)(266.00,95.00)
\put(275.50,81.50){\makebox(0,0)[b]{3}}

\shade
\path(266.00,114.00)(285.00,114.00)(285.00,95.00)(266.00,95.00)(266.00,114.00)
\put(275.50,100.50){\makebox(0,0)[b]{$z'$}}

\path(266.00,133.00)(285.00,133.00)(285.00,114.00)(266.00,114.00)(266.00,133.00)
\put(275.50,119.50){\makebox(0,0)[b]{2}}

\path(266.00,152.00)(285.00,152.00)(285.00,133.00)(266.00,133.00)(266.00,152.00)
\put(275.50,138.50){\makebox(0,0)[b]{1}}

\path(266.00,171.00)(285.00,171.00)(285.00,152.00)(266.00,152.00)(266.00,171.00)
\put(275.50,157.50){\makebox(0,0)[b]{1}}

\path(266.00,190.00)(285.00,190.00)(285.00,171.00)(266.00,171.00)(266.00,190.00)
\put(275.50,176.50){\makebox(0,0)[b]{1}}

\path(285.00,209.00)(285.00,190.00)(266.00,190.00)(266.00,209.00)
\put(275.50,199.50){\makebox(0,0){$\cdot$}}
\put(275.50,205.83){\makebox(0,0){$\cdot$}}
\put(275.50,212.17){\makebox(0,0){$\cdot$}}

\path(285.00,0.00)(285.00,19.00)(304.00,19.00)(304.00,0.00)
\put(294.50,9.50){\makebox(0,0){$\cdot$}}
\put(294.50,3.17){\makebox(0,0){$\cdot$}}
\put(294.50,-3.17){\makebox(0,0){$\cdot$}}

\path(285.00,38.00)(304.00,38.00)(304.00,19.00)(285.00,19.00)(285.00,38.00)
\put(294.50,24.50){\makebox(0,0)[b]{1}}

\shade
\path(285.00,57.00)(304.00,57.00)(304.00,38.00)(285.00,38.00)(285.00,57.00)
\put(294.50,43.50){\makebox(0,0)[b]{$s$}}

\shade
\path(285.00,76.00)(304.00,76.00)(304.00,57.00)(285.00,57.00)(285.00,76.00)
\put(294.50,62.50){\makebox(0,0)[b]{$s'$}}

\shade
\path(285.00,114.00)(304.00,114.00)(304.00,95.00)(285.00,95.00)(285.00,114.00)
\put(294.50,100.50){\makebox(0,0)[b]{$s$}}

\shade
\path(285.00,133.00)(304.00,133.00)(304.00,114.00)(285.00,114.00)(285.00,133.00)
\put(294.50,119.50){\makebox(0,0)[b]{$s'$}}

\path(285.00,152.00)(304.00,152.00)(304.00,133.00)(285.00,133.00)(285.00,152.00)
\put(294.50,138.50){\makebox(0,0)[b]{1}}

\shade
\path(285.00,171.00)(304.00,171.00)(304.00,152.00)(285.00,152.00)(285.00,171.00)
\put(294.50,157.50){\makebox(0,0)[b]{$s$}}

\shade
\path(285.00,190.00)(304.00,190.00)(304.00,171.00)(285.00,171.00)(285.00,190.00)
\put(294.50,176.50){\makebox(0,0)[b]{$s'$}}

\path(304.00,209.00)(304.00,190.00)(285.00,190.00)(285.00,209.00)
\put(294.50,199.50){\makebox(0,0){$\cdot$}}
\put(294.50,205.83){\makebox(0,0){$\cdot$}}
\put(294.50,212.17){\makebox(0,0){$\cdot$}}

\path(304.00,0.00)(304.00,19.00)(323.00,19.00)(323.00,0.00)
\put(313.50,9.50){\makebox(0,0){$\cdot$}}
\put(313.50,3.17){\makebox(0,0){$\cdot$}}
\put(313.50,-3.17){\makebox(0,0){$\cdot$}}

\path(304.00,38.00)(323.00,38.00)(323.00,19.00)(304.00,19.00)(304.00,38.00)
\put(313.50,24.50){\makebox(0,0)[b]{1}}

\path(304.00,57.00)(323.00,57.00)(323.00,38.00)(304.00,38.00)(304.00,57.00)
\put(313.50,43.50){\makebox(0,0)[b]{1}}

\path(304.00,76.00)(323.00,76.00)(323.00,57.00)(304.00,57.00)(304.00,76.00)
\put(313.50,62.50){\makebox(0,0)[b]{2}}

\path(304.00,95.00)(323.00,95.00)(323.00,76.00)(304.00,76.00)(304.00,95.00)
\put(313.50,81.50){\makebox(0,0)[b]{2}}

\path(304.00,114.00)(323.00,114.00)(323.00,95.00)(304.00,95.00)(304.00,114.00)
\put(313.50,100.50){\makebox(0,0)[b]{2}}

\path(304.00,133.00)(323.00,133.00)(323.00,114.00)(304.00,114.00)(304.00,133.00)
\put(313.50,119.50){\makebox(0,0)[b]{1}}

\path(304.00,152.00)(323.00,152.00)(323.00,133.00)(304.00,133.00)(304.00,152.00)
\put(313.50,138.50){\makebox(0,0)[b]{1}}

\path(304.00,171.00)(323.00,171.00)(323.00,152.00)(304.00,152.00)(304.00,171.00)
\put(313.50,157.50){\makebox(0,0)[b]{1}}

\path(304.00,190.00)(323.00,190.00)(323.00,171.00)(304.00,171.00)(304.00,190.00)
\put(313.50,176.50){\makebox(0,0)[b]{1}}

\path(323.00,209.00)(323.00,190.00)(304.00,190.00)(304.00,209.00)
\put(313.50,199.50){\makebox(0,0){$\cdot$}}
\put(313.50,205.83){\makebox(0,0){$\cdot$}}
\put(313.50,212.17){\makebox(0,0){$\cdot$}}

\texture{cccccccc 0 0 0 cccccccc 0 0 0
         cccccccc 0 0 0 cccccccc 0 0 0
         cccccccc 0 0 0 cccccccc 0 0 0
         cccccccc 0 0 0 cccccccc 0 0 0}
\whiten

\path(19.00,95.00)(38.00,95.00)(38.00,76.00)(19.00,76.00)(19.00,95.00)
\put(28.50,85.50){\makebox(0,0){\circle*{9.50}}}

\path(133.00,95.00)(152.00,95.00)(152.00,76.00)(133.00,76.00)(133.00,95.00)
\put(142.50,85.50){\makebox(0,0){\circle*{9.50}}}

\path(133.00,114.00)(152.00,114.00)(152.00,95.00)(133.00,95.00)(133.00,114.00)
\put(142.50,104.50){\makebox(0,0){\circle*{9.50}}}

\path(152.00,114.00)(171.00,114.00)(171.00,95.00)(152.00,95.00)(152.00,114.00)
\put(161.50,104.50){\makebox(0,0){\circle*{9.50}}}

\path(152.00,133.00)(171.00,133.00)(171.00,114.00)(152.00,114.00)(152.00,133.00)
\put(161.50,123.50){\makebox(0,0){\circle*{9.50}}}

\path(171.00,95.00)(190.00,95.00)(190.00,76.00)(171.00,76.00)(171.00,95.00)
\put(180.50,85.50){\makebox(0,0){\circle*{9.50}}}

\path(171.00,114.00)(190.00,114.00)(190.00,95.00)(171.00,95.00)(171.00,114.00)
\put(180.50,104.50){\makebox(0,0){\circle*{9.50}}}

\path(285.00,95.00)(304.00,95.00)(304.00,76.00)(285.00,76.00)(285.00,95.00)
\put(294.50,85.50){\makebox(0,0){\circle*{9.50}}}

\put(66.5,201){\makebox(0,0)[b]{$r$}}
\put(66.5,199){\vector(0,-1){30}}
\put(66.5,40){\makebox(0,0)[b]{$r$}}
\put(66.5,38){\vector(0,-1){30}}
\put(256.5,201){\makebox(0,0)[b]{$s$}}
\put(256.5,199){\vector(0,-1){30}}
\put(256.5,40){\makebox(0,0)[b]{$s$}}
\put(256.5,38){\vector(0,-1){30}}
\put(94,142.5){\makebox(0,0)[r]{$z$}}
\qbezier(95,142.5)(123.5,142.5)(123.5,171)
\put(123.5,171){\vector(0,1){2}}
\put(229,142.5){\makebox(0,0)[l]{$z$}}
\qbezier(228,142.5)(199.5,142.5)(199.5,171)
\put(199.5,171){\vector(0,1){2}}
\mathversion{normal}

\end{picture}
\end{center} \vspace*{-10pt}
\caption{$z$ can be revealed when either $r$ or $s$ is known} \label{pswires}
\end{figure}

Additionally, we make variables $z_{jj}$, where $1 \le j \le n+1 + m$, 
which can be revealed when either $s_2$ or $s_3$ is known. This can be 
done with a reversed splitter, as in figure \ref{pswires}. Furthermore,
for each $j \notin E$, we replace the starting point of $x_j$ by 
circuitry for $x_j = z_{1j} + z_{2j} + \cdots + z_{jj} \bmod 2$, thus
$x_j$ is no longer a real variable when $j \notin E$.

To show that this construction works, notice first that there is no
way to get information about the variables $x_i$ for which $i \in E$.
Thus we have to guess them. Besides that, it suffices to guess one of $s_2$ 
and $s_3$. When we guess either $s_2$ or $s_3$, we guess them to be equal 
to one, since that value is the most likely.

The crucial question is which of $s_2$ and $s_3$ will be subjected to
a guess. At first glance, one might think that that should be the one which 
is most likely to equal to one, but that is not the right criterion. No, 
the criterion for selecting $s_2$ just appears to be that (\ref{mn2e})
is true or that rounding (\ref{m2}) yields $1$, depending on the value of
$m$.

In fact, when $s_2 = 1$ is guessed, the probability of winning will seem to be
\begin{equation} \label{s2guess}
2^{-e}\left(\frac12 + \frac12 \max\Big\{\vartheta\,\Big|\,
R_1 x_1 R_2 x_2 \cdots R_n x_n \Pr[f(x_1,x_2,\ldots,x_n)]=\vartheta\Big\}\right)
\end{equation}
The probability of winning is equal to $2^{-e}(1 - 2^{-m})$ when $s_3 = 1$ is 
guessed.

Notice that the first guess to be done is either $s_2 = 1$ or $s_3 = 1$, since
other guesses do not help more for choosing between $s_2$ and $s_3$ than making 
a corresponding assumption without checking. If $s_3 = 1$ is guessed with success, 
then the probability of revealing all mines is $2^{-e}$ at this stage, and that is it. 
But if $s_2 = 1$ is guessed with success, then the probability of revealing all 
mines is dependent of subsequent guesses, and these guesses are to be chosen to 
maximize the probability of winning.

After guessing $s_2 = 1$ with success, we do the following for 
$j = 1, 2, \ldots, \allowbreak n+1+m$, in that order.
\begin{itemize}

\item If $j \in E$ and $m = n + 2 - e$, then we guess $x_j$ in order to satisfy
$$
A_{j+1} x_{j+1} A_{j+2} x_{j+2} \cdots A_{n+1+m} x_{n+1+m} f(x_1,x_2,\ldots,x_n)
$$
when $m = n+2-e$. If this is not possible, then $s_3 = 1$ should have been guessed
instead of $s_2 = 1$, since in that case, (\ref{s2guess}) can be estimated by
$2^{-e}(1-2^{-(n+1)+e})$, which is less than the probability $2^{-e}(1-2^{-(n+2)+e})$ 
of winning when $s_3 = 1$ is guessed.

\item If $j \in E$ and $m = 2$, then we guess $x_j$ in order to maximize $\vartheta$ 
such that
$$
R_{j+1} x_{j+1} R_{j+2} x_{j+2} \cdots R_{n+1+m} x_{n+1+m} 
\Pr[f(x_1,x_2,\ldots,x_n)] = \vartheta
$$

\item If $j \not\in E$, then we can reveal $x_j$ since we know 
$x_1, x_2, \ldots, x_{j-1}$ and $s_2$.

\end{itemize}
In a similar manner as with choosing between $s_2$ and $s_3$,
doing another guess before guessing $x_j$ with $j \in E$ (except 
guessing $s_2$ and $x_i$ for all $i \in E$ with $i < j$) is not better for
making a good guess for $x_j$ than a corresponding assumption without 
checking. Hence the answer to the crucial question which of $s_2$ and 
$s_3$ should be subjected to a guess is as claimed.
\end{proof}

\end{document}